\newcommand{\CLASSINPUTtoptextmargin}{1.75cm}
\newcommand{\CLASSINPUTbottomtextmargin}{1.75cm}
\newcommand{\CLASSINPUToutersidemargin}{1.7cm}

\documentclass[12pt,draftclsnofoot,onecolumn]{IEEEtran}

\usepackage{amsmath, mathrsfs, mathtools}
\usepackage{amsfonts}
\usepackage{amssymb}
\usepackage{graphicx}
\usepackage{color}
\usepackage{multirow}
\usepackage{graphicx}
\usepackage{stmaryrd}

\usepackage{subfig}

\allowdisplaybreaks

\usepackage{grffile}
\usepackage{amsmath}

\let\labelindent\relax
\usepackage{enumitem}

\newcommand{\subparagraph}{}
\usepackage[compact]{titlesec}
\titlespacing*{\section}{15pt}{1.2\baselineskip}{0.9\baselineskip}

\allowdisplaybreaks

\newcommand{\myhash}{%
  {\settoheight{\dimen0}{C}\kern-.05em\, \resizebox{!}{\dimen0}{\raisebox{\depth}{\#}}}}

\usepackage{caption}

\usepackage[numbers,sort&compress]{natbib}

\usepackage{algorithm}
\usepackage{algpseudocode}
\usepackage{pifont}
\usepackage{varwidth}

\def\betam{{\boldsymbol{\beta}}}
\def\gammam{{\boldsymbol{\gamma}}}

\newcommand{\vc}[1]{\mathbf{#1}}
\newcommand{\mt}[1]{\mathbf{#1}}

\newcommand{\SigmaEmp}{\widehat{\Sigmam}_\bfy} %

\newcommand{\Amatrix}{\bA}  %

\newcommand{\Acol}{\mt{A}} %
\newcommand{\avec}{\bfa} %

\newcommand{\bxtrue}{\gammam^\circ} %
\newcommand{\xtrue}{\gamma^\circ} %

\newcommand{\xsome}{\gammam} %

\newcommand{\bxnnls}{\gammam^*} %

\newcommand{\dimPilots}{L} %

\newcommand{\dimParam}{{K_\text{tot}}} %

\newcommand{\ellp}[1]{{#1}} %

\usepackage{multicol}


\usepackage{mathbbol}

\def\mindex#1{\index{#1}}

\def\sq{\hbox{\rlap{$\sqcap$}$\sqcup$}}
\def\qed{\ifmmode\sq\else{\unskip\nobreak\hfil
\penalty50\hskip1em\null\nobreak\hfil\sq
\parfillskip=0pt\finalhyphendemerits=0\endgraf}\fi\medskip}

\long\def\defbox#1{\framebox[.9\hsize][c]{\parbox{.85\hsize}{%
\parindent=0pt
\baselineskip=12pt plus .1pt      %
\parskip=6pt plus 1.5pt minus 1pt %
 #1}}}

\long\def\beginbox#1\endbox{\subsection*{}%
\hbox{\hspace{.05\hsize}\defbox{\medskip#1\bigskip}}%
\subsection*{}}

\def\endbox{}

\def\diag{{\text{diag}}}

\def\tr{\mathsf{tr}}

\def\supp{{\rm supp\,}}

\newsavebox{\junk}
\savebox{\junk}[1.6mm]{\hbox{$|\!|\!|$}}

\def\argmin{\mathop{\rm arg\, min}}

\def\A{{\sf A}}

\newcommand{\field}[1]{\mathbb{#1}}

\def\Re{\field{R}}

\def\A{\field{A}}

\def\bA{{\mathbb A}}

\def\bC{{\mathbb C}}

\def\bE{{\mathbb E}}

\def\bR{{\mathbb R}}

\def\bfA{{\bf A}}
\def\bfB{{\bf B}}

\def\bfG{{\bf G}}
\def\bfH{{\bf H}}
\def\bfI{{\bf I}}

\def\bfU{{\bf U}}

\def\bfX{{\bf X}}
\def\bfY{{\bf Y}}
\def\bfZ{{\bf Z}}

\def\bfa{{\bf a}}

\def\bfd{{\bf d}}
\def\bfe{{\bf e}}

\def\bfh{{\bf h}}

\def\bfx{{\bf x}}
\def\bfy{{\bf y}}
\def\bfz{{\bf z}}

\def\sfF{{\sf F}}

\def\sfH{{\sf H}}

\def\bfmath#1{{\mathchoice{\mbox{\boldmath$#1$}}%
{\mbox{\boldmath$#1$}}%
{\mbox{\boldmath$\scriptstyle#1$}}%
{\mbox{\boldmath$\scriptscriptstyle#1$}}}}

\def\bfmY{\bfmath{Y}}

\def\bfmhhaY{\bfmath{\hhaY}} %
\def\bfmhhaY{\hbox to 0pt{$\widehat{\bfmY}$\hss}\widehat{\phantom{\raise 1.25pt\hbox{$\bfmY$}}}}

\def\til={{\widetilde =}}

\def\clA{{\cal A}}

\def\clC{{\cal C}}

\def\clK{{\cal K}}

\def\clN{{\cal N}}

\def\clQ{{\cal Q}}

\def\clS{{\cal S}}

 \def\FRAC#1#2#3{\genfrac{}{}{}{#1}{#2}{#3}}

\def\ddtp{{\mathchoice{\FRAC{1}{d^{\hbox to 2pt{\rm\tiny +\hss}}}{dt}}%
{\FRAC{1}{d^{\hbox to 2pt{\rm\tiny +\hss}}}{dt}}%
{\FRAC{3}{d^{\hbox to 2pt{\rm\tiny +\hss}}}{dt}}%
{\FRAC{3}{d^{\hbox to 2pt{\rm\tiny +\hss}}}{dt}}}}

\def\average#1,#2,{{1\over #2} \sum_{#1}^{#2}}

\def\eye(#1){{\bf(#1)}\quad}

\newtheorem{theorem}{{\bf Theorem}}

\newtheorem{corollary}{Corollary}
\newtheorem{definition}{{\bf Definition}}
\newtheorem{remark}{{\bf Remark}}

\def\eq#1/{(\ref{e:#1})}

\newcommand\figref{Figure~\ref}

\newcommand{\beqn}[1]{\notes{#1}%
\begin{eqnarray} \elabel{#1}}

\newcommand{\eeqn}{\end{eqnarray} }

\newcommand{\beq}{\begin{equation}}

\newcommand{\eeq}{\end{equation}}

\def\bdes{\begin{description}}
\def\edes{\end{description}}

\newcounter{rmnum}

\newcounter{anum}

{\end{list}}

\def\ass(#1:#2){(#1\ref{#1:#2})}

\def\ritem#1{
\item[{\sf \ass(\current_model:#1)}]
}

\newenvironment{recall-ass}[1]{%
\begin{description}
\def\current_model{#1}}{
\end{description}
}

\usepackage{float}
\usepackage{afterpage}
\usepackage{balance}

\def\herm{{\sfH}}

\def\snr{{\mathsf{snr}}}

\def\cg{{\clC\clN}} 
\newcommand{\normd}[1]{{\left\vert\kern-0.25ex\left\vert\kern-0.25ex\left\vert #1 
    \right\vert\kern-0.25ex\right\vert\kern-0.25ex\right\vert}}

\long\def\comment#1{}

\newcommand{\CC}{\mathbb{C}}
\newcommand{\PP}{\mathbb{P}}
\newcommand{\RR}{\mathbb{R}}

\newcommand{\EE}{\mathbb{E}}

\newcommand{\av}{{\bf a}}
\newcommand{\bv}{{\bf b}}
\newcommand{\cv}{{\bf c}}
\newcommand{\dv}{{\bf d}}

\newcommand{\gv}{{\bf g}}
\newcommand{\hv}{{\bf h}}

\newcommand{\rv}{{\bf r}}

\newcommand{\tv}{{\bf t}}
\newcommand{\uv}{{\bf u}}
\newcommand{\wv}{{\bf w}}
\newcommand{\vv}{{\bf v}}
\newcommand{\xv}{{\bf x}}
\newcommand{\yv}{{\bf y}}
\newcommand{\zv}{{\bf z}}

\newcommand{\Am}{{\bf A}}
\newcommand{\Bm}{{\bf B}}

\newcommand{\Gm}{{\bf G}}
\newcommand{\Hm}{{\bf H}}
\newcommand{\Id}{{\bf I}}

\newcommand{\Qm}{{\bf Q}}
\newcommand{\Rm}{{\bf R}}

\newcommand{\Um}{{\bf U}}

\newcommand{\Xm}{{\bf X}}
\newcommand{\Ym}{{\bf Y}}
\newcommand{\Zm}{{\bf Z}}

\newcommand{\Cc}{{\cal C}}

\newcommand{\Hc}{{\cal H}}

\newcommand{\Kc}{{\cal K}}
\newcommand{\Lc}{{\cal L}}

\newcommand{\Nc}{{\cal N}}

\newcommand{\Sc}{{\cal S}}

\newcommand{\Xc}{{\cal X}}

\newcommand{\gammav}{\hbox{\boldmath$\gamma$}}

\newcommand{\rhov}{\hbox{\boldmath$\rho$}}

\newcommand{\Gammam}{\boldsymbol{\Gamma}}

\newcommand{\Deltam}{\boldsymbol{\Delta}}
\newcommand{\Sigmam}{\boldsymbol{\Sigma}}
\newcommand{\Phim}{\boldsymbol{\Phi}}

\newcommand{\Xim}{\boldsymbol{\Xi}}

\newcommand{\epsilonf}{\boldsymbol{\mathfrak{e}}}

\newcommand{\trace}{{\hbox{tr}}}

\renewcommand{\Re}{{\rm Re}}
\renewcommand{\Im}{{\rm Im}}
\newcommand{\trasp}{{\sf T}}
\newcommand{\transp}{{\sf T}}
\renewcommand{\vec}{{\rm vec}}

\title{Non-Bayesian Activity Detection, Large-Scale Fading Coefficient Estimation, and Unsourced Random Access
with a Massive MIMO Receiver}
\author{Alexander Fengler, Saeid Haghighatshoar,  Peter Jung, Giuseppe Caire
\thanks{The authors are with the Communications and Information Theory Group, Technische Universit\"{a}t Berlin (\{fengler, saeid.haghighatshoar, peter.jung, caire\}@tu-berlin.de).}
\thanks{Parts of this paper were presented in the 2018 IEEE International Symposium on Information Theory (ISIT) \cite{AD:isit2018} and at the IEEE Asilomar Conference on Signals, Systems, and Computers 2019 \cite{Fen2019e}.}
}

\begin{document}

\maketitle

\vspace{-1cm}

\begin{abstract}
In this paper, we study the problem of user \textit{activity detection} and
large-scale fading coefficient estimation in a random access wireless uplink with a massive MIMO base station
with a large number $M$ of antennas and a large number of wireless single-antenna devices (users).
We consider a block fading channel model where the $M$-dimensional channel vector of each user remains 
constant over a \textit{coherence block} containing $\dimPilots$ signal dimensions in time-frequency. 
In the considered setting, the number of potential users $K_\text{tot}$ 
is much larger than $\dimPilots$ but at each time slot only $K_a \ll K_\text{tot}$ 
of them are active. 
Previous results, based on compressed sensing, require that  $K_a\le \dimPilots$,
which is a bottleneck in massive deployment scenarios.
In this work, we show that such limitation can be overcome when the number of base station antennas $M$ 
is sufficiently large. More specifically, we prove that with a coherence block of dimension $\dimPilots$
and a number of antennas  $M$ such that $K_a/M = o(1)$,
one can identify $K_a = O(\dimPilots^2/\log^2(\frac{K_\text{tot}}{K_a}))$ active users, 
which is much larger than the previously known bounds. 
We also provide two algorithms. One is based on Non-Negative Least-Squares, for which
the above scaling result can be rigorously proved.
The other consists of a low-complexity iterative componentwise minimization of
the likelihood function of the underlying problem. While for this algorithm a rigorous proof cannot be given, we analyze 
a constrained version of the Maximum Likelihood (ML) problem (a combinatorial optimization with exponential complexity)
and find the same fundamental scaling law for the number of identifiable users.
Therefore, we conjecture that the low-complexity (approximated) ML algorithm also achieves
the same scaling law and we demonstrate its performance by simulation.
We also compare the discussed methods with the (Bayesian) MMV-AMP algorithm, 
recently proposed for the same setting, and show superior performance and better numerical stability.
Finally, we use the discussed approximated ML algorithm as the inner decoder
in a concatenated coding scheme for {\em unsourced random access},
a grant-free uncoordinated multiple access scheme where all users make use of the same codebook,
and the receiver must produce the list of transmitted messages, irrespectively of the identity of the transmitters.
We show that reliable communication is possible at any $E_b/N_0$ provided that a sufficiently large number of base station antennas is used,
and that a sum spectral efficiency in the order of $\mathcal{O}(L\log(L))$ is achievable.
\end{abstract}


\section{Introduction}  \label{intro}

One of the paradigms of modern machine-type communications \cite{Tal2012a}
consists of a very large number of devices (here referred to as ``users'') with
sporadic data.  Typical examples thereof are Internet-of-Things (IoT)
applications, wireless sensors deployed to monitor smart infrastructure, and
wearable biomedical devices \cite{Has2013}.  In such scenarios, a Base Station
(BS) should be able to collect data from a large number of devices.  However,
due to the sporadic nature of the data generation and communication, allocating
some dedicated transmission resource to all users in the system may be extremely wasteful.  
In most wireless systems, a dedicated random access slot (or
logical channel) is used to allow the users with some data to transmit to
ask the Base Station (BS) to be granted access to some transmission
resource, which is successively released. For example, most systems
operating today, such as (3G, 4G-LTE, and 5G New Radio, follow this
paradigm \cite{Ses2009,Agi2016a}.  As an alternative, the random access
channel itself can be used to directly transmit data in a {\em grant-free}
mode.  As yet another twist in the system classification of random access
schemes, a recently proposed information theoretic model referred to as
{\em unsourced random access} assumes grant-free operations and, in
addition, that all users make use of exactly the same codebook
\cite{Pol2017}. Unsourced random access is motivated by an IoT scenario
where millions of cheap devices have their codebook hardwired at the moment
of production, and are then disseminated into the environment.  In this case,
the BS receiver must determine the list of transmitted messages irrespectively
of the identity of the active users.
\footnote{If a
    user wishes to communicate its ID, it can send it as part of the payload.
    Therefore, in the paradigm of unsourced random access, if the users make
    use of individually different codebooks, it would be impossible for the BS
    to know in advance which codebook to decode since the identity of the active
    users is not known a priori. Hence, in this context it is in fact essential,
    and not just a matter of implementation costs, that all users utilize the same
codebook.}

In this paper, we are mainly interested in the problem of {\em Activity
Detection} (AD) from a dedicated pilot slot.  The AD function can be included either in a more traditional granted resource random access protocol, or in a grant-free protocol. In the second part of the paper,
we shall use the proposed AD scheme as the inner code/decoder 
of a concatenated coding scheme specifically addressing the problem of unsourced random access.   

AD is a fundamental challenge in massive sensor deployments and random access
scenarios to be expected for IoT (see, e.g.,
\cite{Bockelmann:ETT2013,boljanovic2017user,Sen2017,Liu2018c,liu2017massive,Liu2018b}
for some recent works) We consider a classical block-fading wireless
communication channel between the users and the BS \cite{tse2005fundamentals},
where the channel coefficients remain constant over \textit{coherence blocks}
consisting of $\dimPilots$ signal dimensions in the time-frequency domain, and
change randomly from block to block according to a stationary ergodic process
\cite{tse2005fundamentals}.  A fundamental limitation when considering a
single-antenna BS is that the required signal dimension $\dimPilots$ to
identify reliably a subset of $K_a$ active users among a set consisting of
$K_\text{tot}$ \textit{potentially active}  users scales as $\dimPilots=O(K_a
\log(\frac{K_\text{tot}}{K_a}))$, thus, almost linearly with $K_a$.  To keep up
with the scaling requirements in  practical applications where $K_a$ may be of
the order of $10^2$ and $K_\text{tot}$ may be of the order of $10^4$ -- $10^5$,
it is crucial to overcome this limitation in an efficient way that does not
require devoting too many pilot dimension to AD.  

In a series of recent works \cite{liu2017massive,Liu2018b,Che2018}, AD  with a
massive MIMO BS with a large number $M$ of antennas was considered and
formulated as a Multiple Measurement Vector (MMV)
\cite{Che2006,Cot2005a,Kim2012} problem.  In these works, the activity
detection problem is formulated in a Bayesian way and a method based on an MMV
suited version of {\em Approximated Message Passing} (MMV-AMP) followed by a
componentwise Neyman-Pearson activity estimation by suitable thresholding is
proposed. There are several issues with this problem formulation and with the
proposed MMV-AMP algorithm.  First, the algorithm needs to treat the
Large-Scale Fading Coefficients (LSFCs)\footnote{We refer to LSFC as the
averaged received power from each user when active, up to a suitable common
scaling factor. Users have different LSFCs because of different distances from
the BS and large-scale effects such as log-normal shadowing.} as either as
deterministic known quantities, or as random quantities whose prior
distribution is known. In practice, it is not easy to individually measure the
LSFC from all $K_{\rm tot}$ users, especially when they stay silent for a long
time and move or the propagation conditions change.  Also, the typical distance
dependent pathloss and log-normal shadowing laws used in standard models are
not quite representative of specific environments and the prior ensemble
distribution would assume some spatial distribution (e.g., uniform in a cell as
in \cite{liu2017massive,Che2018}) which is not always the case.  Furthermore,
the MMV-AMP algorithm can be analyzed via the state evolution method
\cite{liu2017massive,Che2018} in the large-dimensional regime where $L, K_a$,
and $K_{\rm tot}$ grow to infinity at fixed ratios $\frac{K_a}{\dimPilots} \to
\alpha$ and $\frac{K_\text{tot}}{\dimPilots}\to \beta$ with $\alpha, \beta \in
(0, \infty)$ while $M$ is finite.  Therefore, the regime of $L$ linear in $K_a$
(which we wish to beat) is somehow unavoidable in this type of analysis.
Finally, it turns out that in practical scenarios where $M$ is fairly larger
than $L$ and comparable to $K_a$ (which are scenarios of interest in our work
and in practical scenarios, where $L$ is between 50 and 200 and $M$ can be up
to 256 antennas \cite{shepard2012argos,Mal2017,Cho2016}), MMV-AMP is quite
numerically unstable and gives pathological and unpredictable behaviors that
one would like definitely to avoid in a real-world implementation. 

In this work, we consider a non-Bayesian approach, treating the LSFCs as
deterministic unknown. We use tools from Compressed Sensing (CS) to
provide a stability analysis of the LSFC estimation and AD problem for finite
SNR and finite number of antennas $M$. As a consequence of this analysis, we
are able to show that with a coherence block of dimension $\dimPilots$, and
with a sufficient number of BS antennas $M$ with $K_a/M=o(1)$, one can
estimate the LSFC, and thus identify the activity, of up to
$K_a=O(\dimPilots^2/\log^2(\frac{K_\text{tot}}{K_a}))$ active users among
$K_\text{tot}$ users.
These results are obtained by analyzing a Non-Negative Least-Squares (NNLS) algorithm
applied to the sample covariance information, which was recently considered for LSFC estimation
in \cite{docomo}. The analysis in \cite{docomo} showed, that with a random choice of pilot sequences
the LSFCs of up to 
$K_a = \mathcal{O}(L^2)$ users could be estimated, but the proof was limited by the assumptions of
$K_a = K_\text{tot}$,$K_\text{tot} \leq L^2$ and $M\to\infty$.
Our result lifts all these restrictions and shows that $K_\text{tot}$ may be potentially much
larger than $\dimPilots^2$ and $K_a$, where one needs to pay only
a logarithmic penalty $O(\log^2(\frac{K_\text{tot}}{K_a}))$ for increasing the
total number of users $K_\text{tot}$.
This makes the proposed scheme very
attractive for IoT setups, in which the number of active users $K_a$ as well
as the total number of users $K_\text{tot}$ may be extremely large. 

Furthermore, we propose to use an improved algorithm for AD based on the
Maximum-Likelihood (ML) estimation of the LSFCs of the active users. The
resulting likelihood maximization is a non-convex problem, that can be
solved (approximately) by iterative componentwise minimization.  This yields an
iterative scheme based on rank-1 updates whose complexity is comparable to that
of NNLS and MMV-AMP. Extensive numerical simulations show that the ML
algorithm is superior to NNLS and to MMV-AMP in any regime, and does not suffer
from the ill-conditioned behavior of MMV-AMP for the case of large $M$.  
The componentwise optimization of the
log-likelihood function was developed in \cite{Tip2001,Fau2002,Tip2003,Pal2004},
where the sparse Bayesian learning (SBL) framework was introduced to find the
optimal vector of weights in a linear regression problem. In the SBL framework
it is assumed that the weight vector follows a Gaussian prior distribution with
zero mean and a diagonal covariance matrix. The entries of the covariance matrix
are estimated by maximizing the likelihood of the data. It was observed that
the likelihood function maximization yields a very sparse result, which is a desirable property
in statistical learning. The maximum of the
likelihood function can be computed iteratively by following the general
expectation-maximization (EM) framework \cite{Dem1977,Moo1996}, but it was
found that a componentwise optimization leads to faster convergence while still
being guaranteed to converge to at least a local maximum of the likelihood
function \cite{Fau2002}, similar to EM.  The SBL framework was extended and
applied to basis selection \cite{Wip2004}, compressed sensing \cite{Ji2008} and
also the MMV problem \cite{Wip2007}, where the latter was termed M-SBL. Here,
the task is to recover $\Xm\in\CC^{K_\text{tot}\times M}$
(here we stick to the notation introduced in this paper)
from multiple measurements $\Ym \in \CC^{L\times M}$ of the form
\beq
\Ym = \Am\Xm + \Zm
\eeq
Following the SBL framework, it is assumed that the rows
of $\Xm$ are distributed according to
\beq
\Xm_{:,i} \sim \mathcal{CN}(0,\gamma_i\Id_M).
\eeq
Note, that in the literature the term M-SBL
has often been used ambiguously to refer to the ML estimate of the parameters
$\gamma_i$ as well as to the algorithm used to find this solution, which may be
either coordinate-wise optimization or EM. Both of these algorithms lead to
similar solutions \cite{Wip2007}, but here we adopt the componentwise optimization algorithm,
since it can be efficiently implemented using rank-1
updates leading to a significant complexity reduction compared to the EM
version. 

Let $\Xm$ have $K_a$ non-zero rows.
The identifiability limits of the ML solution of M-SBL were analysed in
\cite{Tan2010a,Pal2014,Pal2015,Koo2018,Bal2014}.  While the early work
\cite{Tan2010a} was restricted to the case $K_a \leq L$, \cite{Pal2014} made
the distinction between recovering $\Xm$, which necessarily requires $K_a \leq
L$, and recovering the vector $\gammam = (\gamma_1,...,\gamma_{K_\text{tot}})$.
It was noticed in \cite{Pal2014} that the recovery of $\gammam$ is governed by
the properties of the Khatri-Rao product $\Am\odot\Am$ and it was proven that,
with a random choice of $\Am$, up to $K_a = \mathcal{O}(L^2)$ non-zero entries
of $\gammam$ can be recovered uniquely if the covariance matrix of $\Ym$ is
known exactly.  The proof of \cite{Pal2014} (similar proofs were given independently in
\cite{docomo,Bal2014}) is based on the fact that any $2K_a$ columns of
$\Am\odot\Am$ are linearly independent almost surely for $K_a$ up to
$\mathcal{O}(L^2)$.
This proof deals only with the identifiability, i.e., it does not apply to a
specific recovery algorithm and does not take into account the uncertainty
in estimating the covariance matrix of $\Ym$, and therefore gives no clue on the robustness (error bound)
of the recovery.
It is well known
in the compressed sensing literature that stronger conditions are needed to
guarantee algorithmic robust recovery \cite{Foucart2013}.  Upper and lower
bounds on the performance of the ML solution of M-SBL for the noisy case have
been derived in \cite{Tan2010a} for $K_a\leq L$ and in \cite{Koo2018} for $K_a = \mathcal{O}(L^2)$,
but these bounds contained parameters which were
exponentially hard to compute for a given matrix $\Am$ and so no concrete
scaling of $K_a$ could be given.  A coherence based argument was given in
\cite{Pal2014} to analyse the performance of a covariance based LASSO algorithm,
but it was only possible to guarantee recovery for up to $K_a = \mathcal{O}(L)$
coefficients.  This is a well known limitation of coherence based arguments,
known as the "square-root bottleneck" \cite{Foucart2013}.  In this work we are
able to circumvent this bottleneck by proving the restricted isometry property
(RIP) of a properly centered and rescaled version of $\Am\odot\Am$ for random
$\Am$. This allows us to prove recovery guarantees for both the NNLS algorithm
of \cite{docomo} and a constrained variant of the ML solution, showing that $K_a =
\mathcal{O}(L^2)$ coefficients can be recovered.
Although the
constrained ML yields a combinatorial minimization with exponential complexity
and therefore is not useful in practice, we can show that the scaling law for
successful detection of the activity pattern of the constrained ML scheme is
the same (up to logarithmic factors in the scaling of $M$)
as what was found for NNLS. Therefore, we conjecture that the
(low-complexity) ML algorithm achieves the same scaling law. We provide an intuitive
argument which, at least heuristically, explains why we expect the componentwise optimization
to converge to the global optimum in the considered scaling regime.
We would like to
mention that an analysis of the constrained ML estimator was recently presented
in \cite{Kha2017}. However, the results in \cite{Kha2017} are based on a RIP
result that was first claimed and then withdrawn by the same authors
\cite{Kha2019}.  Hence, our result based on a new RIP and a few consequent
modifications which we duly prove, essentially rigorizes the analysis presented
in \cite{Kha2017}.

The full characterisation of the global (unconstrained) ML solution and the conditions
under which the iterative estimate coincides with it remains open.
Some progress has recently been made in \cite{Che2020a}, where it was shown the global optimality
of the algorithmic solution can be checked, given $\Am$ and the true $\gammam^\circ$,
by a linear feasibility program.
In contrast, our recovery guarantees for the NNLS
algorithm hold for all $K_a$-sparse $\gammam$ and are given in closed form (up to unspecified constants).
Based on the asymptotic Gaussianity of ML estimators in general, it was shown in \cite{Che2020a} that
for large $M$ the distribution of the ML estimation error, for a fixed $\gammam^\circ$,
can be characterized numerically by the solution of a quadratic program. 

The coordinate-wise optimization algorithm for M-SBL was also independently
re-discovered and investigated in the context of source localisation
\cite{Abe2013,Gle2014,Yan2015}. It was noted in \cite{Abe2013} that the update
equation can be equivalently derived by an iterative weighted least-squares
(WLS) approach, which asymptotically minimizes the variance of the estimation.
The resulting algorithm has therefore been called \emph{iterative asymptotic
sparse minimum variance stochastic ML} (SAMV-SML).  Recently, in \cite{Ram2018}
a similar WLS estimator was derived and an iterative algorithm was given to
find an approximation of the WLS minimizer. The performance reported was very
similar to M-SBL, but at a much higher complexity-per-iteration of
$\mathcal{O}(L^4K_\text{tot}^2)$, compared to $\mathcal{O}(L^2K_\text{tot})$
for the coordinate-wise optimization with rank-1 updates.

Finally, we focus on unsourced random access with a massive MIMO BS. It is
evident that the AD problem and the random access problem are related.  In
fact, one can immediately obtain a random access scheme from an AD scheme as
follows: assign to each user a unique set of pilot signature sequences
(codewords), such that a user, when active, will transmit the signature
corresponding to its information message.  Since the number of pilot signatures
is $K_{\rm tot} \gg K_a$, this scheme involves only an expansion of the number
of total users from $K_{\rm tot}$ to $K'_{\rm tot} = K_{\rm tot} 2^B$ where $B$
is the number of per-message information bits.  This idea was recently
presented in \cite{Sen2017}, where the MMV-AMP detector of
\cite{Kim2012,liu2017massive, Che2018} was used at the receiver side.  While
conceptually simple, this approach has two major drawbacks: 1) even for
relatively small information packets (e.g., $B = 100$ bits), the dimension of
the pilot matrix is too large for practical computational algorithms; 2) each
user has a different set of pilot sequences, and therefore the scheme is not
compliant with the basic assumption of unsourced random access, that users have
all the same codebook. 

In contrast, we present a novel scheme, build upon the concatenated coding
approach of \cite{Ama2020a}, that does not incur in the large dimension problem
and is independent of the number of ``inactive'' users.  In our scheme, the
message of $B$ bits of each user is split into a sequence of submessages of
potentially different lengths.  These submessages are encoded via a tree code
(the same for each user), such that the encoded blocks have the same length of
$J$ bits.  Then, each user transmits its sequence of $J$-bits blocks in
consecutive blocks of $L$ dimensions, using the same $L \times 2^J$ pilot
matrix (where blocks are encoded in the matrix columns).  The inner detector
perform our ML activity detection scheme and for each slot recovers the set of
active columns of the pilot matrix.  These are passed to the outer tree code,
which recovers each user message by  ``stitching together'' the sequence of
submessages.  We show that an arbitrary small probability of error is
achievable at any $E_b/N_0$ provided that a sufficiently large number of base
station antennas is used, and that the sum spectral efficiency can grow as
$\mathcal{O}(L\log(L))$.  This can be achieved in a completely non-coherent
way, i.e. it is at no point necessary to estimate the channel matrix
(small-scale fading coefficients). These are important properties to enable
easily deployable, low-latency, energy efficient communication in an IoT
setting.
\subsection{Notation}
We represent scalar constants by non-boldface letters {(e.g., $x$ or $X$)}, sets by calligraphic letters (e.g., $\Xc$),  vectors by boldface small letters (e.g., $\xv$), and matrices by boldface capital letters (e.g., $\Xm$). 
We denote the $i$-th row and the $j$-th column of a matrix $\bfX$ with the row-vector $\Xm_{i,:}$ and the column-vector $\Xm_{:,j}$ respectively.  We denote a diagonal matrix with elements $(s_1, s_2, \dots, s_k)$ by $\diag(s_1, \dots, s_k)$. We denote the vectorization operator by $\vec(.)$.
We denote the $\ell_p$-norm of a vector $\bfx$ and the Frobenius norm of
a matrix $\bfX$ by $\|\bfx\|_{\ellp{p}}$ and $\|\bfX\|_{\ellp{p}}$ resp. $\|\xv\|_0:=|\{i:x_i\neq 0\}|$ denotes the number of non-zero entries of a vector $\xv$.
The operator norm of a matrix  $\bfX$ is denoted by $\|\bfX\|_{op}$.
The  $k \times k$ identity matrix is represented by $\bfI_k$.
For an integer $k>0$, we use the shorthand notation $[k]$ for $\{1,2,\dots, k\}$.
We use superscripts $(\cdot)^\transp$ and $(\cdot)^\herm$ for transpose and Hermitian transpose.
$\odot$  denotes the elementwise product of vectors or matrices of the 
same size.
$\langle\xv,\yv \rangle:=\xv^\herm\yv$ denotes the Euclidean scalar product between two vectors.
We define universal constants to be numbers, which are independent of all system parameters.
Such constants are typically denoted by $c,C,c',c_0,c_1$ etc., and different universal constants may be denoted by the
same letter.
$\log(x)$ denotes the natural logarithm of $x$.
\section{Problem Formulation}  \label{problem-formulation}

\subsection{Signal Model}  \label{sec:signal-model}

We consider a classical block-fading wireless channel  between each user and the BS where the channel coefficients remain constant 
over coherence blocks consisting of $\dimPilots$ signal dimensions in time-frequency \cite{tse2005fundamentals}, and change from block to block
according to some stationary and ergodic fading process.  
In general, the BS devotes some time-frequency slots to AD, i.e., to the purpose of identifying the active users who want to request some transmission resource. 
Such slots are generally non-adjacent in the time-frequency domain, since they are multiplexed with other slots, dedicated to
data transmission of the already connected users. Since typically the number of signal dimensions per AD slot is  not larger than one coherence block, 
without loss of generality we assume that each AD slot consists of $\dimPilots$ signal dimensions and coincides with a coherence block. 
We denote the set of all potential users (which may or may not be active) as $\clK_\text{tot}$, of size $K_\text{tot}:=|\clK_\text{tot}|$. 
Each user $k \in \clK_\text{tot}$ is given a user-specific and a priori known pilot sequence. The pilot sequence of user $k$ is denoted as 
$\bfa_k =(a_{k,1}, \dots, a_{k,\dimPilots})^\transp \in \bC^{\dimPilots}$. If user $k$ is active, it transmits the components of $\bfa_k$ in the AD slot
of $\dimPilots$ signal dimensions.  
Denoting by $\bfh_k$ the $M$-dimensional channel vector (small-scale fading coefficients) of the
user $k \in \clK_\text{tot}$ to $M$  antennas at the BS, we can write the received signal at the BS over the AD slot as
\begin{align}\label{sig_eq1}
\bfy[i]=\sum_{k \in \clK_\text{tot}} b_k\sqrt{g_k} a_{k,i} \bfh_k+ \bfz[i], \;\;\;\; i \in [L], 
\end{align}
where $[\dimPilots]:=\{1,\dots, \dimPilots\}$, $g_k \in \bR_+$ denotes the LSFC (channel  strength)
of the user $k \in \clK_\text{tot}$, 
$b_k \in \{0,1\}$ is a binary variable with $b_k=1$ for active and $b_k=0$ for inactive users and
$\bfz[i]\sim\cg(0, \sigma^2 \bfI_M)$ denotes the additive white Gaussian noise (AWGN)
at the $i$-th signal dimension. 

Denoting by $\bfY=[\bfy[1], \dots, \bfy[\dimPilots]]^\transp$ the $\dimPilots \times M$ received signal over $\dimPilots$ signal dimensions and $M$ BS antennas, we can write \eqref{sig_eq1} more compactly as
 \begin{align}\label{pilot_sig}
\bfY=\bfA \Gammam^{\frac{1}{2}} \bfH + \bfZ,
\end{align}
where $\bfA=[\bfa_1, \dots, \bfa_{K_\text{tot}}]$ denotes the $\dimPilots\times K_\text{tot}$ matrix of pilot sequences
of the users in $\clK_\text{tot}$,
where $\Gammam=\bfB \bfG$ where $\bfG$ is a $K_\text{tot} \times K_\text{tot}$ diagonal matrix consisting of the LSFCs  $(g_1, \dots, g_{K_\text{tot}})^\transp$ and where $\bfB$ is a $K_\text{tot}\times K_\text{tot}$ diagonal matrix consisting of the binary activity patterns $(b_1, \dots, b_{K_\text{tot}})^\transp$ of the users, and where $\bfH=[\bfh_1, \dots, \bfh_{K_\text{tot}}]^\transp$ denotes $K_\text{tot} \times M$  matrix containing the $M$-dimensional normalized channel vectors of the users. 

In line with the classical massive MIMO setting \cite{Mar2016},
we assume for simplicity an independent Rayleigh fading model, such that the channel vectors
$\{\bfh_k: k \in \clK_\text{tot}\}$ are independent from each other and are  spatially white (i.e.,
uncorrelated along the antennas), that is,  $\bfh_k \sim \cg(0, \bfI_M)$.
We would like to mention here that massive MIMO has been now investigated under many more realistic
propagation conditions involving antenna correlation and partial Line-of-Sight Rician fading
\cite{Zha2013,Gao2015}. 
Nevertheless, for consistency with respect to \cite{liu2017massive, Che2018}, where this assumption is made, 
and for the sake of isolating the fundamental aspects of the problem without additional model complication,
we stick to the simple i.i.d. Rayleigh fading model. 
A thorough study of the effect of different small-scale fading statistics
(e.g., introducing correlation across the antennas for each user channel) is left for 
future work.

The user pilots  are normalized to unit energy per symbol, i.e.,  
$\|\bfa_k\|_2^2=\dimPilots$. Then, the average SNR of a generic active user $k \in \clK_\text{tot}$ over $\dimPilots$ pilot dimensions is given by
\begin{align}
\snr_k=\frac{\|\bfa_k\|_2^2 \gamma_k \bE[\|\bfh_k\|_2^2]}{\bE[\|\bfZ\|_\sfF^2]}=\frac{\dimPilots \gamma_k M}{\dimPilots M \sigma^2}=\frac{\gamma_k}{\sigma^2}=\frac{g_k}{\sigma^2},
\label{snr_def}
\end{align}
where $\gamma_k=b_k g_k=g_k$ ($b_k=1$ for active users) is the $k$-th diagonal element of $\Gammam$.
We call the vector $\gammam=(\gamma_1, \dots, \gamma_{K_\text{tot}})^\transp$ or equivalently the diagonal matrix $\Gammam=\diag(\gammam)$ 
the ``active LSFC pattern'' of the users in $\clK_\text{tot}$. 
We denote by $\clK_a \subseteq \clK_\text{tot}$ the subset of active users in the current AD slot, with size $K_a:=|\clK_a|$. Thus, $\gammam$ is a non-negative 
sparse vector with only $K_a$ nonzero elements.
The goal of AD is to identify the subset of active users $\clK_a$ or a subset thereof consisting of users with sufficiently strong channels 
$\clK_a(\nu):=\{k \in \clK_\text{tot}: \gamma_k>\nu \sigma^2\}$, for a pre-specified  threshold $\nu>0$, from the noisy observations as in \eqref{pilot_sig}. 
As a side goal, we wish also to estimate the LSFCs $\gamma_k$ of the active users (at least those above threshold).
This information may be useful in practice to accomplish tasks such as user-BS association, user scheduling, and possibly other high-level network optimization
tasks where the knowledge of the user channel strength is relevant.

Since we assume that the channel vectors are spatially white and Gaussian, the columns of $\bfY$ in \eqref{pilot_sig} are i.i.d. Gaussian 
vectors with $\bfY_{:,i} \sim \cg(0, \Sigmam_\bfy)$ where 
\begin{align}\label{eq:true_cov}
\Sigmam_\bfy=\bfA \Gammam \bfA^\herm + \sigma^2 \bfI_\dimPilots=\sum_{k=1}^{K_\text{tot}} \gamma_k \bfa_k \bfa_k^\herm + \sigma^2 \bfI_\dimPilots
\end{align}
denotes the covariance matrix, which is common among all the columns $\bfY_{:,i}$, $i \in [M]$.
We also define the empirical/sample covariance  of the columns of the observation  $\bfY$ in \eqref{pilot_sig} as
\begin{align}\label{eq:samp_cov}
\widehat{\Sigmam}_\bfy=\frac{1}{M} \bfY \bfY^\herm =\frac{1}{M} \sum_{i=1}^M \bfY_{:,i} \bfY_{:,i}^\herm.
\end{align}

\section{Proposed Algorithms for Activity Detection}

In this section, we discuss two algorithms for AD and LSFC estimation.

\subsection{Maximum Likelihood Estimation}   \label{ML-sect}

We first consider the Maximum Likelihood (ML) estimator of $\gammam$ by making explicit use of Gaussianity of the users channel vectors. 
We introduce the negative log-likelihood cost function
\begin{align}
    \label{eq:likelihood_function}
f(\gammam)&:=-\frac{1}{M}\log p(\bfY| \gammam)\stackrel{(a)}{=}-\frac{1}{M}\sum_{i=1}^M \log p(\bfY_{:,i}| \gammam)\\
&\propto  \log | \bfA \Gammam \bfA^\herm+ \sigma^2 \bfI_{\dimPilots}| + \tr\left ( \Big( \bfA \Gammam \bfA^\herm+ \sigma^2 \bfI_{\dimPilots}\Big ) ^{-1} \widehat{\Sigmam}_\bfy \right),\label{eq:ML_cost}
\end{align} 
where $(a)$ follows from the fact that the columns of $\bfY$ are i.i.d. (due to the spatially white user channel vectors), and where $\widehat{\Sigmam}_\bfy$ denotes the sample covariance matrix of the columns of $\bfY$ as in \eqref{eq:samp_cov}.
Note that for spatially white channel vectors considered here, $\widehat{\Sigmam}_\bfy \to \Sigmam_\bfy$ as the number of antennas $M \to \infty$. 
It is apparent that the likelihood function $p(\Ym|\gammam)$ depends on $\Ym$ only through
the covariance matrix $\widehat{\Sigmam}_\yv$. Therefore, $\widehat{\Sigmam}_\yv$ is a {\em sufficient statistic}
for the estimation of $\gammam$ or any function thereof.  Especially in a Massive MIMO scenario, where $M > \dimPilots$,
the use of the covariance matrix $\widehat{\Sigmam}_\yv \in \CC^{\dimPilots\times \dimPilots}$ instead of
the raw measurements $\Ym \in \CC^{M\times \dimPilots}$ results in a significant dimensionality reduction.
\def\pscone{\clS_{\dimPilots}^+}
\def\pqcone{\clQ_{\dimPilots}^+}
Now let us focus on the ML cost function in \eqref{eq:ML_cost}. Assuming the number of active users $K_a$ is known, the 
{\em constrained ML estimator} of $\gammam$ is given by 
\begin{align}
    \gammam^*_\text{c-ML} = \argmin_{\gammam \in \Theta^+_{K_a}} f(\gammam).\label{true_ML}
\end{align}
where the constraint set $\Theta^+_{K_a}=\{\gammam\in \bR_+^{K_\text{tot}}: \|\gammam\|_0 \leq K_a\}$ is the (non-convex) set
of non-negative $K_a$-sparse vectors. 
There are two problems with this estimator: 1) $K_a$ is generally not known a priori, and 2) 
the minimization in (\ref{true_ML}) is combinatorial and has exponential complexity in $K_{\rm tot}$, which can be very large. Therefore, 
this ML estimator has no practical value. Nevertheless, its performance yields a useful bound to the 
performance of other ``relaxed'' versions of ML estimation. 
In particular, we are interested in the \textit{relaxed ML estimator} of $\gammam$ given by 
\begin{align}
    \gammam^*_\text{r-ML} = \argmin_{\gammam \in \bR_+^{K_\text{tot}}} f(\gammam).\label{a_ML}
\end{align}
It is not difficult to check that $f(\gammam)$ in \eqref{eq:ML_cost} is the sum of a concave
function and a convex function, so also the problem in (\ref{a_ML}) is not convex in general.
Notice also that the estimator in (\ref{a_ML}) does not require any prior knowledge of $K_a$. 

In the following, for the sake of analysis, we shall denote the true 
vector of LSFCs as $\gv^\circ$ and the true activity pattern as $\bv^\circ$. 
Next, we consider the  performance of the constrained ML estimator  \eqref{true_ML}.
The idea of the proof is based on \cite{Kha2017},
which was relying on a RIP result \cite{Kha2019} which was then withdrawn since the proof had a flaw. 
In Appendix \ref{appendix:ML} we give a complete and streamlined proof for the case, where the true 
vector of LSFCs $\gv^\circ$ is known at the receiver and all entries satisfy $g_k^\circ \in [g_\text{min},g_\text{max}]$.
Therefore, the goal consist of estimating the activity pattern $\bv^\circ$ and the active LSFC pattern is eventually given by 
$\gammam^*_{\rm c-ML}  = \bv^* \odot\gv^\circ$, where $\bv^*$ is the estimate of $\bv^\circ$. We hasten to say that our proof technique
extends easily also to the case where $\gv^\circ$ is unknown, provided that the per-component upper and lower bounds 
$g_\text{min}$ and $g_\text{max}$ are known, using the arguments of \cite{Kha2017}. We have omitted this general case for the sake of 
brevity, since it requires a few more technicalities which can be found in \cite{Kha2017}.

For the case at hand, we define the constrained ML estimator of the activity pattern
$\bv^\circ \in \{0,1\}^{K_\text{tot}}$ as
\beq
\bv^* :=
\argmin_{\bv \in \Theta_{K_a}}  f(\bv\odot \gv^\circ), 
\label{eq:ml_binary}
\eeq
with $f(\cdot)$ as defined in \eqref{eq:likelihood_function} and
$\Theta_{K_a} = \{\bv\in\{0,1\}^{K_\text{tot}}:\sum b_k = K_a\}$, the set of binary $K_a$-sparse vectors.
We have the following result: 

\begin{theorem}
    \label{thm:ml:support}
    Let the LSFCs be such that for all $k$ it holds that 
    $g_\text{min}\leq g_k \leq g_\text{max}$.
    Let $\Am \in \CC^{\dimPilots\times K_\text{tot}}$, be the pilot matrix with columns drawn uniformly
    i.i.d. from the
    sphere of radius $\sqrt{\dimPilots}$ and let $K_\text{tot} > \dimPilots^2$.
    For any $\bv^\circ\in\Theta_{K_a}$ the estimate $\bv^*$, defined in \eqref{eq:ml_binary},
    satisfies
    $\bv^* = \bv^\circ$ with probability exceeding $1 - 2\epsilon-\exp(-C\dimPilots)$
    (jointly, on a draw of $\Am$ and a random channel realization), provided that
    \beq
        K_a \leq c\frac{\dimPilots^2}{\log^2(eK_\text{tot}/\dimPilots^2)},
        \label{eq:rip_condition}
    \eeq
    and 
    \beq
    \begin{split}
        &M \geq \frac{4}{1-\delta}
        \left(\frac{C'g_\text{max}\left(2\log (\frac{eK_\text{tot}}{2K_a}) + \frac{\log(2/\epsilon)}{\max\{K_a,\dimPilots\}}\right) \max\left\{1,\frac{K_a}{\dimPilots}\right\}
        + \frac{\sigma^2}{\dimPilots}}{g_\text{min}}\right)^2
               \log\left(3eK_\text{tot}K_a\frac{1+\epsilon}{\epsilon}\right)
    \label{eq:ml_M_condition}
    \end{split}
    \eeq
    where $0<\delta<1$ and
    $0<c,C,C'$ are universal constants
    that may depend on each other but not on the system parameters.
    The precise relation is given in the proof.\hfill $\square$
\end{theorem}
\begin{proof}
    See Appendix \ref{appendix:ML}
\end{proof}
Theorem \ref{thm:ml:support} gives sufficient conditions under which the error probability of 
the estimator \eqref{eq:ml_binary} vanishes and it shows that $K_a$ can be larger than $L$,
although then $M$ has to grow at least as fast as $(K_a/L)^2$.
Simple algebra (omitted for the sake of brevity) shows the following:
\begin{corollary}
    \label{cor:ml}
    Let $\Am$ be as above and let $M,K_a,\dimPilots\to\infty$, then it is possible to choose
    \beq
    K_a = \mathcal{O}(\dimPilots^2/\log^2(K_\text{tot}/\dimPilots^2))
    \label{eq:scaling_K_a}
    \eeq
    and
    \beq
    M = \mathcal{O}\left(K_a(g_\text{max}/g_\text{min})^2\log^2(K_\text{tot}/K_a)\log(K_\text{tot}K_a)\right)
    \eeq
    such that the estimation error of the ML estimator \eqref{eq:ml_binary} vanishes.   \hfill $\square$
\end{corollary}
Note, that the scaling condition \eqref{eq:scaling_K_a} can be replaced with the stricter condition
\beq
K_a = \mathcal{O}(\dimPilots^2/\log^2(K_\text{tot}/K_a)).
\label{eq:scaling_K_a_2}
\eeq
This is because $K_a \leq L^2$ and therefore
$\dimPilots^2/\log^2(K_\text{tot}/K_a) \leq \dimPilots^2/\log^2(K_\text{tot}/\dimPilots^2)$,
which implies
\beq
\dimPilots^2/\log^2(K_\text{tot}/K_a) = \mathcal{O}(\dimPilots^2/\log^2(K_\text{tot}/\dimPilots^2)).
\eeq

As said, the minimization in (\ref{true_ML}) or (\ref{eq:ml_binary})  is in
general computationally unfeasible (beyond the problem of not knowing $K_a$).
Next, we consider the relaxed ML estimator \eqref{a_ML}, where the domain of search 
is relaxed to the whole non-negative orthant.
This estimator is formally equivalent to the ML estimator of the model parameters
in the sparse Bayesian learning framework,
posed in \cite{Tip2001}. In \cite{Tip2001} a low-complexity iterative algorithm was given
and it was shown in \cite{Fau2002}
that the iterative algorithm is guaranteed
to converge to at least a local minimum of \eqref{eq:likelihood_function}.
We derive the iterative update equations here for completeness and show that they
can be efficiently implemented by rank-1 updates.
While this algorithm is not know to converge to the exact minimum of \eqref{eq:likelihood_function},
empirical evidence suggests it converges very well. 
The algorithm proceeds as follows:\\
For each coordinate $k \in [K_\text{tot}]$, define the scalar function
$f_k(d)=f(\gammam+ d \bfe_k)$ where $f(\gammam)$ is the likelihood function \eqref{eq:ML_cost} and $\bfe_k$
denotes the $k$-th canonical basis vector with a single $1$ at its $k$-th coordinate and zero elsewhere.
Setting $\Sigmam=\Sigmam(\gammam)=\bfA \Gammam \bfA^\herm + \sigma^2 \bfI_{\dimPilots}$
where $\Gammam=\diag(\gammam)$ and applying the well-known Sherman-Morrison rank-1 update identity \cite{sherman1950adjustment} we obtain that
\begin{align}
\big (\Sigmam + d \bfa_k \bfa_k^\herm \big )^{-1}=\Sigmam^{-1} - \frac{d\, \Sigmam^{-1} \bfa_k \bfa_k^\herm \Sigmam^{-1}}{1+ d \, \bfa_k^\herm \Sigmam^{-1} \bfa_k}. \label{dumm_ml_app2}
\end{align} 
Using \eqref{dumm_ml_app2} and applying the well-known determinant identity 
\begin{align}
\big | \Sigmam + d \bfa_k \bfa_k^\herm \big | =(1+ d\, \bfa_k^\herm \Sigmam ^{-1} \bfa_k) \big | \Sigmam\big |,
\end{align}
we can simplify $f_k(d)$ as follows
\begin{align}
f_k(d)= c + \log (1+ d\, \bfa_k^\herm \Sigmam ^{-1} \bfa_k) 
 -  \frac{  \bfa_k^\herm \Sigmam^{-1} \widehat{\Sigmam}_\bfy \Sigmam^{-1} \bfa_k }{1+ d \, \bfa_k^\herm \Sigmam^{-1} \bfa_k} d \label{dumm_ml_app3}
\end{align}
where $c=\log \big | \Sigmam \big | +\tr( \Sigmam^{-1} \widehat{\Sigmam}_\bfy)$ is a constant term independent of $d$. Note that from \eqref{dumm_ml_app3}, $f_k(d)$ is well-defined only when $d>d_0:=-\frac{1}{\bfa_k^\herm \Sigmam ^{-1} \bfa_k}$.
Taking the derivative of $f_k(d)$ yields
\begin{align}
f_k'(d)= \frac{  \bfa_k^\herm \Sigmam^{-1}\bfa_k }{1+ d \, \bfa_k^\herm \Sigmam^{-1} \bfa_k} - \frac{  \bfa_k^\herm \Sigmam^{-1} \widehat{\Sigmam}_\bfy \Sigmam^{-1} \bfa_k }{(1+ d \, \bfa_k^\herm \Sigmam^{-1} \bfa_k)^2}.
\end{align}
The only solution of $f_k'(d)=0$ is given by 
\begin{align}
d^*= \frac{ \bfa_k^\herm \Sigmam^{-1} \widehat{\Sigmam}_\bfy \Sigmam^{-1} \bfa_k -  \bfa_k^\herm \Sigmam^{-1}\bfa_k }{(\bfa_k^\herm \Sigmam^{-1}\bfa_k )^2}.
\label{eq:ml_update}
\end{align}
Note that $d^* \geq d_0=-\frac{1}{\bfa_k^\herm \Sigmam^{-1}\bfa_k }$, thus, one can check from \eqref{dumm_ml_app3} that $f_k$ is
indeed well-defined at $d=d^*$.
Moreover, we can check from \eqref{dumm_ml_app3} that
$\lim_{\epsilon \to 0^+} f_k(d_0+\epsilon)=\lim_{d \to \infty} f_k(d)=\infty$, thus, $d=d^*$ must be the global minimum of $f_k(d)$ in $(d_0, \infty)$. Note that since after the update we have $\gamma_k \leftarrow \gamma_k + d$, to preserve the positivity of $\gamma_k$, the optimal update step $d$ is in fact given by $\max \big \{d^*, -\gamma_k\big \}$ as illustrated in Algorithm \ref{tab:ML_coord}.

The exact characterization of the performance of this algorithm
remains at the moment an open problem. Specifically, it is not known under which
conditions the iterative algorithm actually reaches the global minimum of \eqref{eq:likelihood_function}.
A heuristic intuition for why the local
minima become rare in the large scale limit may be obtained as follows.
Let us first note some property of the negative log-likelihood cost function \eqref{eq:likelihood_function}.
Define
\beq
\Sigmam(\gammam) := \bfA \Gammam\bfA^\herm + \sigma^2\Id_{\dimPilots}
\label{eq:Sigma_def}
\eeq
and let 
\beq
\phi(\Sigmam):=-\log| \Sigmam| + \tr(\Sigmam\widehat{\Sigmam}_\yv).
\label{eq:phi}
\eeq
Since $\Sigmam(\gammam)$ is positive definite for every non-negative vector
$\gammam$ and $\sigma^2>0$, it is also invertible and the negative log-likelihood cost function can be expressed as
$f(\gammam) = \phi((\Sigmam(\gammam))^{-1})$. Now $\phi:\CC^{\dimPilots\times \dimPilots}\to \RR$
is {\em strictly} convex. Hence, it has a unique minimal value over a convex set.
Let $\Sigmam_*^{-1}$ denote the unique positive definite matrix with
$0\prec\Sigmam_*^{-1}\preceq 1/\sigma^2$ that minimizes
\eqref{eq:phi}, and let
$\Sigmam_*$ be its inverse. Now if the set of pilot sequences $\{\av_k: k\in
    \mathcal{K}_\text{tot}\}$ is such that the set $\{\sum_{k=1}^{K_\text{tot}}
\gamma_k\av_k\av_k^\herm:\gamma_k \geq 0\}$ spans the whole cone of
positive semidefinite matrices, then $\Sigmam_*$ can be represented as
$\Sigmam_* = \Sigmam(\gammam^*)$ and therefore $\gammam^*$ is a global minimizer of $f(\gammam)$ over
$\{\gammam:\gamma_i\geq 0\}$, i.e $\gammam^*_\text{r-ML} = \gammam^*$. Since there
are no local minimizers, the componentwise optimization algorithm will necessarily converge
to a global minimizer.
We cannot apply this argument though, because $\{\sum_{k=1}^{K_\text{tot}}
\gamma_k\av_k\av_k^\herm:\gamma_k \geq 0\}$ will never span the whole cone of
positive semidefinite matrices for any finite $K_\text{tot}$.
Nonetheless, if $K_\text{tot}$ is
large enough we expect the approximation of the cone of positive semidefinite matrices to be good enough
such that the log-likelihood function has few and small local
minima. That explains, at least heuristically, the good convergence behavior of the
componentwise optimization algorithm. 

Another open problem are the conditions, under which it is guaranteed that the solutions of
\eqref{true_ML} and \eqref{a_ML} coincide.
It is only possible to confirm the validity 
of $\gammam^*_\text{r-ML}$ a-posteriori,
i.e. , if $\gammam^*_\text{r-ML}$ happens to be $K_a$-sparse, then it follows that 
$\gammam^*_\text{c-ML} = \gammam^*_\text{r-ML}$. 
Hence, if $\gammam^*_\text{r-ML}$ is $K_a$-sparse {\rm and} the conditions on
$\Am, K_a,M,L$ and $K_\text{tot}$ of Theorem \ref{thm:ml:support} are fulfilled, 
then $\gammam^*_\text{r-ML}$ coincides with the correct solution $\gammam^\circ$ with high probability.
Intuitive explanations for the sparsity inducing nature of \eqref{a_ML} have been provided
in \cite{Tip2001} for the SMV case and in \cite{Wip2007} for the MMV case. 

\subsection{Non-Negative Least Squares}  \label{NNLS-sect}

\label{sec:nnls}
In this section we investigate a different approach to estimate $\gammam$ which can
be directly analyzed and for which we can provide a rigorous non-asymptotic
bound on the $\ell_1$ recovery error. Interestingly, analyzing this bound we find
that the estimation error vanishes for $M\to\infty$ under the same scaling condition
\eqref{eq:rip_condition}
for $K_a,L$ and $K_\text{tot}$
as in Theorem \ref{thm:ml:support}.
The strict convexity of $\phi(\cdot)$ defined in \eqref{eq:phi}
suggests the following approach:
first, find the matrix $\argmin_{\Xi\succeq 0} \phi(\Xi)$,
where $\{\Xi:\Xi\succeq 0\}$ denotes the set of positive
semidefinite matrices. A simple calculation shows that
the minimizer is simply the inverse of the empirical covariance matrix
$\widehat{\Sigmam}_\yv$.
Then, find  the estimate of $\gammam$ as 
\begin{align}\label{eq_nnls}
    \gammam^*=\argmin_{\gammam \in \bR_+^K} \|\Sigmam(\gammam) - \widehat{\Sigmam}_\bfy\|^2_{\sfF}.
\end{align}
Let us introduce the matrix $\bA\in\CC^{\dimPilots^2\times K_\text{tot}}$, whose $k$-th column is defined by:
\beq
\bA_{:,k} := \text{vec}(\av_k\av_k^\herm).
    \label{eq:bA_definition}
\eeq
and let $\wv = \vec(\widehat{\Sigmam}_\bfy - \sigma^2\Id_{\dimPilots})$ 
denote the
$\dimPilots^2 \times 1$ vector obtained by stacking the columns of
$\widehat{\Sigmam}_\bfy- \sigma^2\Id_{\dimPilots}$.
Then, we can write \eqref{eq_nnls} in the convenient form
\begin{align}\label{eq_nnls_vec}
  \gammam^*=\argmin_{\gammam \in \bR_+^K} \|\bA \gammam -\wv\|^2_2,
\end{align}
as a linear \textit{least squares} problem with non-negativity constraint,
known as \textit{non-negative least squares} (NNLS).
Such an algorithm was proposed for the activity detection problem in \cite{docomo}.

A key property of the matrix $\bA$ is that a properly
centered and rescaled version of it has the RIP.
Let us define the centered version of $\bA$, denoted by 
$\mathring{\bA}$ as the $\dimPilots(\dimPilots-1)\times K_\text{tot}$ dimensional matrix,
with the $k$-th column given by
\beq
\label{eq:a_centered}
\mathring{\bA}_{:,k}:=\text{vec}_\text{non-diag}(c_L\av_k\av_k^\herm-\diag(\av_k\av_k^\herm)).
\eeq
Where $\text{vec}_\text{non-diag}(\cdot)$ denotes the vectorization of only the non-diagonal elements,
which in the case of $\av_k\av_k^\herm-\diag(\av_k\av_k^\herm)$, are zero anyway.
The term $c_L = (L-1)/(L-\kappa_a)$ with $\kappa_a = \EE[\|a_{k,i}\|^4]$ ensures a proper normalisation.
Let $m=\dimPilots(\dimPilots-1)$, then the restricted isometry
constant $\delta_{2s}=\delta_{2s}(\mathring{\bA}/\sqrt{m})$ of
$\mathring{\bA}/\sqrt{m}$
of order $2s$ is defined as:
\begin{equation}
    \delta_{2s}:=\sup_{0<\|\vc{v}\|_{\ellp{0}}\leq 2s}\left|\frac{\lVert \mathring{\bA}
    \vc{v} \rVert^2_{\ellp{2}}}{m\|\vc{v}\|^2_{\ellp{2}}}-1\right|
\end{equation}
and if $\delta_{2s}\in[0,1)$ the matrix $\mathring{\bA}/\sqrt{m}$ is said to have RIP of order $2s$.
The normalization by $m$ is necessary to ensure that the expected
norm of the columns
of $\mathring{\bA}$ is of order $\mathcal{O}(1)$ for all $\dimPilots$,
which is a necessary condition for the RIP to hold with high probability.
It is well known that matrices with iid sub-Gaussian entries satisfy
the RIP of order $2s$ with high probability for 
$s = \mathcal{O}(m/\log(eK_\text{tot}/s))$ \cite{Foucart2013}.
The entries of $\mathring{\bA}$ though are neither sub-Gaussian nor
independent which makes the analysis more complicated. 
Nonetheless, recent results in \cite{Adamczak2011,Guedon2014:heavy:columns}
show that matrices which have independent columns (with possibly
correlated entries) satisfy the RIP with high probability if the columns have a bounded sub-exponential
norm. Using the results of \cite{Adamczak2011}
we can establish the following Theorem, which is central for
both Theorem \ref{thm:ml:support} and Theorem \ref{NNLS-theorem}.
\begin{theorem}
    \label{thm:rip_main}
    Let $\Am \in \CC^{\dimPilots\times K_\text{tot}}$, be the pilot matrix with columns drawn uniformly i.i.d. from the
    sphere of radius $\sqrt{\dimPilots}$. 
    Then, with probability exceeding $1- \exp(-c_\delta\sqrt{m})$ on a draw of $\Am$,
    it holds that $\mathring{\bA}/\sqrt{m}$ has the RIP
    of order $2s$ with RIP-constant $\delta_{2s}(\mathring{\bA}/\sqrt{m})< \delta$
    as long as
    \begin{equation}
        2s \leq C_\delta\frac{m}{\log^2(eK_\text{tot}/m)}
    \end{equation}
    for some constants $c,c_\delta,C_\delta>0$ depending only on $\delta$.
    \hfill $\square$
\end{theorem}
\begin{proof}
    See Appendix \ref{appendix:ripproof}.
\end{proof}
NNLS has a special property, as discussed for
example in \cite{slawski2013non} and referred to as the $\mathcal{M}^+$-criterion in \cite{kueng2016robust}, 
which makes it particularly suitable for recovering sparse vectors:
If  the row span of $\bA$ intersects the positive orthant, NNLS
implicitly also performs $\ell_1$-regularization.  
Because of these features, NNLS has recently gained interest in many applications in signal
processing \cite{song2017scalable},  compressed  sensing \cite{kueng2016robust}, and machine learning.  
In our case the $\mathcal{M}^+$--criterion is fulfilled in an optimally--conditioned
manner. Combined with the RIP of $\A..$ it  allows us to establish the following result:
\begin{theorem}  \label{NNLS-theorem}
Let $\Am \in \CC^{\dimPilots\times K_\text{tot}}$, be the pilot matrix with columns drawn uniformly i.i.d. from the
sphere of radius $\sqrt{\dimPilots}$. 
There exist universal 
constants $c_i > 0$, $i=1,...,5$, depending only on some common parameter, but not on the system parameters,
(see the proof in Appendix \ref{appendix:nnlsproof} for details) such that, if
  \begin{equation}    
    s \leq c_1\frac{\dimPilots^2}{\log^2(eK_\text{tot}/\dimPilots^2)},
    \label{eq:phase_transition}
  \end{equation}
  then with probability exceeding $1-\exp(-c_5L)$ (on a draw of $\Am$) the following holds: 
  For all $s$-sparse activity pattern vectors $\gammam^\circ$
  {and all realizations of $\widehat{\Sigmam}_\yv$}, the solution $\bxnnls$ of \eqref{eq_nnls_vec} fulfills for $1\leq p\leq 2$ the bound:
  \begin{equation}
    \begin{split}
      \lVert \bxtrue &- \bxnnls\rVert_{\ellp{p}}
      \leq \frac{c_2}{s^{1-\frac{1}{p}}} \sigma_s (\bxtrue)_{\ellp{1}} +
      \frac{c_3}{s^{\frac{1}{2}-\frac{1}{p}}} \left( \frac{\sqrt{\dimPilots}}{\sqrt{s}} +
        c_4 \right) \frac{\lVert \vc{d} \rVert_{\ellp{2}}}{\dimPilots},
    \end{split}
    \label{eq:short:thm:mse}
  \end{equation}
  where  $\sigma_s (\bxtrue)_1$ denotes the $\ell_1$--norm of
  $\gammam^\circ$ after removing its $s$ largest components and where
   \beq
   \bfd=\vec\left(\widehat{\Sigmam}_\bfy - \sum_{k=1}^{K_\text{tot}} \gamma^\circ_k
   \bfa_k \bfa_k^\herm - \sigma^2\Id_{\dimPilots}\right).
   \eeq
  \hfill $\square$
\end{theorem}

The proof is based on
a combination of the NNLS results of \cite{kueng2016robust} and an
extension of RIP-results for the heavy-tailed column-independent model
\cite{Adamczak2011,Guedon2014:heavy:columns}.
The common parameter on which the constants $c_i$ depend is the RIP constant of a
properly centered version of $\bA$, defined in \eqref{eq:bA_definition}. We state this dependence explicitly
to emphasize that
Theorem \ref{NNLS-theorem} holds also for more general random models for $\Am$, for which $\bA$ has the RIP.
Then the constants
$c_2,c_3,c_4$ can be computed explicitly (see Appendix \ref{appendix:nnlsproof}) depending on the RIP
constant of the other matrix model.
The probability
term $1-\exp(-c_5L)$ is precisely the probability that the centered version of the random matrix $\bA$ has the RIP.
The result is uniform meaning that with high probability (on a draw of
$\Am$) it holds {\em for all} $\gammam^\circ$ 
and {\em for all}
realizations of the random variable $\widehat{\Sigmam}_\yv$.
For
$s=K_a=\lVert\gammam^\circ\rVert_0$ it implies (up to
the $\|\vc{d}\|_{2}$-term) exact recovery since in this case
$\sigma_s (\bxtrue)_{\ellp{1}}=0$.  
A relevant extension of this result to the case $p\rightarrow\infty$ would be important but, in
this generality, it is not known whether one can hope for a linear scaling in $s$ (see, for example \cite[Theorem
3.2]{dirksen:ripgap}). Nonetheless, since $\|\cdot\|_{\ellp{\infty}}\leq\|\cdot\|_{\ellp{p}}$ our result
\eqref{eq:short:thm:mse} also implies an estimate for the
communication relevant $\ell_\infty$-case but with sub-optimal scaling
(we will discuss this below).  Furthermore improvements for this
particular case may be possible in the non-uniform or averaged case, as it has been investigated for the sub-Gaussian case in
\cite{slawski2013non}.

The analysis of the random variable $\| \dv\|_2$ given in Appendix \ref{cov_err_app} shows that,
for every realization of $\Am$ it holds that
\begin{align}
    \EE_{\Ym|\Am}[\|\bfd\|_2]
  &=\frac{\dimPilots}{\sqrt{M}}(\|\bxtrue\|_1+\sigma^2)
  \label{eq:d_expectation}
\end{align}
with a deviation tail distribution satisfying 
\beq
\PP_{\Ym|\Am}\left(\|\dv\|_2 > \sqrt{\alpha_\epsilon} \EE_{\Ym|\Am}[\|\dv\|_2]\right) \leq \epsilon
\label{eq:d_bound}
\eeq
for 
\beq
\alpha_\epsilon = c\log((e\dimPilots)^2/\epsilon)
\eeq
with some universal constant $c>0$.
The bounds \eqref{eq:d_expectation} and \eqref{eq:d_bound} are independent of the
realization of $\Am$, so the conditional expectation/probability
can be replaced by the total expectation/probability.
Assuming that $\Am$ is chosen independent of the channel realization, it holds that
with probability $(1-\epsilon)(1-\exp(-c_5L))\geq 1 - \epsilon - \exp(-c_5L)$ the pilot matrix $\Am$
satisfies the condition in Theorem \ref{NNLS-theorem} {\em and} the channel realization
$\dv$ satisfies \eqref{eq:d_bound}.
Setting $s = K_a$ in Theorem \ref{NNLS-theorem} (yielding $\sigma_s(\gammam^\circ)=0$), for $p=1$ we get the following:
\begin{corollary}
    \label{cor:nnls}
  With the assumptions as in Theorem \ref{NNLS-theorem}, the following holds:
  For any $K_a$--sparse
  $\bxtrue$ with
  \begin{equation}    
      K_a \leq c_1\frac{\dimPilots^2}{\log^2(eK_\text{tot}/\dimPilots^2)},
      \label{eq:phase_transition_K_a}
  \end{equation}
  the NNLS estimate $\bxnnls$ fulfills:
  \begin{equation}
      \begin{split}
      &\frac{\lVert \bxtrue -
      \bxnnls\rVert_{\ellp{1}}}{\|\bxtrue\|_{\ellp{1}}}
      \leq c_3\left(
          \sqrt{\dimPilots} +
          c_4\sqrt{K_a}
      \right) \frac{1+\frac{\sigma^2}{\|\bxtrue\|_{\ellp{1}}}}{\sqrt{M/\alpha_\epsilon}}
      \end{split}
      \label{eq:short:nsp4:p1}
  \end{equation}
  with
  probability at least $1-\epsilon - \exp(-c_5\dimPilots)$, where $c_1,c_3,c_4,c_5$ are the same constants
  as in Theorem \ref{NNLS-theorem}.
  \label{cor:nnls:p1}
  \hfill $\square$
\end{corollary}
Using the well-known inequality
$\|\gammam^\circ\|_1 \leq \sqrt{K_a} \|\gammam^{\circ}\|_2$,
Theorem \ref{NNLS-theorem} for the case $p=2$ gives:
\begin{corollary}
    Under the same conditions as in Corollary \ref{cor:nnls:p1}
  \begin{align}
      \frac{\lVert \bxtrue - \bxnnls\rVert_{\ellp{2}}}{\|\gammam^\circ\|_2} \leq
      c_3 \left (\sqrt{\dimPilots}+c_4\sqrt{K_a}\right)
      \frac{\left(1+\frac{\sigma^2}{\sqrt{K_a}\|\gammam^\circ\|_2}\right)}{\sqrt{M/\alpha_\epsilon}}
    \label{gamma_perf}
  \end{align}
  holds with probability at least $1- \epsilon - \exp(-c_5\dimPilots)$
  where $c_3,c_4,c_5$ are the same constants
  as in Theorem \ref{NNLS-theorem}
  provided that \eqref{eq:phase_transition_K_a} holds.  \hfill $\square$
  \label{cor:nnls:mse}
\end{corollary}
In conclusion,  the following scaling law is sufficient to achieve a vanishing estimation error.
\begin{corollary}
Let $M,K_a,\dimPilots\to\infty$ with
$K_a$ as in \eqref{eq:scaling_K_a}
and  $M=K_a^\kappa$ for $\kappa>1$ then for $p=1,2$
it holds with probability 1 that
\beq
\lim_{M\to\infty}\frac{\lVert \bxtrue -
        \bxnnls\rVert_{\ellp{p}}}{\|\bxtrue\|_{\ellp{p}}} = 0.
\eeq
\label{cor:nnls_scaling}
\hfill $\square$
\end{corollary}

This shows that the NNLS estimator \eqref{eq_nnls} can identify up to $O(\dimPilots^2)$ active users
by paying only a poly-logarithmic penalty  $O(\log^2 (\frac{K_\text{tot}}{K_a}))$
for increasing the number of potential users $K_\text{tot}$.
This is a very appealing property in practical IoT setups where, as already mentioned in the introduction, 
$K_\text{tot}$ may be very large. 
Note, that
the scaling of the identifiable users is the same as that of the (uncomputable)
restricted ML estimator, see Corollary \ref{cor:ml}, while the scaling of the minimum required $M$
agrees up to poly-logarithmic factors.

\subsection{Iterative Algorithms}  \label{Algo-sect}

Finding the ML estimate $\gammam^*$ in \eqref{a_ML} or the NNLS estimate 
\eqref{eq_nnls} requires the optimization of a function over
the positive orthant $\bR_+^{K_\text{tot}}$.
In Section \ref{ML-sect} we have derived the componentwise minimization condition \eqref{eq:ml_update}
of the log-likelihood cost function. 
Starting from an initial point $\gammam$, 
at each step of the algorithm we minimize  $f(\gammam)$ with respect to only one of its arguments $\gamma_k$ according to 
\eqref{eq:ml_update}.
We refer to the resulting scheme as an {\em iterative componentwise  minimization algorithm}. 
As discussed before, hopefully this will converge to the solution of \eqref{a_ML}.
Variants of the algorithm may differ in the way the initial point is chosen and in the way
the components are chosen for update.  The noise variance $\sigma^2$ can also be included as an additional 
optimization parameter and estimated along $\gammam$ \cite{Tip2001}.

The same iterative componentwise minimization approach can be used to solve (iteratively) the NNLS problem (\ref{eq_nnls}).
Of course, the component update step is different in the case of ML and in the case of NNLS. We omit the derivation of the 
NNLS component update since it consists of a straightforward differentiation operation.  Since NNLS is convex, in this case the componentwise 
minimization algorithm is guaranteed to converge to the solution of the NNLS problem (\ref{eq_nnls}). 
Given the analogy of the two iterative componentwise minimization algorithms for ML and for NNLS, we summarize them 
in a unified manner in Algorithm \ref{tab:ML_coord}.

\subsubsection{ML and NNLS with Knowledge of the LSFCs}
\label{sec:box-constraint}

Since the ML and NNLS algorithms are non-Bayesian in nature, they work well without
any a-priori information on the LSFCs. 
If $\gv^\circ$ (true values of the LSFCs of all users, active and not) is known, 
the algorithms can be slightly improved by projecting each $k$-th coordinate update on the interval $[0,g_k^\circ]$
(see step
\ref{alg:update_step}) in Algorithm \ref{tab:ML_coord}.
In this case the thresholding step can be improved by choosing the thresholds
relative to the channel strength
$\widehat{\mathcal{A}}_{\gv^0} = \{i:\widehat{\gamma}_i > \theta g_k^0\}$.
\begin{algorithm}[t]
	\caption{Activity Detection via Coordinate-wise Optimization } %
	\label{tab:ML_coord} 
{\small
	\begin{algorithmic}[1]
		\State {\bf Input:} The sample covariance matrix $\widehat{\Sigmam}_\bfy=\frac{1}{M} \bfY \bfY^\herm$ of the $\dimPilots \times M$ matrix of samples $\bfY$.
		
		\State {\bf Input:} The LSFCs of $K_\text{tot}$ users $(g_1, \dots, g_{K_\text{tot}})$ if available.
		
		\vspace{1mm}

		\State {\bf Initialize:} $\Sigmam=\sigma^2 \bfI_{\dimPilots}$, $\gammam={\bf 0}$.
		\vspace{1mm}
		
		\For { $i=1,2, \dots$}
		\State {Select an index $k \in [K_\text{tot}]$ corresponding to the $k$-th component of  $\gammam=(\gamma_1, \dots, \gamma_{K_\text{tot}})^\transp$  randomly or according to a specific schedule.}
		\vspace{2mm}
		\State {\bf If ML:} Set $d_0^*= \max \Big \{\frac{ \bfa_k^\herm \Sigmam^{-1} \widehat{\Sigmam}_\bfy \Sigmam^{-1} \bfa_k -  \bfa_k^\herm \Sigmam^{-1}\bfa_k }{(\bfa_k^\herm \Sigmam^{-1}\bfa_k )^2}, -\gamma_{k} \Big \}$.
		
		\State {\bf If NNLS:} Set $d_0^*=\max \Big \{  \frac{\bfa_k^\herm (\widehat{\Sigmam}_\bfy - \Sigmam ) \bfa_k}{\|\bfa_k\|_2^4}, -\gamma_{k} \}$. 
		
		\vspace{1mm}

		\State Set $d^*=\min \{ d^*_0, g_k - \gamma_k\}$ if LSFC $g_k$ is available and
		$d^*=d_0^*$ otherwise.
        \label{alg:update_step}
		
		\State Update $\gamma_{k} \leftarrow \gamma_{k}+ d^*$.
		\State Update 
        $\Sigmam^{-1}
\leftarrow \Sigmam^{-1} -
\frac{d^* \Sigmam^{-1} \bfa_k \bfa_k^\herm \Sigmam^{-1}}
{1+ d^* \bfa_k^\herm \Sigmam^{-1} \bfa_k}$
		\EndFor
		\State {\bf Output:}  The resulting estimate $\gammam$. 
	\end{algorithmic}}
\end{algorithm}

\section{Empirical Comparison: ML, NNLS and MMV-AMP}  \label{simulation-sect}

In this section, we compare the performance of ML, NNLS and MMV-AMP via numerical simulations.
\subsection{Simulation Setting and Performance Criteria}
We assume that the output of each algorithm is an estimate $\gammam^*$
of the active LSFC pattern of the users. We use the relative $\ell_1$ norm of the difference
$\|\gammam^* - \gammam^\circ\|_1/\|\gammam^\circ\|_1$
as a measure of estimate quality. The $\ell_1$ norm is the natural choice here, since the coefficients $\gamma_i$
represent signal received, i.e., they are related to the square of the signal amplitudes. Therefore, a more traditional ``Square Error''
($\ell_2$ norm), related to the 4th power of the signal amplitude, does not really have any relevant physical meaning for the underlying communication 
system. We define $\widehat{\clA}_c(\nu):=\{i: {\gamma}^*_i >\nu\sigma^2\}$, with  $\nu>0$,
as the estimate of the set of active users.
We also define the  misdetection and false-alarm probabilities as
\begin{align}
    P_\text{md}(\nu) =1 - \frac{\bE[| \clK_a \cap \widehat{\clA}_c|]}{K_a},\ \  P_\text{fa}(\nu) =\frac{\bE[| \widehat{\clA}_c \backslash \clK_a|]}{K_\text{tot}-K_a}
\end{align}
where $K_a$ and $K_\text{tot}$ denote the number of active and the  number of potential users,
respectively. By varying $\nu \in \bR_+$, we get the
\textit{Receiver Operating Characteristic} (ROC) \cite{poor2013introduction} of the
algorithms. For simplicity of comparison, in the results presented here we have restricted to the point of the ROC where
$P_\text{md}(\nu) = P_\text{fa}(\nu)$.

We consider several models for the distribution of the LSFCs $g_k$.
The simplest case is when all LSFCs are constant, $g_k \equiv 1$,
this corresponds to a scenario with perfect power control. 
We also consider the case of variable signal strengths such that $10\log_{10}(g_k)$ is randomly distributed uniformly
in some range $[10\log_{10}(g_\text{min}),10\log_{10}(g_\text{max})]$ (uniform distribution in dB scale). 
This corresponds to the case of partial power control, where users partially compensate for their physical pathloss 
and reach some target SNR out of a set of possible values. 
In practice, these prefixed target SNR values corresponds to the various Modulation and Coding Schemes (MCS)
of a given communication protocol, 
which in turn correspond to different data transmission rates (see for example the MCS modes of
standards such as IEEE 802.11 \cite{8022016} or 3GPP-LTE \cite{Ses2009}).
In passing, we notice here the importance of estimating not only the user activity pattern but their LSFCs, in order to perform rate allocation.
Such a distribution, for specific values of $g_\text{min}$ and $g_\text{max}$ was also considered in \cite{liu2017massive}.

\subsection{MMV-AMP}
\label{sec:vamp}

This version of AMP, as introduced in \cite{Kim2011}, is a Bayesian iterative
recovery algorithm for the MMV problem, i.e., it aims to recover an unknown matrix
with i.i.d. rows from linear Gaussian measurements. 
As said in the introduction, the use of MMV-AMP has been proposed in \cite{liu2017massive, Che2018}
for the AD problem in a Bayesian setting, 
where the LSFCs are either known, or its distribution is known. 
Since unfortunately the formulation of MMV-AMP is often lacking details and certain terms (e.g., 
derivatives of matrix-valued functions with matrix arguments) are left indicated without explanations, 
for the sake of clarity and in order to provide a self-contained exposition we briefly review this algorithm here
in the notation of this paper. 

We can rewrite the received signal as
\beq
    \Ym = \Am\Xm + \Zm
\eeq
with $\Xm = \Gm \Bm \Hm$. Let $\Xm_{k,:}$ denote the $k$-th row of $\Xm$.
Letting $\lambda = \frac{K_a}{K_\text{tot}}$ be the fraction of active users, in the Bayesian setting underlying the MMV-AMP algorithm it is assumed that
the rows of $\Xm$ are mutually statistically independent and identically distributed according to
\beq
p_X(\xv) = (1-\lambda)\delta_0 +
\lambda\int_0^{+\infty} \frac{e^{-\frac{\|\xv\|_2^2}{\zeta}}}{\pi \zeta} \mathrm{d}p_G(\zeta), 
\eeq
where $p_G(\cdot)$ is the distribution of the LSFCs, i.e., for each $k$, it is assumed that 
$\Xm_{k,:}$ is either the identically zero vector (with probability $\lambda$) or a conditionally 
complex i.i.d. $M$-dimensional Gaussian vector with mean 0 and conditional variance $g_k$. Furthermore, 
the $g_k$'s are i.i.d. $\sim p_G(\cdot)$. 
The conditional distribution of $\Xm_{k,:}$ given $g_k$ is obviously given by
\beq
p_{X|g}(\xv | g_k) = (1-\lambda)\delta_0 +
\lambda\frac{e^{-\frac{\|\xv\|_2^2}{g_k}}}{\pi g_k}.
\eeq
The MMV-AMP 
iteration is defined as follows: 
\begin{align}
    \Xm^{t+1} &= \eta_t(\Am^\herm \Zm^t + \Xm^t)     \label{eq:vamp_1} \\
    \Zm^{t+1}   &= \Ym - \Am\Xm^{t+1} + \frac{K_\text{tot}}{\dimPilots} \Zm^{t}
    \langle \eta_{t}^\prime(\Am^\herm \Zm^{t} + \Xm^{t})\rangle       \label{eq:vamp_2} 
\end{align}
with $\Xm^0 = 0$ and $\Zm^0 = \Ym$. The function
$\eta_t: \CC^{K_\text{tot}\times M} \to \CC^{K_\text{tot}\times M}$ is defined row-wise as
\begin{equation} \eta_t(\Rm) = \left [ \begin{array}{c} \eta_{t,1}(\Rm_{1,:}) \\ \vdots \\ 
\eta_{t,{K_\text{tot}}}(\Rm_{K_{\rm tot},:}) \end{array} \right ],  \label{etamatrix}
\end{equation}
where each row function $\eta_{t,k}: \CC^M \to \CC^M$ is chosen as 
the posterior mean estimate of the random vector $\xv$, with a priori distribution as the rows of $\Xm$ as given above,  
in the {\em decoupled} Gaussian observation model
\beq
    \rv = \xv + \zv,
    \label{eq:vamp_decoupled}
\eeq
where $\zv$ is an i.i.d. complex Gaussian vector with components $\sim \mathcal{CN}(0,\Sigmam_t)$.
When $\gv$ is known, such posterior mean estimate is conditional on the knowledge of $g_k$, i.e., we define
\begin{equation} \label{PME-cond}
\eta_{t,k}(\rv) = \widetilde{\eta}_t(\rv, g_k) := \EE[ \xv | \rv, g_k].
\end{equation}
If $\gv$ is not known, the posterior mean estimate is unconditional, i.e., we define (with some abuse of notation)
\begin{equation} \label{PME-uncond}
\eta_{t,k}(\rv) = \widetilde{\eta}_t(\rv)  := \EE[ \xv | \rv].
\end{equation}
 Notice that in the latter case $\eta_{t,k}(\cdot)$ does not depend on $k$, i.e., the same mapping $\widetilde{\eta}_t(\cdot)$ 
 is applied to all the rows in (\ref{etamatrix}). 
The noise variance in the decoupled observation model, $\Sigmam_t$ is provided at each iteration $t$
by the following recursive equation termed  
{\em State Evolution} (SE),
\beq
\Sigmam_{t+1} = \sigma^2\Id_M + \frac{K_\text{tot}}{\dimPilots}\EE[\epsilonf_{t} \epsilonf_t^\herm]
\label{eq:vamp_SE}
\eeq
where
\begin{equation} 
    \epsilonf_t = 
    \left \{ \begin{array}{ll}
            (\widetilde{\eta}_t(\xv + \zv, g_k) - \xv)^\transp &   \;\;\; \mbox{if $\gv$ is known} \\
            (\widetilde{\eta}_t(\xv + \zv) - \xv)^\transp         &   \;\;\; \mbox{if $\gv$ is not known} 
        \end{array} 
    \right .
    \label{residual-def}
\end{equation}
The initial value of the SE is given by $\Sigmam_0 = \sigma^2\Id_M + \frac{K_\text{tot}}{\dimPilots}\EE[\xv\xv^\herm]$. 
The sequence $(\Sigmam_t)_{t=0,1,2,...}$ does not depend on a specific input $\Xm$ and can be precomputed.
The SE equation has the important property that it predicts the estimation error of the AMP output 
$\{\Xm^t\}_{t=0,1,...}$ asymptotically in the sense that in the limit of $K_\text{tot},L \to\infty$ with
$L/K_\text{tot} = \text{const.}$ it holds that \cite{Jav2013}
\beq
\lim_{K_\text{tot}\to\infty}\frac{\|\Xm^{t+1}-\Xm\|_F^2}{K_\text{tot}} = \trace(\EE[\epsilonf_t\epsilonf_t^\herm])
= \trace(\Sigmam_t -\sigma^2\Id_M)\frac{L}{K_\text{tot}}.
\eeq
Formally this was proven for the case when the entries of $\Am$ are Gaussian iid. In practice
this property holds also when the columns of $\Am$ are sampled uniformly from the sphere,
as in our case.
Note, that $\trace(\EE[\epsilonf_t\epsilonf_t^\herm])$ is the MSE of the estimator
$\widetilde{\eta}$ in the Gaussian vector channel \eqref{eq:vamp_decoupled} and therefore the choice
\eqref{PME-cond} (or \eqref{PME-uncond} resp.) is asymptotically optimal as it minimizes the MSE in each
iteration.

Since there is no spatial correlation between the receive antennas, 
$\Sigmam_0$ is diagonal and it can be shown (see \cite{Che2018}) 
that $\Sigmam_t$ is diagonal for all $t$. In the case of $\gv$ is known to the AD estimator,
a simple calculation yields
the function $\widetilde{\eta}_{t,k}(\rv)$ defined in (\ref{PME-cond}) in the form
\beq
\widetilde{\eta}_{t,k}(\rv) = \phi_{t,k}(\rv)g_k(g_k\Id_M + \Sigmam_t)^{-1}\rv,
\eeq
where the coefficient $\phi_{t,k}(\rv)\in[0,1]$ is the posterior mean estimate of
the $k$-th component $b_k$ of the activity pattern $\bv$,
when rewriting the decoupled observation model (\ref{eq:vamp_decoupled}) as
$\rv = \sqrt{g_k} b_k \hv + \zv$. In particular,
we have (details are omitted and can be found in \cite{Che2018})
\begin{align}
    \phi_{t,k}(\rv) & = \EE[b_k|\rv,g_k] \nonumber \\
    &= p(b_k=1|\rv,g_k) \nonumber \\
    &= \left\{1 + \frac{1-\lambda}{\lambda}  \prod_{i=1}^M \left [ \frac{g_k+\tau_{t,i}^2}{\tau^2_{t,i}} \exp\left ( - \frac{g_k |r_i|^2}{\tau_{t,i}^2(g_k + \tau^2_{t,i})}\right)\right ] \right\}^{-1}
    \label{eq:vamp_pme_b}
\end{align}
The term $\langle\eta^\prime(\cdot)\rangle$ in \eqref{eq:vamp_2} is defined as
\beq
\langle\eta_t^\prime(\Rm)\rangle 
= \frac{1}{K_\text{tot}}\sum_{k=1}^{K_\text{tot}} \eta_{t,k}^\prime(\Rm_{k,:}), 
\eeq
where $\eta_{t,k}^\prime(\cdot) \in \CC^{M\times M}$
is the  Jacobi matrix of the function $\eta_{t,k}(\cdot)$ evaluated at the $k$-th row $\Rm_{k,:}$ of the matrix argument $\Rm$.
For known LSFCs and uncorrelated antennas (yielding diagonal $\Sigmam_t = \text{diag}(\tau^2_{t,1},...,\tau^2_{t,M})$ for 
all $t$), the derivative is explicitly given by 
\begin{align}
\eta_{t,k}^\prime(\rv) &= \phi_{t,k}(\rv) \text{diag}(\Xim_{t,k}) + (\Xim_{t,k} \rv) (\widetilde{\Xim}_{t,k} \rv)^\herm 
(\phi_{t,k}(\rv) - \phi_{t,k}(\rv)^2)
\label{eq:d_eta}
\end{align}
where we define 
$\Xim_{t,k} = \diag\left ( \frac{g_k}{g_k + \tau^2_{t,i}} : i \in [M] \right )$ and 
$\widetilde{\Xim}_{t,k} = \diag\left ( \frac{g_k}{\tau_{t,i}^2(g_k + \tau^2_{t,i})} : i \in [M] \right )$.
Analogous expressions for the case where the LSFCs $\gv$ are unknown to the receiver can be found, but their expression cannot be 
generally given in a compact form and in general depends on the LSFC distribution $p_G(\cdot)$ (see \cite{Che2018} for more details).

\subsubsection{MMV-AMP Scaling}
\label{sec:vamp_scaling}

For the single measurement vector (SMV) case ($M = 1$) it was shown in \cite{Bay2011} that in the asymptotic limit
$\dimPilots,K_\text{tot},K_a \to \infty$
with fixed ratios $\dimPilots/K_\text{tot}$ and $K_a/K_\text{tot}$ the estimate
$\Am^\herm \zv^t + \xv^t$ in the AMP algorithm 
in the $t$-th iteration is indeed distributed like the true target signal in Gaussian noise with noise
variance $\Sigma^t$ given by the SE.
A generalized version of this statement that includes the MMV case was proven in \cite{Jav2013}.
It was shown in \cite{liu2017massive} that, based on the state evolution equation \eqref{eq:vamp_SE},
the error of activity detection vanishes in the limit
$M\to\infty$ for any number of active users. It is important to notice that, in this type of SE-based analysis, 
first the limit $K_a,\dimPilots \to \infty$ is taken at fixed $M$ and then the limit $M\to\infty$ is taken. 
This makes it impossible to derive a scaling relation between $M$ and $K_a$. 
Furthermore, this order of taking limits assumes that $K_a$ is much larger then $M$. 
Hence, this type of analysis does not generally describe the case when $M$ scales proportional to
$K_a$ or even a bit faster. Finally, it is implicit in this type of analysis that 
$L$, $K_a$ and $K_{\rm tot}$ are asymptotically in linear relation, i.e., 
$\frac{K_a}{\dimPilots} \to \alpha$ and $\frac{K_\text{tot}}{\dimPilots}\to \beta$ for some $\alpha, \beta \in (0, \infty)$.
Hence again, it is impossible to capture the scaling studied in our work, where $K_a$ is essentially quadratic in $L$, 
$K_{\rm tot}$ can be much larger than $K_a$, and $M$ scales to infinity slightly faster than $K_a$.  

The above observation is a possible explanation for the behavior described in Section \ref{sec:instability},
which is in fact quite different from what is predicted by the SE and in fact reveals an
annoying non-convergent behavior of MMV-AMP when $M$ is large with respect to $L$ and the dimensions are 
or ``practical interest'', i.e., not extremely large. 

\subsubsection{Approximations}
\label{sec:vamp_approximations}

Instead of pre-computing the sequence $(\Sigmam_t)_{t=0,1,...}$, 
in the SMV case, where $\Sigmam_t$ reduces to a single parameter $\tau_t^2$,
it is common to use the norm of the residual $\|\Zm_t\|_2^2/K_\text{tot}$ as an empirical estimate
of $\Sigmam_t$ \cite{Ran2011,Bar2017a}, since it leads to faster convergence \cite{Gre2018}
while disposing the need of pre-computing the state evolution recursion. 
We find empirically that, analogous to the SMV case,
estimating the $i$-th diagonal entry of
$\Sigmam_t = \text{diag}(\tau_{t,1}^2,...,\tau_{t,M}^2)$ as $\tau^2_{t,i} = \|\Zm^t_{:,i}\|_2^2/K_\text{tot}$ (i.e., 
the empirical variance of the $i$-th column of the matrix $\Zm^t$ in \eqref{eq:vamp_2}) 
leads to a good performance.

Another possible approximation arises from the observation that in the derivative
\eqref{eq:d_eta}, the diagonal terms are typically much larger then the off-diagonal terms,
which is to be expected, since in expectation the off-diagonal entries of
the term $(\Xi_{t,k}\rv)(\tilde{\Xi}_{t,k}\rv)^\herm$ vanish. 
So we find empirically that reducing the calculation of the derivative to just the 
diagonal entries, barely alters the performance in a large parameter regime,
while significantly reducing the complexity of the MMV-AMP iterations from
$\mathcal{O}(M^2)$ to $\mathcal{O}(M)$.

\subsubsection{Activity detection with MMV-AMP}

For known LSFCs an estimate of the activity pattern can be obtained directly by thresholding
the posterior mean estimate of $b_k$ \eqref{eq:vamp_pme_b}. For statistically known LSFCs
we have to calculate the integral of \eqref{eq:vamp_pme_b} over the distribution of the LSFCs.
For large $M$ this integral may become numerically unstable, in that case we can also
use the following method: Let $\Xm^{t_0}$ and $\Zm^{t_0}$ denote the output of the MMV-AMP algorithm
at the final iteration. Let $\Rm^{t_0} :=\Am^\herm \Zm^{t_0} +\Xm^{t_0}$. Under the assumption that
the asymptotic decoupling phenomenon described in Section \ref{sec:vamp_scaling} holds,
i.e. that the decoupled observation model
represents faithfully the statistics of the rows of $\Rm^{t_0}$, each row $\Rm^{t_0}_{k,:}$ is distributed
as $\sqrt{\gamma_k} \hv_k + \zv_k$ with $\zv_k \sim \mathcal{CN}(0,\Sigmam^{t_0})$ and $\hv_k$ has the statistics of the
Gaussian MIMO i.i.d. channel vector of user $k$. 
Furthermore we assume that $\Sigmam^{t_0}$ is diagonal, with entries $\tau^2_{t_0,i} : i = 1, \ldots, M$, which are
estimated as described in the previous section. 
Then the ML estimate of $\gamma_k$ from $\Rm^{t_0}$ is given by
\beq
\widehat{\gamma}_k = \max \left(0,\frac{\|\Rm_{k,:}^{t_0} \|_2^2}{M} - \frac{\sum_{i=1}^M \tau^2_{t_0,i}}{M}\right).
\eeq
Then, the activity pattern as well as the active LSFC pattern can be obtained by thresholding the $\widehat{\gamma}_k$.

\subsubsection{Instability of MMV-AMP}
\label{sec:instability}

In simulations, we have observed that the MMV-AMP algorithm as described in
section \ref{sec:vamp},
for certain parameter settings, exhibits an annoying non-convergent behavior that occurs at random with some non-negligible probability
(according to the realization of the random pilot matrix $\Am$, the random channel matrix $\Hm$, and the random observation noise).
We find that this behavior occurs most frequently for either small $K_a<<L$ and $M$ similar to or larger then $K_a$, or for 
$M>K_a>L$. Also the dynamic range of the LSFCs plays an important role.
While this behavior occurs less frequently or completely vanishes for a small dynamic range or constant LSFCs, 
it occurs more frequently for large dynamic ranges.
For example if we
let $g_k$ be distributed uniformly in dB scale between 0 and 20dB,
known at the receiver, for
$K_a = 20$ the algorithm is stable for $M=4$, in the sense that the effective
noise variance $\tau_t^2$ 
decreases consistently, but unstable for $M=10$, i.e. for many instances the actual measured
values of $\|\Xm^t-\Xm\|^2_F/(MK_\text{tot})$ diverge a lot from their SE prediction \eqref{eq:vamp_SE}. This behavior is illustrated
in \figref{fig:amp_divergence}, where $\|\Xm^t-\Xm\|^2_F/(MK_\text{tot})$ is plotted for $t=1,2,...$ for several samples
along with $\tau_t^2/M$, where $\Sigmam_t = \tau_t^2\Id_M$ is calculated according to the SE \eqref{eq:vamp_SE}.
For $K_a<L$ one may argue that this is an artificial behavior,
which can be circumvented by simply discarding
the information from some of the antennas,
but this is certainly not possible for $K_a>L$, where $M>K_a$ measurements
are necessary. We find that specifically in this regime $M>K_a>L$ the MMV-AMP performance differs significantly from its
state evolution prediction, which is consistent with what was argued in section \ref{sec:vamp_scaling}.
These outliers occur even if none of the approximations mentioned in section
\ref{sec:vamp_approximations} are applied. Although we find that approximating the derivative $\eta'(\cdot)$ as described in section
\ref{sec:vamp_approximations} helps to reduce the number of samples that do not converge to the state evolution prediction.
Another observation is that the use of normalized pilots ($\|\av_k\|_2^2 = L$)
improves the convergence to the SE prediction compared to Gaussian iid pilots.

\begin{figure}
   \centering
   \subfloat[$K_a = 20, M=4$]{\includegraphics[width=0.4\linewidth]{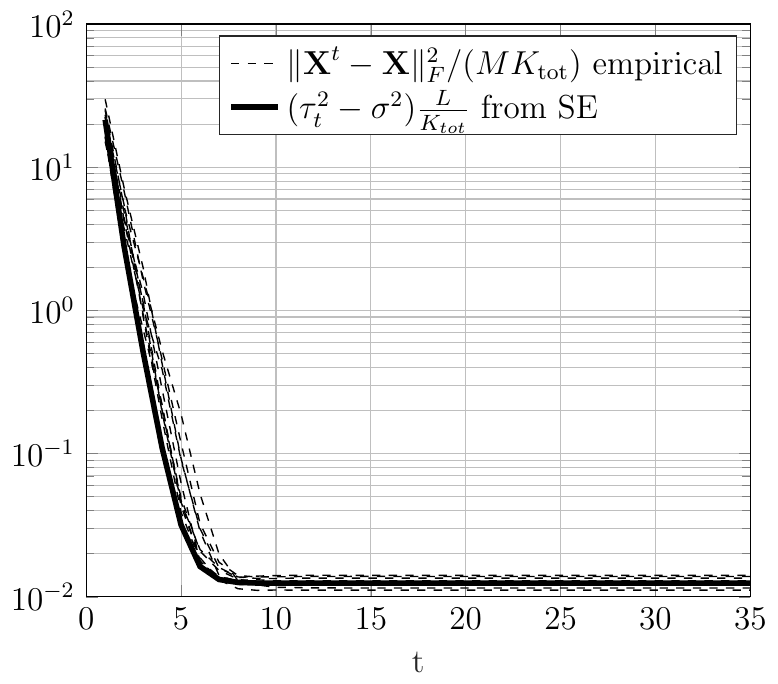}}\quad
   \subfloat[$K_a=20,M=10$]{\includegraphics[width=0.4\linewidth]{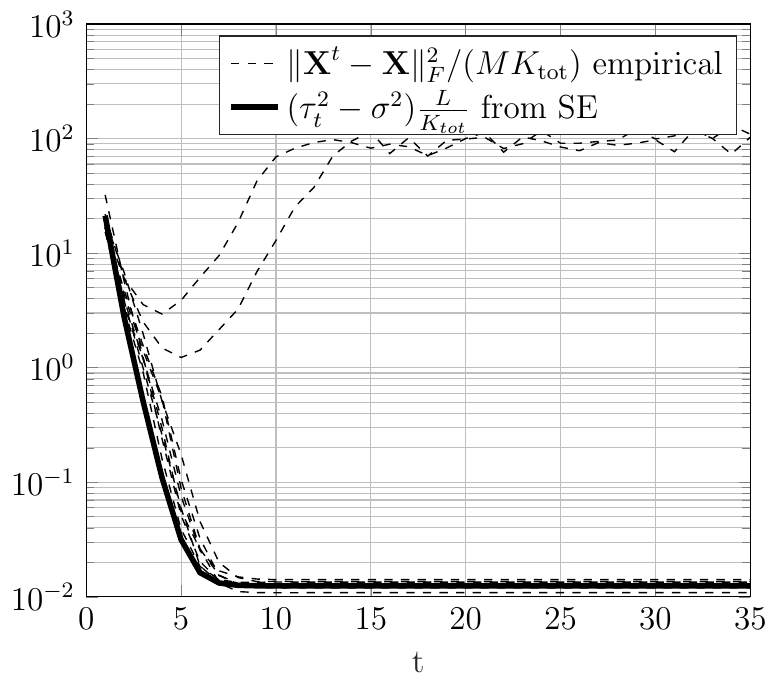}}\quad
   \subfloat[$K_a=120,M=150$]{\includegraphics[width=0.4\linewidth]{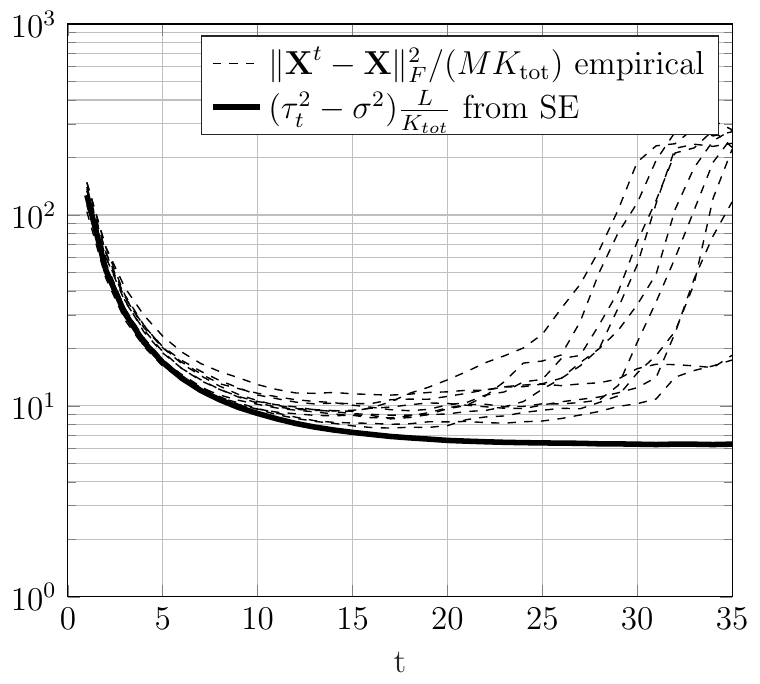}}
   \caption[Titel des Bildes]{Evolution of the normalized MSE in the AMP iterations
       \eqref{eq:vamp_1}-\eqref{eq:vamp_2} for 10 sample
       runs and its state evolution prediction from \eqref{eq:vamp_SE}.
       $L=100$, $K_\text{tot}=2000$ and the LSFCs are chosen such that $\text{snr}_k$ (see \eqref{snr_def})
       are uniformly distributed between 0 and 20dB
   and are assumed to be known at the receiver.}
   \label{fig:amp_divergence}
\end{figure}

\subsection{Complexity Comparison}

The complexity of the discussed covariance-based AD algorithms (ML and NNLS) scales with the
size of the covariance matrix and the total number of users, i.e. 
$\mathcal{O}(K_\text{tot}\dimPilots^2)$,
plus the complexity of once calculating the empirical covariance matrix which is linear in $M\dimPilots$.

The complexity of MMV-AMP in each iteration
scales like $\mathcal{O}(M^2 \dimPilots K_\text{tot})$ or, with a sub-sampled FFT matrix as pilot matrix,
like $\mathcal{O}(M^2 K_\text{tot}\log K_\text{tot})$. Using the simplified derivative as described in
paragraph \ref{sec:vamp_approximations} the complexity is reduced to
$\min(\mathcal{O}(MK_\text{tot}\log K_\text{tot}),\mathcal{O}(MK_\text{tot}\dimPilots))$.
In any case the covariance-based algorithms scale better with $M$,
while MMV-AMP scales better with $\dimPilots$. 

\subsection{Scaling}
\label{sec:scaling}

The performance of AD is visualized in \figref{fig:scale_with_M_cs} (`CS regime’, i.e. $K_a \leq \dimPilots$)
and \figref{fig:scale_with_M} ($K_a > \dimPilots$).
Here we assumed all the LSFCs to be identically equal to 1, 
MMV-AMP was run with the full knowledge of the LSFCs and the ML and NNLS algorithms were run with the
box-constraints described in Section \ref{sec:box-constraint}.
In \figref{fig:scale_with_M_cs} the NNLS algorithm is comparably worse than MMV-AMP and ML. 
This is to be expected, since $M$ is small compared to $\dimPilots$, which leads to a significant gap between
the true and the empirical covariance
matrix $\|\widehat{\Sigmam}_\yv-\Sigmam_\yv\|_F$. Interestingly, although the ML algorithm is also
covariance based, it still outperforms MMV-AMP. In \figref{fig:scale_with_M} we see that beyond the CS regime,
the performance of MMV-AMP significantly deteriorates, while the activity detection error probability of ML and NNLS still decays 
exponentially with $M$.
In \figref{fig:l1_over_M} we compare the LSFC estimation performance of the 
ML and NNLS algorithms. The simulations confirms Corollary \ref{cor:nnls:p1} and show
that the relative $\ell_1$ recovery error of NNLS indeed decays like $1/\sqrt{M}$. 
We see that the same decay behavior holds for the ML algorithm only with significantly better constants. 
Note, that the number of required antennas for the ML algorithm
scales fundamentally different depending on whether $K_a\leq L$ or
$K_a > L$. In the first case the probability of error decays a lot faster with increasing $M$, matching
qualitatively the
scaling derived in Theorem \ref{thm:ml:support},
which states that (up to constant or logarithmic factors) $M = \mathcal{O}((K_a/L)^2)$.  
\begin{figure}
    \centering
    \includegraphics[width=0.5\textwidth]{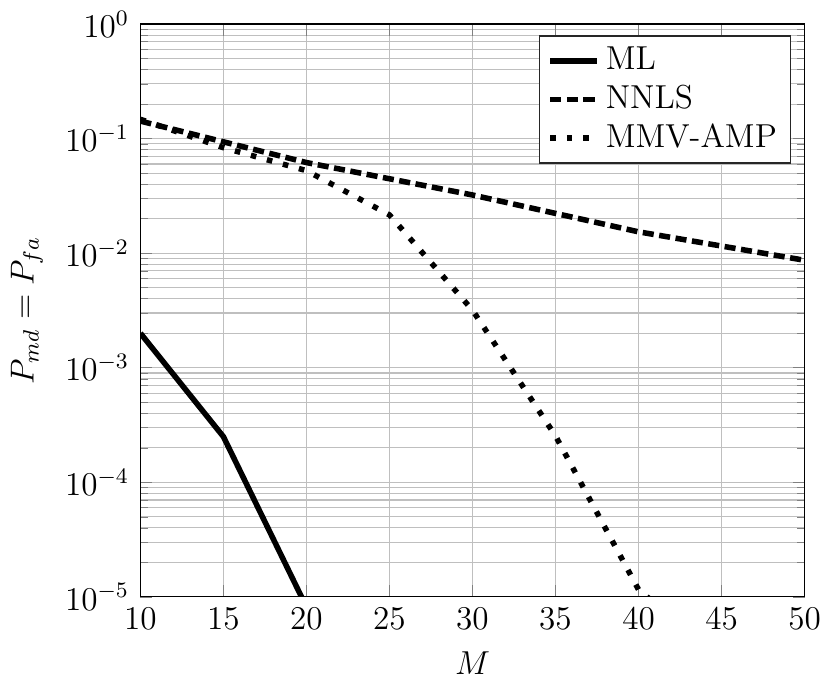}
    \caption{Scaling of the support detection error vs. $M$ at the border of the CS regime
    for $K_a = L = 100, K_\text{tot} = 2000$ with
    constant LSFCs at $\snr_k = 0$ dB.}
    \label{fig:scale_with_M_cs}
\end{figure}

\begin{figure}
    \centering
    \includegraphics[width=0.5\textwidth]{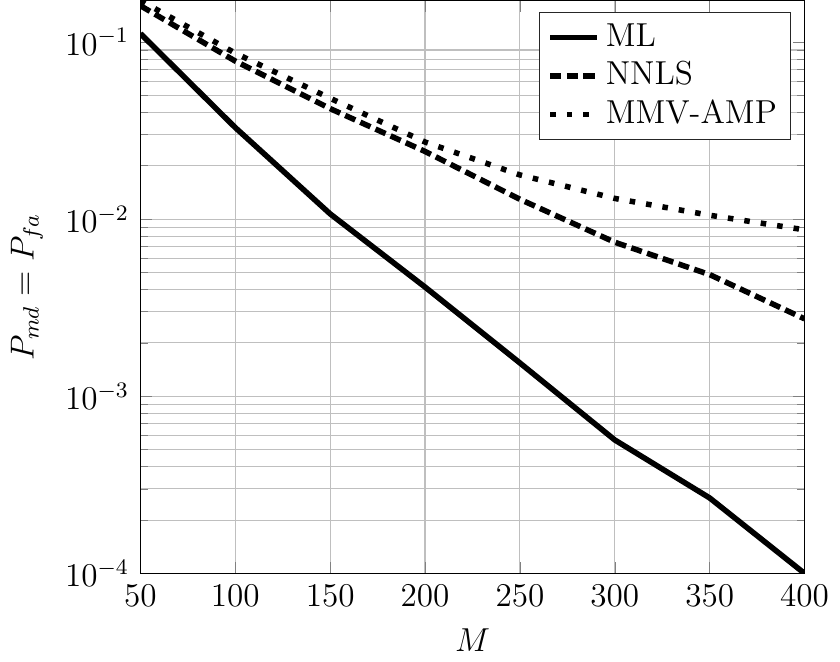}
    \caption{Scaling of the support detection error vs. $M$
    beyond the CS regime (i.e. $K_a > L$). Here $K_a=300, L=100,K_\text{tot} = 2000$ with constant LSFCs at $\snr_k = 0$ dB.}
    \label{fig:scale_with_M}
\end{figure}

\begin{figure}
    \centering
    \includegraphics[width=0.5\textwidth]{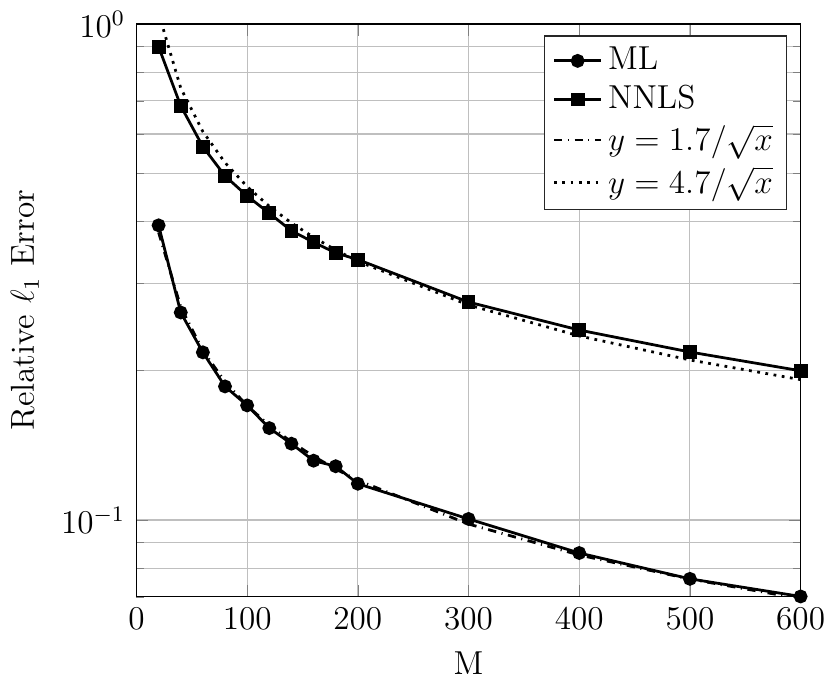}
    \caption{Relative $\ell_1$ error of the estimation of the LSFCs of the active users for
        $D_c = 100,K_a = 200,K_\text{tot} = 2000$.
        The LSFCs are chosen such that $\snr_k$  are uniform in the range 0-20dB. The dotted lines
    show that the curves are well represented by a $c/\sqrt{M}$ behavior, for some constant $c$, as
    predicted by Corollary \ref{cor:nnls}. }
    \label{fig:l1_over_M}
\end{figure}

Corollary \ref{cor:nnls:p1} predicts that, in the limit $M \to \infty$,
the recovery error of NNLS vanishes, as long as the number of active users fulfils
condition \eqref{eq:scaling_K_a}. We confirm this behavior empirically in
\figref{fig:nnls_phase_transition_1000}a, where we solve the NNLS
problem \eqref{eq_nnls} using the true covariance matrix
$\Sigmam^\circ = \Am\diag(\gammam^\circ)\Am^\herm + \sigma^2\Id_{\dimPilots}$ instead of the empirical 
covariance matrix $\widehat{\Sigmam}_\yv$. In this case,  
$\|\dv\|_2 = 0$ in \eqref{eq:short:thm:mse} and the recovery error should be identically zero
when the true vector $\gammav^\circ$ is $K_a$-sparse and the system parameters are such that Theorem
\ref{NNLS-theorem} holds.
This is confirmed by \figref{fig:nnls_phase_transition_1000}a,
showing a quadratic curve, below which the recovery 
error vanishes. 
\begin{figure}
   \centering
   \subfloat[ML]{\includegraphics[height=6cm]{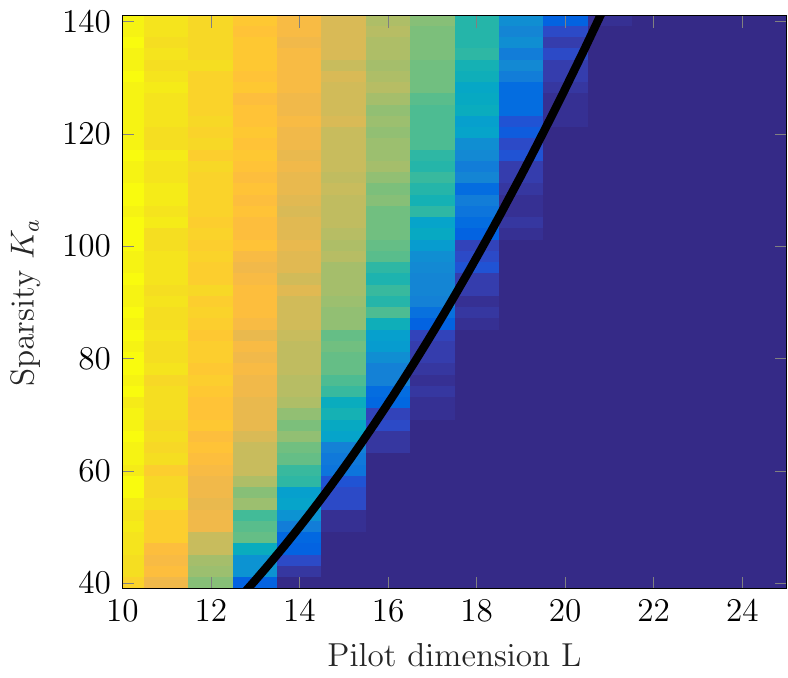}}\quad
   \subfloat[NNLS]{\includegraphics[height=6cm]{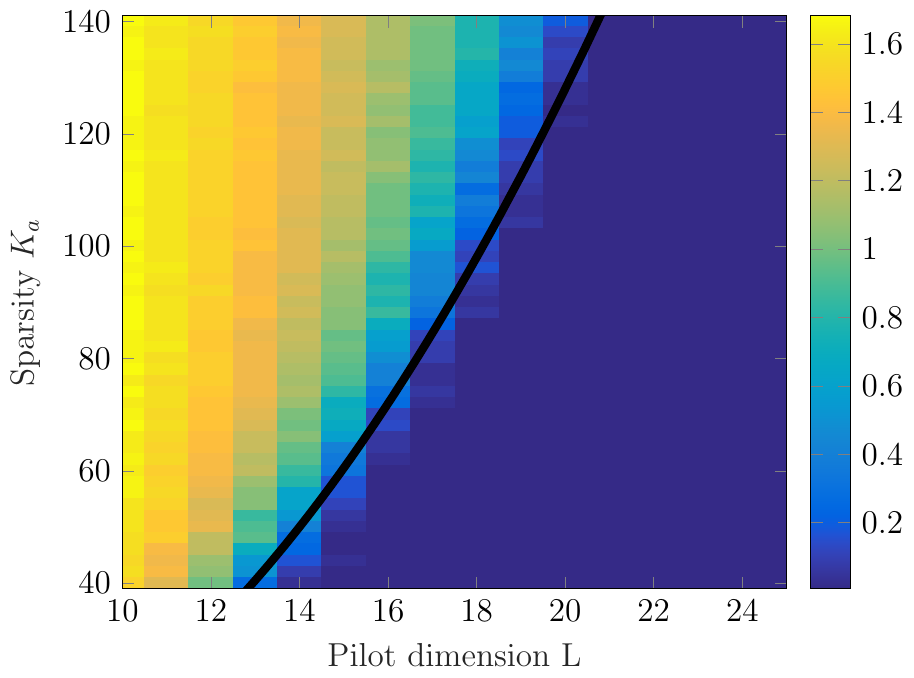}}
  \caption{Phase transition of the recovery error for NNLS and ML
      in the limit $M \to \infty$ for $K_\text{tot}=1000$.
      The function $x\rightarrow (x-4)^2/2$ is overlayed in black to emphasize the super-linear scaling.
      The color indicates the normalized $\ell_1$-error as
    it is subject of Corollary \ref{cor:nnls:p1} in the NNLS case.
    The LSFC are constant and the activity pattern is chosen uniformly at random from 
    all $K_a$-sparse vectors.
    The results are obtained by averaging over random pilot matrices and
    activity patterns. }
    \label{fig:nnls_phase_transition_1000}
\end{figure}
We also observe a very similar behavior for the ML algorithm,
(see \figref{fig:nnls_phase_transition_1000}b).
This suggests that the condition \eqref{eq_nnls} is indeed necessary independent of the algorithm.

\figref{fig:box-constraint} shows the gain in performance when the LSFCs are known at the receiver and
the box-constraint (step \ref{alg:update_step} in Algorithm \ref{tab:ML_coord}) is employed.

\begin{figure}
    \centering
    \includegraphics[width=0.5\linewidth]{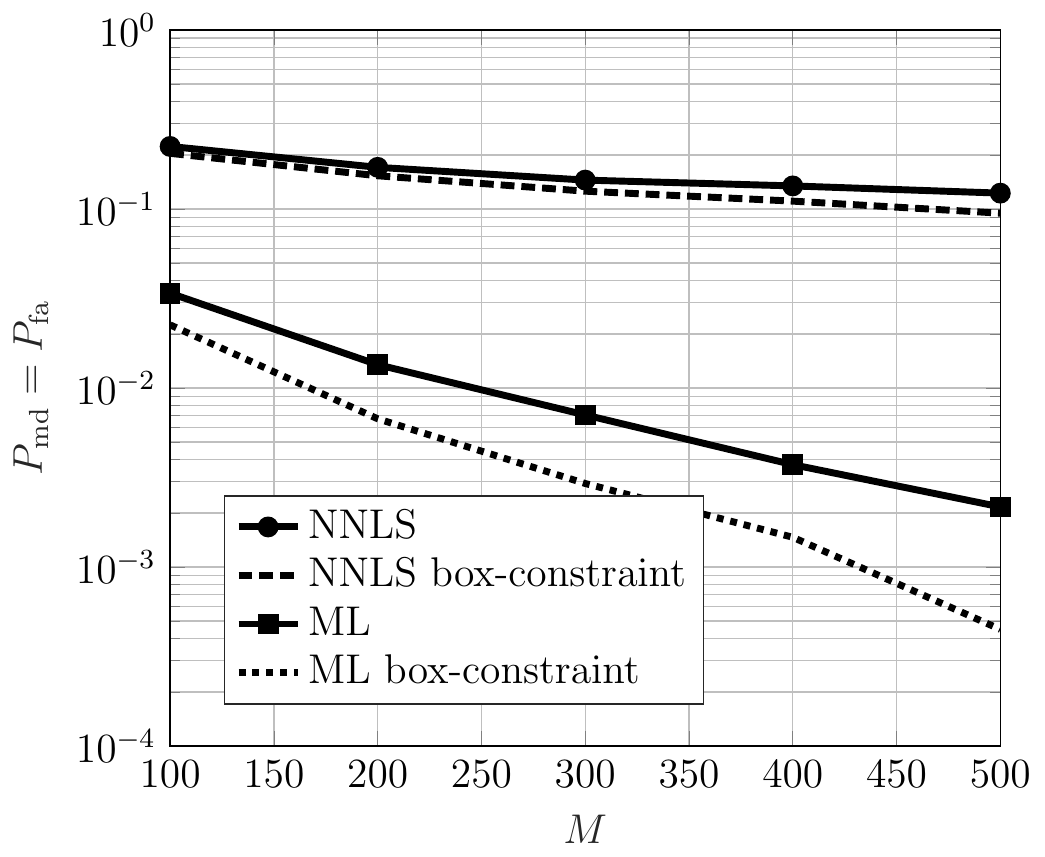}
    \caption{Effect of using the box-constraint (see step \ref{alg:update_step} in Algorithm \ref{tab:ML_coord})
        when the LSFCs $g_k$ are known at the receiver.
    Here $K_a = 150, L=100, K_\text{tot} = 2000$ and the LSFCs are distributed such that $\snr_k$ are
    uniform in the range $0-20$dB.
}
	\label{fig:box-constraint}
\end{figure}

\section{Application: Massive MIMO Unsourced Random Access}

As an application of the presented non-Bayesian algorithms and their analysis, in this section
we introduce an extension of the recently posed unsourced random access problem \cite{Pol2017}
to the case of a massive MIMO BS receiver and show that the ML scheme
(see Algorithm \ref{tab:ML_coord}) provides an efficient low-complexity approach.  
The presented scaling properties in Corollary \ref{cor:nnls_scaling} enable us to estimate 
the required per-user-power,  in terms of $E_b/N_0$, and the required number of receive antennas $M$ 
for reliable transmission.

The channel model is the same as described in Section \ref{sec:signal-model}, i.e., a block-fading channel with 
blocks of $L$ signal dimensions over which the user channel vectors are constant.
We assume $n = S L$, for some integer $S$, such that the transmission of a codeword spans
$S$ fading blocks.
Following the problem formulation in \cite{Pol2017}, each user is given the same codebook
$\Cc = \{ \cv(m) : m \in [2^{nR}]\}$, formed by 
$2^{nR}$ codewords $\cv(m) \in \CC^n$. A fixed but unknown number $K_a$ of users transmit their messages
over the coherence block.
\footnote{Here, as in 
\cite{Pol2017} and in \cite{Ama2020a}, we assume that users are synchronized.
}
The BS must then produce a list $\Lc$ of the transmitted messages $\{m_k : k\in \Kc_a\}$ (i.e., 
the messages of the active users). 
The system performance is expressed in terms of the {\em Per-User Probability of Misdetection}, defined as the average fraction of 
transmitted messages not contained in the list, i.e., 
\beq
p^\text{msg}_\text{md} = \frac{1}{K_a}\sum_{k \in \mathcal{K}_a} \PP(m_k \notin \mathcal{L}),
\label{eq:ura_pmd}
\eeq
and the {\em Probability of False-Alarm}, defined as the average fraction of decoded messages that 
were indeed not sent, i.e., 
\beq
p^\text{msg}_\text{fa} = \frac{|\mathcal{L}\setminus \{ m_k : k \in \Kc_a \}|}{|\mathcal{L}|}.
\label{eq:ura_pfa}
\eeq
The size of the list is also an outcome of the decoding algorithm, and therefore it is a random variable. 
As customary, the average error probabilities of false-alarm/misdetection are defined as the expected
values of $p^\text{msg}_\text{fa}/p^\text{msg}_\text{md}$ resp. over all involved random variables.
That is in this case the Rayleigh fading coefficients, the AWGN noise
and the choice of messages,
where the messages are assumed to be chosen uniformly and independent of each other.
Notice, that in this problem formulation the total number of users $K_{\rm tot}$ is completely irrelevant,
as long as it is
much larger than the number of active user $K_a$ (e.g., we may consider $K_{\rm tot} = \infty$). 
Letting the average energy per symbol of the codebook $\Cc$ be denoted by $E_s = \frac{1}{n 2^{nR}} \sum_{m=1}^{2^{nR}}  \|\cv(m)\|_2^2$,
the received signal can be re-normalized such that the AWGN per-component variance is $\sigma^2 = N_0/E_s$ and the received energy per code symbol is 1. 
In this way, the notation introduced for the AD model in \eqref{sig_eq1} is preserved. Furthermore, as customary in coded systems, 
we express energy efficiency in terms of the standard quantity $E_b/N_0 :=  \frac{E_s}{R N_0}$.

\subsection{Unsourced random access as AD problem}
\label{sec:ura_as_ad}
For now assume $S=1$, i.e. each user transmits his codeword in a single block of length $L$.
Further fix $J=LR$ and let $\Am \in \CC^{L \times 2^J} = [\av_1,...,\av_{2^J}]$,
be a matrix with columns normalized such that
$\| \av_i \|_2^2 = L$. Each column of $\Am$ represents one codeword.
Let $i_k$ denote the
$J$-bit messages produced by the active users $k \in \Kc_a$, represented as integers in $[1:2^J]$, 
user $k$ simply sends the column $\av_{i_k}$ 
of the coding matrix $\Am$. The received signal at the $M$-antennas BS takes on the form 
\begin{eqnarray} 
\Ym & = & \sum_{k \in \Kc_a} \sqrt{g_k} \av_{i_k} \hv^\transp_{k}  + \Zm \nonumber \\
& = & \Am \Phim \Gm^{1/2} \Hm + \Zm
\label{matrix-channel}
\end{eqnarray}
where, as for the AD model in (\ref{pilot_sig}),  
$\Gm = \diag(g_1, \ldots, g_{K_{\rm tot}})$ is the diagonal matrix of LSFCs,
$\Hm \in \CC^{K_{\rm tot} \times M}$ is the matrix containing, by rows, 
the user channel vectors $\hv_k$ formed by the small-scale fading antenna coefficients (Gaussian i.i.d. entries $\sim \Cc\Nc(0,1)$), 
$\Zm \in \CC^{L \times M}$ is the matrix of AWGN samples (i.i.d. entries $\sim \Cc\Nc(0,\sigma^2)$),
and $\Phim \in \{0,1\}^{2^J \times K_{\rm tot}}$ is a binary selection matrix where
for each $k \in \Kc_a$ the corresponding column $\Phim_{:,k}$ is all-zero but a single one in position $i_k$, 
and for all $k \in \Kc_{\rm tot} \setminus \Kc_a$ the corresponding column $\Phim_{:,k}$ contains all zeros.

Let's focus on the matrix $\Xm =  \Phim \Gm^{1/2} \Hm$ of dimension $2^J \times M$. 
The $r$-th row of such matrix is given by 
\beq
\Xm_{r,:} = \sum_{k\in\mathcal{K}_a}\sqrt{g_k} \phi_{r,k} \hv^\transp_k, 
\eeq
where $\phi_{r,k}$ is the $(r,k)$-th element of  $\Phim$, equal to one if $r = i_k$ and zero otherwise. 
It follows that $\Xm_{r,:}$ is Gaussian with i.i.d. entries $\sim \Cc\Nc\left (0, \sum_{k \in \Kc_a} g_k \phi_{r,k} \right )$. 
Since the messages are uniformly distributed over $[1:2^J]$ and statistically independent across the users,
the probability that $\Xm_{r,:} $ is identically zero is given by $(1 - 2^{-J})^{K_a}$. 
Hence, for $2^J$ significantly larger than $K_a$, the matrix $\Xm$ is row-sparse. 

In order to map the decoding into a problem completely analogous to the AD problem already discussed before, 
with some abuse of notation we define the modified LSFC-activity coefficients 
$\gamma_r := \sum_{k\in\mathcal{K}_a} g_k \phi_{r,k}$ and
$\Gammam = \text{diag}(\gamma_1,...,\gamma_{2^J})$. Then, (\ref{matrix-channel}) can be written as
\beq
\Ym = \Am \Gammam^{1/2} \widetilde{\Hm} + \Zm,
\label{eq:matrix_model2}
\eeq
where $\widetilde{\Hm} \in \CC^{2^J \times M}$ with i.i.d. elements $\sim \Cc\Nc(0,1)$.  
Notice that in (\ref{eq:matrix_model2}) the number of total users $K_{\rm tot}$ plays no role. 
In fact, none of the matrices involved in (\ref{eq:matrix_model2}) depends on $K_{\rm tot}$. 

The task of the inner decoder at the BS is to identify the non-zero elements of the modified active LSFC  pattern $\gammav$, the vector of 
diagonal coefficients of $\Gammam$. 
The active (non-zero) elements correspond to the indices of the transmitted messages. 
Notice that even if two or more users choose the same sub-message, the corresponding modified LSFC
$\gamma_r$ is positive since it corresponds to the sum of the signal powers.
In other words, since the detection scheme is completely 
non-coherent (it never explicitly estimates the complex channel matrix) and active signals add in power, there is no risk of 
signal cancellation or destructive interference. 

At this point, it is clear that the problem of identifying the set of transmitted messages
from observation (\ref{eq:matrix_model2})
is completely analogous to the AD problem from the observation in (\ref{pilot_sig}),
where the role of the total number of users $K_{\rm tot}$ in 
the AD problem is replaced by the number of messages $2^J$ in the inner decoding problem. 
Building on this analogy, we shall use the discussed ML algorithm to decode the inner code. 

It is interesting to notice that the modified LSFCs in $\gammav$ are random sums of the individual user channel 
gains $\{g_k\}$.  Hence, even if the $g_k$'s were exactly individually known, or their statistics was known, 
these random sums would have unknown values and unknown statistics (unless averaging over all possible
active subsets, which would involve an exponential complexity in $K_{\rm tot}$ which is clearly infeasible in our context).
Hence, Bayesian approaches such as MMV-AMP (see Section \ref{sec:vamp}) as advocated in 
\cite{liu2017massive,Che2018,Lar2012,Sen2017} do not find a straightforward application here. 
In contrast, the proposed non-Bayesian approaches (in particular, the ML algorithm in Algorithm \ref{tab:ML_coord}),
that treats  $\gammav$ as a deterministically unknown vector.  

Notice also that in a practical unsourced random access scenario such as a large-scale IoT application, 
the slot dimension $L$ may be of the order of 100 to 200 symbols, while
for a city-wide IoT data collector it is not unreasonable to have $M$ of the order of 500 to 1000 
antennas (especially when considering narrowband signals such as in LoRA-type applications
\cite{Cen2016,Ban2016}). 
This is precisely the regime where we have observed a critical behavior of MMV-AMP,
while our algorithm uniformly improves as $M$ increases, for any slot dimension $L$. 

\subsection{Discussion and analysis}
\label{sec:analysis}

In this section we discuss the performance of the ML decoder in a single slot ($S=1$).
For the sake of simplicity, in the discussion of this section we assume $g_k = 1$ for all $k$.
In this case, the SNR $E_s/N_0$ is also the SNR at the receiver, for each individual (active) user. 

Corollary \ref{cor:nnls:p1} shows that, if the coding matrix $\Am$ is chosen randomly,
the probability of an error
in the estimation of the support of $\gammav$ vanishes in the limit $M \to \infty$ for any SNR $\frac{E_s}{N_0} > 0$
as long as
$K_a = \mathcal{O}(L^2/\log^2(e2^J/L^2))$. 
Then Corollary \ref{cor:nnls:p1} gives the following bound for the reconstruction error of 
\begin{align}
    \frac{\lVert \gammav - \gammav^*\rVert_1}{\|\gammav\|_1 } \leq
    \kappa   \left(1+\left(K_a\frac{E_s}{N_0}\right)^{-1}\right)\sqrt{\frac{K_a}{M}}   
\label{eq:gamma_perf}
\end{align}
where $\kappa$ is some universal constant and $\gammav^*$ denotes the estimate of
$\gammav$ by the NNLS algorithm (see section \ref{sec:nnls}).
Our numerical results (section \ref{sec:scaling}) suggest that the reconstruction error of the ML algorithm
is at least as good as that of NNLS (in practice it is typically {\em much better}). 
This bound is indeed very conservative.
Nevertheless, this is enough to give achievable scaling laws for the probability of
error of the inner decoder. 
It follows from \eqref{eq:gamma_perf} that $\frac{\lVert \gammav - \gammav^*\rVert_1}{\|\gammav\|_1 }\to 0$
for $(M,K_a, \frac{E_s}{N_0}) \to (\infty, \infty, 0)$ as long as 
\beq
\frac{K_a( 1 + (K_aE_s/N_0)^{-1})^2}{M} = o(1),
\label{eq:scaling}
\eeq
which is satisfied if $M$ grows as
\beq
M = \max(K_a,(E_s/N_0)^{-1})^{\kappa}
\label{eq:M_lower}  
\eeq
for some $\kappa >1$.
Assuming that $J$ scales such that $2^J = \delta L^2$ for some fixed $\delta\geq 1$,
i.e. $J=\mathcal{O}(\log L)$, then
the condition in Corollary \ref{cor:nnls:p1} becomes 
$K_a = \mathcal{O}(L^2)$ and we can conclude that the recovery error vanishes for sum spectral efficiencies
up to
\beq
\frac{K_aJ}{L} = \mathcal{O}(L\log L).
\label{eq:spectral_eff_S=1}
\eeq
This shows that we can achieve a total spectral efficiency that grows without bound, by encoding
over larger and larger blocks of dimension $L$, as long as the number of messages per user
and the number of active users both grow
proportionally to $L^2$ and the number of BS antennas scales as in (\ref{eq:M_lower}).
The achievable sum spectral efficiency
grows as $L \log(L)$ and the error probability can be made as small as desired, for any given $E_b/N_0 > 0$. 
Of course, in this regime the rate per active user vanishes as $\log(L)/L$. 

We wish to stress again that this system is completely non-coherent,
i.e., there is no attempt to either explicitly (via pilot symbols) or
implicitly to estimate the channel matrix (small-scale fading coefficients).
\subsection{Reducing complexity via concatenated coding}
\label{sec:coding_concatenated}
In practice it is not feasible to transmit even small messages (e.g. $J\sim100$) within one coherence block ($S=1$),
because the number of columns of the coding matrix $\Am$ grows exponentially in $J$. 
Aside from
the computational complexity $L$ may also be limited physically by the coherence time of the channel.
In both cases it is necessary to transmit the message over multiple blocks.
Let each user transmit his message over a \emph{frame} of $S$ fading blocks and within each block use the code
described in section \ref{sec:ura_as_ad} as \emph{inner} code with the ML decoder as inner decoder.

We follow the concatenated coding scheme approach of \cite{Ama2020a},
suitably adapted to our case.  Let $B = nR$ denote the number of bits per user message. 
For some suitable integers $S \geq 1$ and $J > 0$, we divide the $B$-bit message into blocks of size 
$b_1, b_2, \ldots, b_S$ such that $\sum_s b_s = B$ and such that
$b_1 =  J$ and $b_s < J$ for all $s = 2, \ldots, S$. 
Each subblock $s = 2, 3, \ldots, S$ is augmented to size $J$ by appending $p_s = J - b_s$ parity bits,  
obtained using pseudo-random linear combinations of the information bits of the previous blocks 
$s' < s$. Therefore, there is a one-to-one association between the set of all sequences of coded blocks and
the paths of a tree of depth $S$. The pseudo-random parity-check equations generating the parity bits 
are identical for all users, i.e., each user makes use exactly of the same outer {\em tree code}.
For more details on the outer coding scheme, please refer to  \cite{Ama2020a}. 

Given $J$ and the slot length $L$, the inner code is used to transmit in sequence the 
$S$ (outer-encoded) blocks forming a frame. Let $\Am$ be the coding matrix as defined in section \ref{sec:ura_as_ad}.
Each column of $\Am$ now represents one inner codeword.
Letting $i_k(1), \ldots, i_k(S)$ denote the sequence of $S$ (outer-)encoded 
$J$-bit messages produced by the outer encoder of active user $k \in \Kc_a$. The user $k$ now simply sends in
sequence, over consecutive slots of length $L$, the columns $\av_{i_k(1)},\av_{i_k(2)},...,\av_{i_k(S)}$
of the coding matrix $\Am$. As described in section \ref{sec:ura_as_ad}, the inner decoding problem is equivalent to
the AD problem \eqref{eq:matrix_model2}.
For each subslot $s$, let $\widehat{\gammav}[s] = (\widehat{\gamma}_1[s], \ldots, \widehat{\gamma}_{2^J}[s])^\trasp$ 
denote the ML estimate of $\gammav$ in subslot $s$ obtained by the inner decoder. 
Then, the list of active messages at subslot $s$ is defined as
\begin{equation}
\label{support-detection} 
\Sc_s = \left \{ r \in [2^J] : \widehat{\gamma}_r[s] \geq \nu_s \right \}, 
\end{equation}
where $\nu_1, \ldots, \nu_S$ are suitable pre-defined thresholds.  
Let $\Sc_1, \Sc_2, \ldots, \Sc_S$ the sequence of lists of active subblock 
messages. Since the subblocks contain parity bits with parity profile $\{0,p_2, \ldots, p_S\}$, 
not all message sequences in $\Sc_1 \times \Sc_2 \times \cdots \times \Sc_S$ are possible. 
The role of the outer decoder is to identify all possible message sequences, i.e., those corresponding to
paths in the tree of the outer tree code \cite{Ama2020a}.  
The output list $\Lc$ is initialized as an empty list. Starting from $s=1$ and proceeding in order, the decoder converts the integer indices 
$\Sc_s$ back to their binary representation, separates data and parity bits, computes the parity 
checks for all the combinations with messages from the list $\Lc$ and extends only the paths 
in the tree which fulfill the parity checks.  A precise analysis of the error probability of such a decoder
and its complexity in terms of surviving paths in the list is given in \cite{Ama2020a}.
The performance of the concatenated system
is demonstrated via simulations in the following section.

\subsection{Asymptotic analysis - Outer code}
\label{sec:analysis2}

We define the support $\rhov[s]$ of the estimated $\widehat{\gammav}[s]$ as a binary vector
whose $r$-th element is equal to 1 if  $\widehat{\gamma}_r[s] \geq \nu_s$ and to zero otherwise. 
In the case of error-free support recovery, 
$\rhov[s]$  can be interpreted as the output of a vector ``OR'' multiple access channel (OR-MAC) where the inputs are the binary columns of the activity matrix 
$\Phim[s]$ and the output is given by 
\beq
\rhov[s] =  \bigvee_{k \in \Kc_a} \Phim_{:,k}[s], 
\eeq
where $\bigvee$ denotes the component-wise binary OR operation.  
The logical ``OR'' arises from the fact that if the same sub-message is selected by 
multiple users, it will show up as ``active'' at the output of the ``activity-detection'' inner decoder since the signal energy adds up
(as discussed before). 
Classical code constructions for the OR-MAC, like \cite{Kau1964,Dya1983},
have been focussed on zero-error decoding, 
which does not allow for positive per-user-rates as $K_a\to\infty$,
see e.g. \cite{Gyo2008} for a recent survey. 
Capacity bounds for the OR-MAC under the given input constraint have been derived in \cite{Cha1981}
and \cite{Gra1997},
where it was called the ``T-user M-frequency noiseless MAC without
intensity information"" or ``A-channel"". An asynchronous version of this channel was studied in
\cite{Han1996}. Note, that the capacity bounds in the literature are combinatorial and hard to evaluate
numerically for large numbers of $K_a$ and $2^J$.
In the following we will show that, in the typical case of $K_a \ll 2^J$,
a simple upper bound on the achievable rates based on the componentwise entropy is already tight
because it is achievable by the outer code of \cite{Ama2020a}. 
\subsubsection{Achievability}
\label{sec:analysis2:achieve}
The analysis in \cite{Ama2020a} shows that the error probability of the outer code goes to zero
in the so called logarithmic regime with constant outer rate, i.e. for $K_a,J\to \infty$
as $J = \alpha \log_2 K_a$ and $B = SR_\text{out}J$
\footnote{We deviate slightly from the notation in \cite{Ama2020a}, where the scaling parameter
$\alpha'$ is defined by
$B = \alpha'\log_2 K_a$ and the number of subslots is considered to be constant. It is apparent that
those definitions are connected by $\alpha' = SR_\text{out}\alpha$.}
if the number of parity bits $P$ is chosen as (\cite[Theorem 5 and 6]{Ama2020a})
\begin{enumerate}
    \item $P = (S+\delta - 1)\log_2 K_a$ for some constant $\delta>0$ if all the parity bits
        are allocated in the last slots.
    \item $P = c(S-1)\log_2 K_a$ for some constant $c>1$ if the parity bits are allocated evenly
        at the end of each subslot except for the first.
\end{enumerate}
In the first case the complexity scales like $\mathcal{O}(K_a^{R_\text{out}S}\log K_a)$,
since there is no pruning in the first $R_\text{out}S$ subslots,
while in the second case the complexity scales linearly with $S$
like $\mathcal{O}(SK_a\log K_a)$. The corresponding outer rates are
\beq
\begin{split}
    R_\text{out} &= B/(B+P)\\
                 &= 1 - P/(B+P)\\
                 &= 1 - P/(SJ)\\
                 &= 1 - \frac{S+\delta - 1}{S\alpha}\\
                 &= 1 - \frac{1}{\alpha}  + \frac{1}{S}\frac{\delta-1}{\alpha}
\end{split}
\eeq
for the case of all parity bits in the last sections and
\beq
\begin{split}
R_\text{out} &= 1 - \frac{c(S-1)}{S\alpha}\\
             &= 1 - \frac{c}{\alpha}  - \frac{c}{S\alpha}
\end{split}
\eeq
for the case of equally distributed parity bits. In the limit $S\to\infty$ the achievable rates are
therefore $R_\text{out} = 1 - 1/\alpha$ and $R_\text{out} = 1 - c/\alpha$ respectively.

\subsubsection{Converse}
The output entropy of the vector OR-MAC of dimension $2^J$ is bounded by the entropy of 
$2^J$ scalar OR-MACs. The marginal distribution of the entries of $\rhov[s]$  is Bernoulli with 
$\PP(\rho_r[s] = 0) = (1 - 2^{-J})^{K_a}$. Hence, we have
\begin{equation} \label{outputH}
H(\rho[s]) \leq 2^J \Hc_2 ( (1 - 2^{-J})^{K_a}). 
\end{equation}
We stay in the logarithmic scaling regime, introduced in the previous sections,
i.e. we fix $J = \alpha \log_2 K_a$ for some $\alpha >1$ and consider the limit $K_a,J\to\infty$.
In this regime
$K_a/2^J = K_a^{-(\alpha - 1)}\to 0$ and we have
$1 - (1-2^{-J})^{K_a} = K_a/2^J + \mathcal{O}((K_a/2^J)^2) \to 0$.
This gives that
\beq
2^J\Hc_2( (1 - 2^{-J})^{K_a}) \to K_a(J-\log_2 K_a) = (\alpha - 1)K_a\log_2 K_a.
\eeq
Since all users make use of the same code
we have that the number of information bits sent by the $K_a$ active users over a slot is  
$B_\text{sum} = K_a J R_\text{out}$.
Therefore, in order to hope for small probability of error a necessary condition is
\beq
    K_a J R_\text{out} \leq 2^J\mathcal{H}_2((1-1/2^J)^{K_a}).
    \label{eq:sumrate}
\eeq
So the outer rate is limited by 
\beq
    R_\text{out} \leq (\alpha-1)\frac{\log_2 K_a}{J} = 1 - \frac{1}{\alpha}.
\eeq
We have shown in the previous Section \ref{sec:analysis2:achieve} that this outer rate can be achieved in the limit of infinite subslots
$S \to \infty$ by the described outer tree
code at the cost of a decoding complexity of at least $\mathcal{O}(K_a^{R_\text{out}S})$
or up to a constant factor
$\Delta R_\text{out} = (c-1)/\alpha$ for some $c>1$ with a complexity of $\mathcal{O}(SK_a\log K_a)$.
This is a noteworthy results on its own, since it is a priori not clear, whether the bound 
\eqref{eq:sumrate} is achievable by an unsourced random access scheme, i.e. each user using
the same codebook.\\
The resulting achievable sum spectral efficiency can be calculated as in section
\ref{sec:analysis} with a subtle
but important difference, since the results on the outer code are valid only in the logarithmic regime
$J = \alpha\log_2 K_a$, i.e. $2^J = K_a^\alpha$ for $\alpha >1$.
According to Corollary \ref{cor:nnls} the error probability of
the inner code vanishes if the number of active users
scale no faster then $K_a = \mathcal{O}(L^2/\log^2(e2^J/L^2))$.
Using the scaling condition $J = \alpha \log_2 K_a$ and that $K_a \leq L^2$,
this implies that in the logarithmic regime the error probability of the inner code vanishes
if the number of active users scales as $K_a = \mathcal{O}(L^2/\log^2(L))$.
This gives a sum spectral efficiency
of 
\beq
\frac{K_aR_\text{out}J}{L} = \mathcal{O}\left(\frac{K_a\log K_a}{L}\right)
= \mathcal{O}\left(\frac{L}{\log L}\right).
\eeq
The order of this sum spectral efficiency is, by a factor $\log^2 L$, smaller then the one we calculated
in section \ref{sec:analysis}.
This is because the order of supported active users is smaller by exactly the same
$\log^2 L$ factor.
In section \ref{sec:analysis} we assumed that $J$ scales as $2^J = \delta L^2 = \mathcal{O}(K_a)$
for some $\delta>1$,
so that the ratio $K_a/2^J$ remains constant. It is not clear from the analysis in \cite{Ama2020a},
whether the probability of error of the outer tree code would vanish in the regime.
We can get a converse by evaluating the entropy bound \eqref{eq:sumrate}.
Let $2^J = \delta K_a$ with $\delta>1$, then 
$(1-2^{-J})^{K_a} = (1-\delta/K_a)^{K_a}\xrightarrow[K_a\to\infty]{} \exp(-\delta)$.
Therefore the binary entropy $\Hc_2((1-2^{-J})^{K_a})$ remains a constant in the limit
$J,K_a \to \infty$ and we get that 
\beq
K_aR_\text{out}J\leq \delta K_a \Hc_2(\exp(-\delta)).
\eeq
This shows that $R_\text{out} \to 0$
in the limit $J,K_a \to \infty$ is the best achievable asymptotic per-user outer rate,
but the outer sum rate $K_aR_\text{out}J$ is proportional to $K_a$.
The resulting sum spectral efficiencies scale as
\beq
\frac{K_aR_\text{out}J}{L} = \mathcal{O}\left(\frac{K_a}{L}\right) = \mathcal{O}(L).
\eeq
This means it could be possible to increase the achievable sum spectral efficiencies by a factor of $\log L$
by using an outer code that is able to achieve the entropy bound \eqref{eq:sumrate}
in the regime $2^J = \delta K_a$. It is not clear though whether the code of \cite{Ama2020a} or some other code can achieve
this.

\subsection{Simulations}
\label{sec:numerics}

The outer decoder requires a hard decision on the support of the estimated $\widehat{\gammav}[s]$. 
When $K_a$ is known, one approach consists of selecting the $K_a + \Delta$ largest entries in each section, where $\Delta \geq 0$
can be adjusted to balance between false alarm and misdetection in the outer channel.
However, the knowledge of $K_a$ is a very restrictive assumption in such type of systems. 
An alternative approach, which does not require this knowledge,
consists of fixing a sequence of thresholds $\{\nu_s : s \in [S]\}$ and let $\rhov[s]$ to be the binary vector of dimension $2^J$ 
with elements equal to 1 for all components of $\widehat{\gammav}[s]$ above threshold $\nu_s$. 
By choosing the thresholds, we can balance between missed detections and false alarms. 
Furthermore, we may consider the use of a  non-uniform decaying power allocation across the 
slots as described in \cite{Fen2019c}.

For the simulations in \figref{fig:sim} we choose $B = 96$ bits as payload size for each user, 
a frame of $S = 32$ slots of $L = 100$ dimensions per slot, yielding an overall block length $n = 3200$. Choosing
the binary subblock length $J = 12$, the inner coding matrix $\Am$ has dimension $100 \times 4096$ and
therefore is still quite manageable. We choose the columns of $\Am$ uniformly i.i.d. from
the sphere of radius $\sqrt{L}$.
For the outer code, we choose the following parity profile $p = [0,9,9, \ldots ,9,12,12,12]$, yielding an outer coding
rate $R_{out} = 0.25$ information bits per binary symbol. 
Notice also that if one wishes to send the same payload message using the piggybacking 
scheme of  \cite{Lar2012,Sen2017}, each user should make use of $2^{96}$ columns, which 
is totally impractical. 
All large scale fading coefficients are fixed to $g_k \equiv 1$.
In \figref{fig:sim} we fix $N_0 = 1$ and choose the transmit power (energy per symbol),
such that $E_b/N_0 = 0$dB and plot the sum of the
two types of error probabilities
$P_e = p^\text{msg}_\text{md} + p^\text{msg}_\text{fa}$, (see \eqref{eq:ura_pmd} and \eqref{eq:ura_pfa})
as a function of the number of active users for
different numbers of receive antennas $M$. 
\figref{fig:pe_ebn0} shows how $P_e$ falls as a function of $E_b/N_0$ for different values of $K_a$ and $M$.
\figref{fig:ebn0_over_K} shows the required values of $E_b/N_0$ as a function of $K_a$
to achieve a total error probability $P_e < 0.05$ for the
code parameters in Table \ref{tab:params}.
We use three different settings here, depending on the values
of the coherence block-length $L$. In all three the total block-length is fixed to $n=3200$
and $B\approx 100$, which gives a per-user spectral efficiency of $R\approx 0.031$ bits per channel use.
With $K_a = 300$ this corresponds to a total spectral efficiency
$\mu \approx 9$ bits per channel 
use, which is  significantly larger than today's LTE cellular systems (in terms of bit/s/Hz per sector) and definitely
much larger than IoT-driven schemes such as  LoRA \cite{Cen2016,Ban2016}. 
The simulations confirm qualitatively the behavior predicted in Sections \ref{sec:analysis}
and \ref{sec:analysis2}. The achievable
total spectral efficiencies seem to be mainly limited by the coherence block-length $L$,
and for a given total spectral efficiency the required energy-per-bit can be made arbitrary small by increasing $M$.
According to the Shannon-limit for the scalar Gaussian multiple access channel
(only one receive antenna, no fading) $E_b/N_0 > (2^{K_a R}-1)/(K_a R)$, and therefore one needs at least
$\approx 17.5$ dB to achieve a total spectral efficiency of $9$ bits per channel use. 
Here we find that gains of $20$ dB or more are possible even with non-coherent detection
by the use of multiple receive antennas. 
This shows also quantitatively that the non-coherent massive MIMO channel is very attractive for 
unsourced random access, since it preserves the same desirable characteristics of unsourced random access as in the 
non-fading Gaussian model of \cite{Pol2017} (users transmit without any pre-negotiation, and no use of pilot symbols
is needed), while the total spectral efficiency can be made as large as desired simply by increasing the number of receiver 
antennas. 

\begin{figure}
    \centering
    \includegraphics[width=0.5\linewidth]{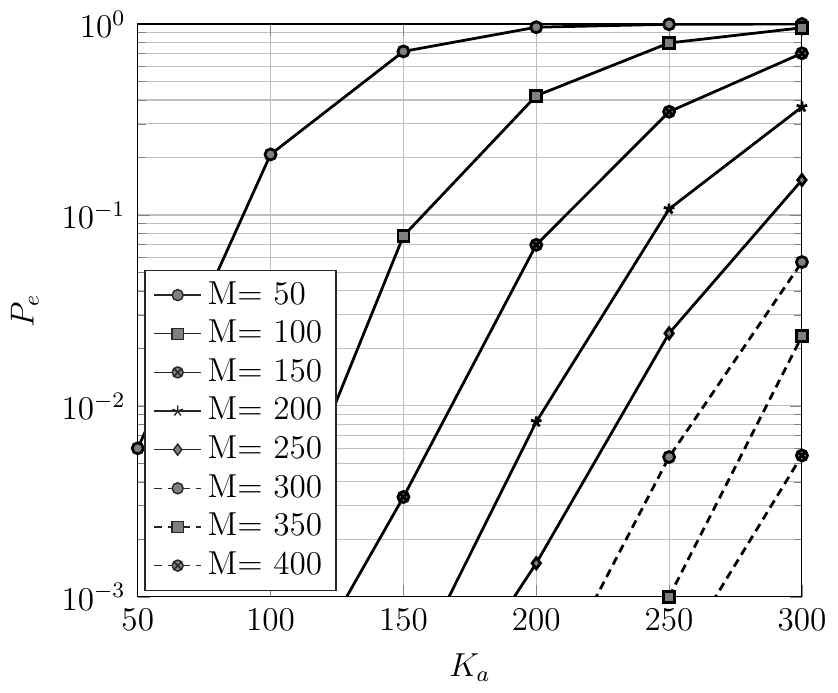}
    \captionof{figure}{ Error probability ($P_e = p^\text{msg}_\text{md} + p^\text{msg}_\text{fa}$) as a function of the number of
    active users for different numbers of receive antennas.
    $E_b/N_0=$ 0 dB, $L=100$, $n=3200$, $B = 96$ bits, $S = 32, J=12$. }
    \label{fig:sim}
\end{figure}
\begin{figure}
   \centering
   \subfloat[$K_a = 300$]{\includegraphics[height=6cm]{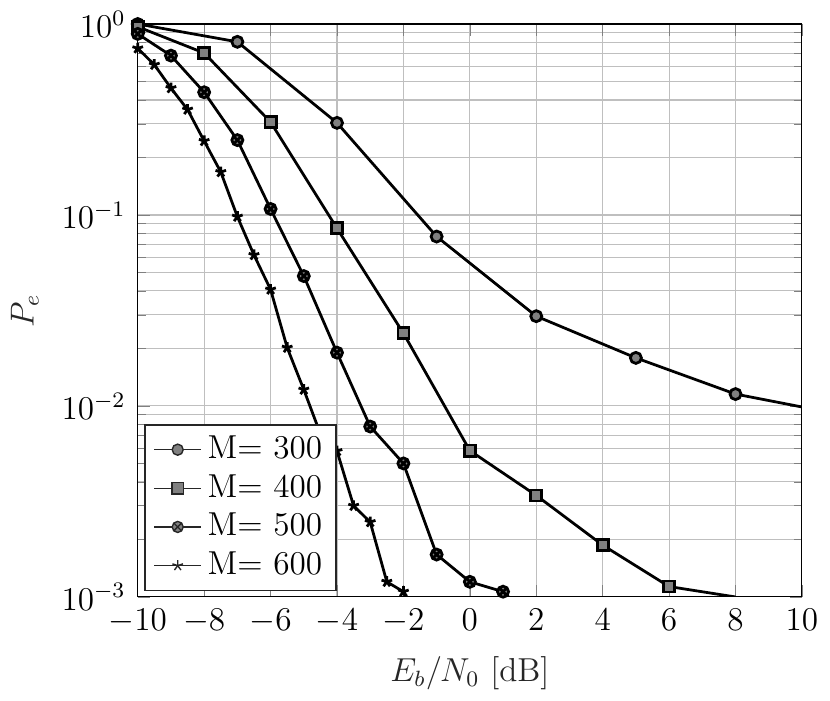}}\quad
   \subfloat[$M=300$]{\includegraphics[height=6cm]{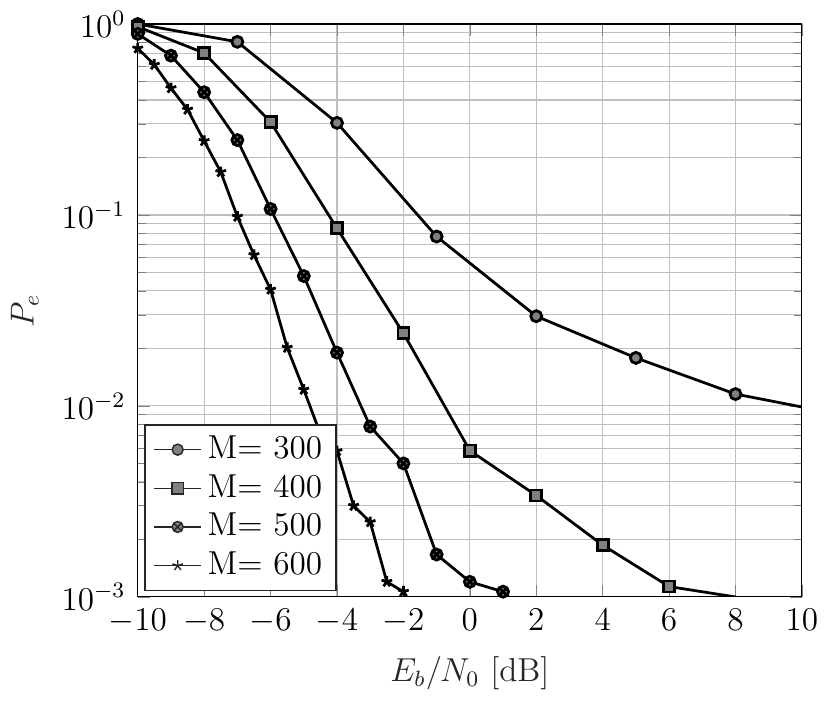}}
  \caption{Error probability ($P_e = p^\text{msg}_\text{md} + p^\text{msg}_\text{fa}$) as a function of $E_b/N_0$.
    $L=100$, $n=3200$, $b = 96$ bits, $S = 32, J=12$. }
    \label{fig:pe_ebn0}
\end{figure}
\begin{figure}
   \centering
   \includegraphics[width=0.8\linewidth]{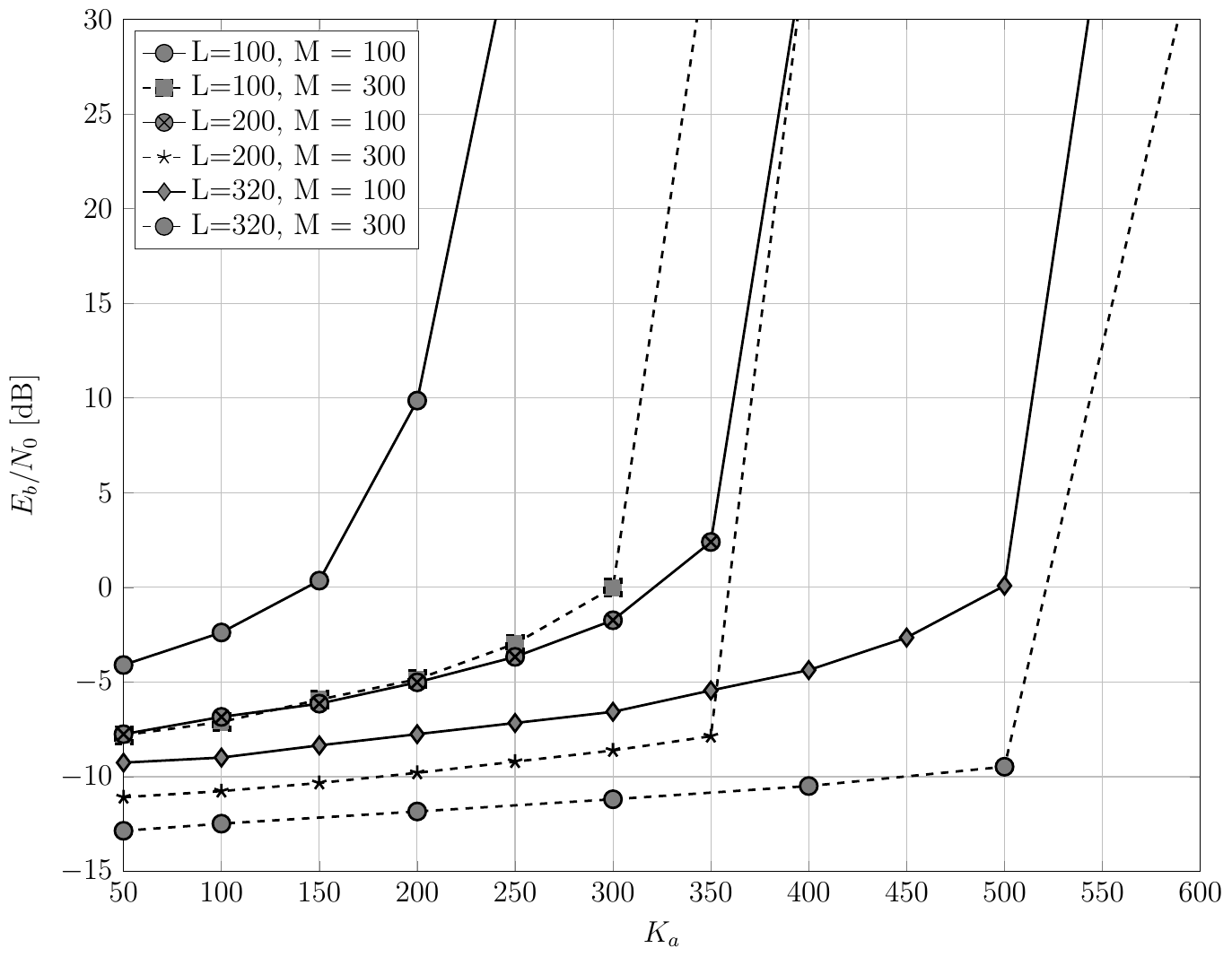}
  \caption{Required energy-per-bit to achieve $P_e < 0.05$.
  $L$ and $S$ are varied, while $n=3200$ and $B\approx 100$ are fixed. The precise parameters
  are given in Table \ref{tab:params}.}
  \label{fig:ebn0_over_K}
\end{figure}
\begin{table}
    \centering
        \begin{tabular}{|c|c|c|c|c|c|}
            \hline
            & S & J & $R_\text{out}$& Parity profile & B \\ \hline
            $L=100$ & 32 & 12 & 0.25 & [0,9,...,9,12,12,12] & 96 \\ \hline
            $L=200$ & 16 & 15 & 0.42 & [0,7,8,8,9,...,9,13,14] & 100 \\ \hline
            $L=320$ & 10 & 19 & 0.52 & [0,9,...,9,19] & 99 \\ \hline
        \end{tabular}
    \caption{Parameters for \figref{fig:ebn0_over_K}}
    \label{tab:params}
\end{table}

\section{Conclusion}

In this paper, we studied the problem of user activity detection in a massive MIMO setup,
where the BS has $M \gg 1$ antennas. We showed that with a coherence block containing $L$
signal dimensions one can reliably estimate the activity of $K_a = O(L^2/ \log^2 ( K_{\rm tot}/K_a))$
active users in a set of $K_{\rm tot}$ users,
which is a much larger than the previous bound $K_a = O(L)$
obtained via traditional compressed sensing techniques.
In particular, in our proposed scheme one needs to pay only a poly-logarithmic penalty $O(\log^2 (K_{\rm tot}/K_a))$
with respect to the number of potential users $K_{\rm tot}$,
which makes the scheme ideally suited for activity detection in IoT setups where the number of potential users
can be very large.
We discuss low-complexity algorithms for activity detection and provided numerical simulations to illustrate
our results and compared them with approximated message passing schemes recently proposed for the same scenario.
In particular, as a byproduct of our numerical investigation, we also showed a curious unstable behavior of
MMV-AMP in the regime where the number of receiver antennas is large,
which is precisely the case of interest with a massive MIMO receiver.
Finally, we proposed a scheme for unsourced random access where all users make use of the same codebook
and the receiver task is to come up with the list of transmitted messages.
We use our activity detection scheme(s) directly, where now the users’ signature sequences play the role of codewords,
and the number of total users plays the role of the number of total messages.
We showed that an arbitrarily fixed probability of error can be achieved
at any $E_b/N_0$ for sufficiently large number of antennas, and a total spectral efficiency that grows as $O(L \log L)$,
where $L$ is the code block length, can be achieved.
Such one-shot scheme is conceptually nice but not suited for typical practical applications with message payload
of the order of $B \approx 100$ bits, since it would require a codebook matrix with $2^B$ columns.
Hence, we have also considered the application of the concatenated approach pioneered in \cite{Ama2020a},
where the message is broken into a sequence of smaller blocks
and the activity detection scheme is applied as an inner encoding/decoding stage at each block,
while an outer tree code takes care of ``stitching together’’ the sequence of decoded submessages over the blocks.
Numerical simulations show the effectiveness of the proposed method.
It should be noticed that these schemes are completely non-coherent,
i.e., the receiver never tries to estimate the massive MIMO channel matrix of complex fading coefficients.
Therefore, the scheme pays no hidden penalty in terms of pilot symbol overhead,
often connected with the assumption of ideal coherent reception,
i.e., channel state information known to the receiver.

\appendices

\section{Proof of Theorem \ref{thm:ml:support}}
\label{appendix:ML}

The main line of arguments in this section is based on \cite{Kha2017}. In turns, the proof in \cite{Kha2017} is based 
on a RIP result which was claimed and successively retracted \cite{Kha2019}. 
The result was applied to a non-centered matrix and therefore could not have the claimed property. 
We fix this here, using our own new RIP result (Theorem \ref{thm:rip_main}) and, for the sake of clarity and self-contained presentation,
give a complete streamlined proof for the case of known LSFCs. At several points our proof technique
differs from \cite{Kha2017}, which results in the slightly better bound on $M$. 
Let us first introduce some notation.
\begin{definition}
    \label{def:renyi}
    For $t >1$ define the Renyi divergence of order $t$ between two probability
    densities $p$ and $q$ as
    \beq
        \mathcal{D}_t(p,q):=\frac{1}{t-1}\ln\int p(x)^tq(x)^{1-t}\mathrm{d}x
    \eeq
    \hfill $\lozenge$
\end{definition}

\begin{definition}
    A differentiable function $f$ is called strongly convex with parameter $m>0$ if the following
    inequality holds for all points $x,y$ in its domain:
    \beq
    f(y) \geq f(x) + \nabla f(x)^\top(y-x) + \frac{m}{2}\|x-y\|^2_2
    \eeq
    \hfill $\lozenge$
\end{definition}
Let $\bv^\circ$ denote the true activity pattern with known sparsity $K_a$, and $\bv^*$ be the output of the estimator \eqref{eq:ml_binary}. 
Using the union bound, we can write
\beq
\begin{split}
    \PP(\bv^*\neq\bv^\circ) 
    &= \PP\left(\max_{\bv\in\Theta_{K_a}\setminus\{\bv^\circ\}} p(\Ym|\bv) \geq p(\Ym|\bv^\circ)\right) \\
&= \PP\left(\bigcup_{\bv\in\Theta_{K_a}\setminus\{\bv^\circ\}}\{p(\Ym|\bv) \geq p(\Ym|\bv^\circ)\}\right)\\
    &\leq \sum_{\bv\in\Theta_{K_a}\setminus\{\bv^\circ\}}\PP(\Ym : p(\Ym|\bv) - p(\Ym|\bv^\circ)\geq 0)\\
    &\leq \sum_{\bv\in\Theta_{K_a}\setminus\{\bv^\circ\}}\PP(\Ym : p(\Ym|\bv) - p(\Ym|\bv^\circ)> -\alpha)
\end{split}
\label{eq:union_bound}
\eeq
for any $\alpha >0$.
With slight abuse of notation we define $\Sigmam(\bv):= \Am\Bm\Gm^\circ\Am^\herm + \sigma^2\Id_L$,
the covariance matrix for a given binary pattern $\bv$ for a fixed vector of LSFCs $\gv^\circ$,
with $\Bm = \diag(\bv)$ and $\Gm^\circ = \diag(\gv^\circ)$. 
Let $p_\bv := \mathcal{CN}(0,\Sigmam(\bv))$ denote the Gaussian distribution with covariance matrix
$\Sigmam(\bv)$, then 
$\log p(\Ym|\bv) = \sum_j \log p_\bv(\Ym_{:,j})$.
The following large deviation property of $\log p(\Ym|\bv)$ is established in \cite[Corollary 1]{Kha2017}:
\begin{theorem}
    \label{thm:ld_renyi}
    \beq
    \begin{split}
        &\PP\left(\log p(\Ym|\bv) - \log p(\Ym|\bv^\circ) > -\frac{M}{2}\mathcal{D}_{1/2}(p_\bv,p_{\bv^\circ})\right) \\
                                                    &\leq \exp\left(-\frac{M}{4}\mathcal{D}_{1/2}(p_\bv,p_{\bv^\circ})\right)
    \label{eq:ld_renyi}
    \end{split}
    \eeq
    where $\mathcal{D}_{1/2}(p_\bv,p_{\bv^\circ})$ is the Renyi divergence of order $1/2$ between
    $p_\bv$ and $p_\bv^\circ$ defined in Definition \ref{def:renyi}. \hfill $\square$
\end{theorem}

The result of Theorem \ref{thm:ld_renyi} holds only if $\mathcal{D}_{1/2}(p_\bv,p_{\bv^\circ})>0$, so in the following
we will establish conditions under which this is true. First, note that since $p_\bv$ and $p_{\bv^\circ}$
are zero-mean Gaussian distributions with covariance matrices
$\Sigmam(\bv)$ and $\Sigmam(\bv^\circ)$ resp.,
their Renyi divergence of order $t$ can be expressed in closed form as:
\beq
    \mathcal{D}_{t}(p_\bv,p_{\bv^\circ}) 
    = \frac{1}{2(1-t)}\log \frac{|(1-t)\Sigmam(\bv)+t\Sigmam(\bv^\circ)|}{|\Sigmam(\bv)|^{1-t}|\Sigmam(\bv^\circ)|^t}
\eeq
Let $\psi(\bv):=-\log|\Sigmam(\bv)|$, then we can see that
$\mathcal{D}_{t}(p_\bv,p_{\bv^\circ})\geq t\frac{m^*}{4} \|\bv-\bv^\circ\|_2^2$,
with $m^*$ being the strong convexity constant of $\psi(\cdot)$,
is equivalent to
\beq
\psi((1-t)\bv + t\bv^\circ) \leq (1-t)\psi(\bv) + t\psi(\bv^\circ) - \frac{1}{2}m^*t(1-t)\|\bv-\bv^\circ\|_2^2.
\label{eq:s_convex}
\eeq
Here we used the fact that 
\beq
\begin{split}
-\psi((1-t)\bv + t\bv^\circ) &= \log|\Sigmam((1-t)\bv+t\bv^\circ)| \\
                            &= \log|\Am((1-t)\Bm + t\Bm^\circ)\Gm^\circ\Am^\herm + \sigma^2\Id_{\dimPilots}|\\
                            &= \log|(1-t)\Sigmam(\bv) + t\Sigmam(\bv^\circ)|
\end{split}
\eeq
Inequality \eqref{eq:s_convex} is precisely the condition that $\psi(\cdot)$
is strongly convex along the line connecting
$\bv$ and $\bv^\circ$. So if $\psi(\cdot)$ is strongly convex on the set
of $2K_a$-sparse vectors, then 
\beq
\mathcal{D}_{t}(p_\bv,p_{\bv^\circ})\geq t\frac{m^*}{4} \|\bv-\bv^\circ\|_2^2
\label{eq:renyi_bound}
\eeq
holds for any $K_a$-sparse vectors $\bv$ and $\bv^\circ$.
Let $\bv_1,\bv_2\in\Theta_{K_a}$ be two arbitrary $K_a$-sparse vectors.
Since $\log|\cdot|$ is differentiable on $\RR^+$, a Taylor expansion of $\psi(\bv_1)$ around $\bv_2$ gives:
\beq
\begin{split}
\psi(\bv_1) = \psi(\bv_2) 
&+ \langle\nabla \psi(\bv_1),\bv_2 - \bv_1\rangle \\
&+ \frac{1}{2}(\bv_2-\bv_1)^\top \nabla^2\psi(\bv_r)(\bv_2 - \bv_1) 
\end{split}
\eeq
for $\bv_r = (1-r)\bv_1 + r\bv_2$ with some $r\in[0,1]$.
Let $\Delta\bv := \bv_2 - \bv_1$, then the strong convexity of $\psi(\cdot)$ is
equivalent to
\beq
\sum_{i,j} \left.\frac{\partial \psi}{\partial b_i\partial b_j}\right|_{\bv = \bv_r}\Delta b_i\Delta b_j 
\geq m^*\|\bv_2 - \bv_1\|_2^2.
\label{eq:strong_convexity}
\eeq
The derivatives of $\psi$ are given by:
\beq
\left.\frac{\partial \psi}{\partial b_i}\right|_{\bv = \bv_r}
= -\tr(\Sigmam( \bv_r )^{-1}g_i^\circ\av_i\av_i^\herm)
\eeq
\beq
\left.\frac{\partial \psi}{\partial b_i\partial b_j}\right|_{\bv = \bv_r}
= \tr(\Sigmam(\bv_r)^{-1}g_i^\circ\av_i\av_i^\herm\Sigmam(\bv_r)^{-1}g_j^\circ\av_j\av_j^\herm)
\eeq
Next we will calculate $m^*$. It holds that
\beq
\begin{split}
    \sum_{i,j} \left.\frac{\partial \psi}{\partial b_i\partial b_j}\right|_{\bv = \bv_r}  
        \Delta b_i\Delta b_j  
        &= \sum_{i,j}\tr\left(\Sigmam(\bv_r)^{-1}\Delta b_i g_i^\circ\av_i\av_i^\herm
        \Sigmam(\bv_r)^{-1}\Delta b_j g_j^\circ\av_j\av_j^\herm\right)\\
        &= \tr\left(\Sigmam(\bv_r)^{-1}\left(\sum_i \Delta b_i g_i^\circ\av_i\av_i^\herm\right)
        \Sigmam(\bv_r)^{-1}\left(\sum_j \Delta b_j g_j^\circ\av_j\av_j^\herm\right)\right)\\
        &= \tr\left(\Sigmam(\bv_r)^{-1}(\Sigmam(\bv_2)-\Sigmam(\bv_1))
        \Sigmam(\bv_r)^{-1}(\Sigmam(\bv_2)-\Sigmam(\bv_1))\right)\\
        &\geq \sigma_\text{min}(\Sigmam(\bv_r)^{-1})
         \tr\left((\Sigmam(\bv_2)-\Sigmam(\bv_1))
        \Sigmam(\bv_r)^{-1}(\Sigmam(\bv_2)-\Sigmam(\bv_1))\right)\\
        &\geq \sigma^2_\text{min}(\Sigmam(\bv_r)^{-1})\|\Sigmam(\bv_2)-\Sigmam(\bv_1)\|_F^2\\
        &= \frac{\|\Sigmam(\bv_2)-\Sigmam(\bv_1)\|_F^2}{\sigma^2_\text{max}(\Sigmam(\bv_r))}.
\end{split}
\label{eq:2nd_deriv_bound}
\eeq
Here $\sigma_\text{min}(\Am)$ (resp., $\sigma_\text{max}(\Am)$) denotes the minimum (resp., maximum) singular value of $\Am$.
In the first and the second inequality  in (\ref{eq:2nd_deriv_bound}) 
we used the fact that $\tr(\Am\Bm)\geq \sigma_\text{min}(\Am)\tr(\Bm)$ for positive semi-definite matrices
$\Am,\Bm$, and in the second inequality in (\ref{eq:2nd_deriv_bound})  we used the fact that the covariance matrix is symmetric and
$\tr(\Am^\top \Am) = \|\Am\|_F^2$.

We can rewrite
$\|\Sigmam(\bv_2)-\Sigmam(\bv_1)\|_F^2 = \|\bA(\gv^\circ\odot(\bv_2 - \bv_1))\|_2^2$, where
$\bA\in\CC^{\dimPilots^2\times K_\text{tot}}$ is the matrix defined in (\ref{eq:bA_definition}), 
obtained by stacking  the $L^2$-dimensional vectors $\text{vec}(\av_k\av_k^\herm)$ by columns. We show in 
\eqref{eq:thm:short:Av:lowerbound} that 
$\|\bA\xv\|_2\geq\|\mathring{\bA}\xv\|_2$ holds $\forall \xv \in \RR^{K_\text{tot}}$,
with $\mathring{\bA}$ being the centered version
of $\bA$, which is defined in \eqref{eq:a_centered}.

We show in Theorem \ref{thm:rip_main},
that, with probability at least
$1-\exp(-c_\delta\dimPilots)$, $\mathring{\bA}/\sqrt{\dimPilots(\dimPilots-1)}$,
the centered and rescaled
version of $\bA$
has RIP of order $2K_a$ with constant $\delta_{2K_a} < \delta$ if condition \eqref{eq:rip_condition} is fulfilled.
In particular 
$\|\mathring{\bA}\xv\|_2^2\geq(1-\delta_{2 K_a})\dimPilots(\dimPilots-1)\|\xv\|_2^2$
holds for all $2K_a$-sparse vectors $\xv$. So the RIP of $\mathring{\bA}$ implies that
\beq
\begin{split}
\|\Sigmam(\bv_2)-\Sigmam(\bv_1)\|_F^2
&\geq (1-\delta_{2K_a})\dimPilots(\dimPilots-1)\|\gv^\circ\odot(\bv_2 - \bv_1)\|_2^2\\
&\geq (1-\delta_{2K_a})\dimPilots(\dimPilots-1)g_\text{min}^2\|\bv_2-\bv_1\|_2^2\\
&\geq \frac{1}{2}(1-\delta_{2K_a})\dimPilots^2g_\text{min}^2\|\bv_2-\bv_1\|_2^2
\end{split}
\label{eq:fro_norm_lower}
\eeq
An upper bound on $\sigma_\text{max}^2(\Sigmam(\bv_r)) = \|\Sigmam(\bv_r)\|_{op}^2$ 
can be found as follows. Note that for any binary $2K_a$-sparse vector $\bv$, it holds that
\beq
\begin{split}
    \sigma_\text{max}(\Sigmam(\bv))
    &= \|\Sigmam(\bv)\|_{op} \\
    &= \left\Vert\sum_{k=1}^{K_\text{tot}}g^\circ_kb_k\av_k\av_k^\herm + \sigma^2\Id\right\Vert_{op} \\
    &\leq g_\text{max}\left\Vert\sum_{k\in \supp(\bv)}\av_k\av_k^\herm\right\Vert_{op} + \sigma^2 \\
    &= g_\text{max}\left\Vert\sum_{k\in \supp(\bv)}(\av_k\av_k^\herm-\Id) + 2K_a\Id \right\Vert_{op} +\sigma^2\\
    &\leq g_\text{max}\left\Vert\sum_{k\in \supp(\bv)} (\av_k\av_k^\herm - \Id)\right\Vert_{op} + g_\text{max}2K_a + \sigma^2 
\end{split}
\label{eq:s_max_1}
\eeq
Now $\sum_{k\in \supp(\bv)} (\av_k\av_k^\herm - \Id)$
is a sum of $2K_a$ random matrices  
$\av_k\av_k^\herm$, with $\av_k$ drawn i.i.d. from the sphere of radius $\sqrt{\dimPilots}$, and therefore
sub-Gaussian.
A generic large deviation result for such matrices, e.g., the complex version of
\cite[Theorem 4.6.1]{vershynin_2018}, shows that
\beq
\left\Vert\sum_{k\in \supp(\bv)} (\av_k\av_k^\herm - \Id)\right\Vert_{op}
\leq \left(\sqrt{K_a} + C\left(\sqrt{\dimPilots} + t\right)\right)^2
\label{eq:op_norm_bound}
\eeq
holds with probability at least $1 - 2\exp(-t^2)$ for some universal constant $C>0$.
Let $t=\sqrt{\beta}\max(\sqrt{K_a},\sqrt{\dimPilots})$ for some $\beta>0$.
Then \eqref{eq:s_max_1} gives that
\beq
\sigma_\text{max}(\Sigmam(\gammam)) \leq (1+\beta C')g_\text{max}\max\{K_a,\dimPilots\} + \sigma^2
\label{eq:s_max_2}
\eeq
holds with probability at least $1-\exp(-\beta\max\{K_a,\dimPilots\})$
for some universal constants $C'>0$. 
So \eqref{eq:2nd_deriv_bound} can be further bounded using \eqref{eq:fro_norm_lower} and \eqref{eq:s_max_2} as
\beq
    \frac{\|\Sigmam(\bv_2)-\Sigmam(\bv_1)\|_F^2}{\sigma^2_\text{max}(\Sigmam(\bv_r))}
    \geq \frac{(1-\delta_{2K_a})g_\text{min}^2\|\bv_2-\bv_1\|_2^2}
    {2\left((1+C'\beta)g_\text{max}\max\left\{\frac{K_a}{\dimPilots},1\right\}+\frac{\sigma^2}{\dimPilots}\right)^2}
\eeq
Together with \eqref{eq:strong_convexity} and \eqref{eq:2nd_deriv_bound} this implies
that, if the pilot matrix satisfies the RIP of order $2K_a$ with constant
$\delta_{2K_a}<1$,
then $\psi(\cdot)$ is strongly convex along the line between any two $K_a$-sparse vectors
with constant
\beq
m^* \geq
\frac{(1-\delta_{2K_a})g_\text{min}^2}
{2\left((1+C'\beta)g_\text{max}\max\left\{\frac{K_a}{\dimPilots},1\right\}+\frac{\sigma^2}{\dimPilots}\right)^2}
\label{eq:m_bound}
\eeq
with probability exceeding $1 - \exp(-\beta \max\{K_a,\dimPilots\})$. Since the bound is independent
of the chosen vectors and the number of $2K_a$ sparse binary vectors
is bounded by $\binom{K_\text{tot}}{2K_a}\leq(eK_\text{tot}/K_a)^{2K_a} \leq (eK_\text{tot}/K_a)^{2\max\{K_a,\dimPilots\}}$,
\eqref{eq:m_bound} holds in the set of \emph{all} $2K_a$-sparse vectors with probability
exceeding 
\beq
1 - \exp\left(-2\max\{K_a,\dimPilots\}\left(\frac{\beta}{2} - \log\left(\frac{eK_\text{tot}}{2K_a}\right)\right)\right)
\eeq
This probability exceeds $1-\epsilon$ if
\beq
\beta \geq 2\log \left(\frac{eK_\text{tot}}{2K_a}\right) + \frac{\log(2/\epsilon)}{\max\{K_a,\dimPilots\}}
\eeq

We get that 
\beq
\begin{split}
    m^* \geq
    \frac{(1-\delta_{2K_a})g_\text{min}^2}
{2\left(C'\left(2\log\left(\frac{eK_\text{tot}}{2K_a}\right)+ \frac{\log(2/\epsilon)}{\max\{K_a,\dimPilots\}}\right)g_\text{max}\max\left\{\frac{K_a}{\dimPilots},1\right\}+\frac{\sigma^2}{\dimPilots}\right)^2}
\label{eq:m_bound2}
\end{split}
\eeq
holds with probability exceeding $1-\epsilon$.

Let $k_d=\|\bv_2-\bv_1\|_0\leq 2K_a$
denote the number of positions in which $\bv_2$ and $\bv_1$ differ, i.e. their
Hamming distance. Then the Renyi divergence \eqref{eq:renyi_bound}
can be lower bound as:
\beq
\mathcal{D}_{t}(p_\bv,p_{\bv^\circ})\geq t\frac{m^*}{4} k_d
\label{eq:renyi_lower}
\eeq
Putting everything together, we can complete the union bound. Note that 
there are $\binom{K_a}{k_d}\binom{K_\text{tot}-K_a}{k_d}\leq (3eK_\text{tot}K_a)^{k_d}$ ways to
choose a support which differs from the true support in $k_d$ positions.
Now, denote by $\mathcal{C}$ the event that the pilot matrix $\Am$ is such that the RIP condition
\eqref{eq:fro_norm_lower} holds, and the bound \eqref{eq:m_bound2}.
Using \eqref{eq:union_bound}, Theorem \ref{thm:ld_renyi} and \eqref{eq:renyi_lower} we get that
\beq
\begin{split}
    \PP(\bv^*\neq\bv^\circ,\mathcal{C}) 
    &\leq \sum_{\bv\in\Theta_{K_a}\setminus\{\bv^\circ\}}\exp\left(-\frac{M}{4}\mathcal{D}_{1/2}(p_\bv,p_{\bv^\circ})\right) \\
    &\leq \sum_{k_d=1}^{2K_a}(3eK_\text{tot}K_a)^{k_d}\exp\left(-M\frac{m^*}{4}k_d\right) \\
    &= \sum_{k_d=1}^{2K_a}\exp\left(-k_d\left(M\frac{m^*}{4}-\log(3eK_\text{tot}K_a)\right)\right) 
\end{split}
\eeq
So let 
\beq
M \geq \frac{4}{m^*}\log\left(3eK_\text{tot}K_a\frac{1+\epsilon}{\epsilon}\right)
\eeq
which is precisely condition \eqref{eq:ml_M_condition}, then
\beq
\begin{split}
    \PP(\bv^*\neq\bv^\circ,\mathcal{C}) 
    &\leq \sum_{k_d=1}^{2K_a}\left(\frac{\epsilon}{1+\epsilon}\right)^{k_d}\\
    &\leq \epsilon
\end{split}
\eeq
Finally
\beq
\begin{split}
\PP(\bv^*\neq\bv^\circ) 
&\leq  \PP(\bv^*\neq\bv^\circ,\mathcal{C}) + \PP(\bar{\mathcal{C}})\\
&\leq \epsilon + \epsilon + \exp(-C_\delta\dimPilots)
\end{split}
\eeq
This concludes the proof of Theorem \ref{thm:ml:support}. 

\section{Proof of the RIP, Theorem \ref{thm:rip_main}}
\label{appendix:ripproof}
Let us first define some basic properties.
\begin{definition}[Sub-Exponential Norm]
    \label{def:sub-exp}
    Let $X$ be a real scalar random variable. Define the sub-exponential norm
    of $X$ as
    \beq
    \|X\|_{\psi_1}:= \inf\left\{t>0:\EE \left [ \exp\left(\frac{|X|}{t}\right) \right ] \leq 2\right\}.
    \eeq
\end{definition}
A well known property of sub-exponential variables is that
\beq
\PP(|X|>t)\leq 2\exp(-ct/\|X\|_{\psi_1}) \quad \forall t>0
\eeq
for some universal constant $c>0$.
\begin{definition}[Sub-Exponential Random Vector]
    \label{def:sub-exp-vector}
    Let $\Xm$ be a random vector in $\RR^n$. $\Xm$ is said to be sub-exponential
    if all its marginals are scalar sub-exponential random variables, i.e. if
    \beq
        \sup_{\xv\in S^{n-1}}\|\langle \Xm,\xv\rangle\|_{\psi_1} < \infty
    \eeq
    then we define $\|\Xm\|_{\psi_1}:=\sup_{\xv\in S^{n-1}}\|\langle \Xm,\xv\rangle\|_{\psi_1}$,
    where $S^{n-1}$ is the unit sphere in $\RR^n$.
\end{definition}
Basic properties of sub-exponential random variables and vectors can be found e.g. in Ch. 2 and 3 of
\cite{vershynin_2018}.
\begin{definition}[Convex Concentration Property (2.2 in \cite{Ada2015})]
    \label{def:convex_concentration}
    Let $\Xm$ be a random vector in $\RR^n$. $\Xm$ has the convex concentration property
    with constant $K$ if for every 1-Lipschitz convex function $\phi: \RR^n\to \RR$,
    we have $\EE [|\phi(\Xm)|] <\infty$ and for every $t>0$,
    \beq
    \PP(|\phi(\Xm) - \EE [ \phi(\Xm)] |\geq t) \leq 2\exp(-t^2/K^2)
    \eeq
\end{definition}
For the RIP of $\mathring{\bA}/\sqrt{m}$ we first establish the following results for generic
matrices $\Rm\in\RR^{m\times N}$ with independent normalized columns.
\begin{theorem}   
    \label{thm:rip}
    Let 
    $\Rm\in\mathbb{R}^{m\times N}$ be a matrix 
    with independent columns $\Rm_{:,i}$
    normalized such that $\EE[\|\Rm_{:,i} \|_2^2] = m$, with
    $\psi_1$-norm at most $\psi$. Furthermore, assume that the distribution of the columns
    satsifies
    \beq
    \PP(|\|\Rm_{:,i}\|_2^2 - m| > tm)\leq \exp(-c\sqrt{m})
    \label{eq:tail_bound}
    \eeq for some constant $c>0$.
    Also assume that $N \geq m$ and $\sqrt{m} > c^\prime\log 4N$
    for some universal
    constant $c^\prime>0$. Then
    the RIP constant of $\Rm/\sqrt{m}$ satisfies $\delta_{2s}(\Rm/\sqrt{m})<\delta$ 
    with probability $\geq 1 - \exp(-C'\sqrt{c_{\delta,\xi}m})$ for
    \begin{equation}
        2s = c_{\delta,\xi}\frac{m}{\log^2(e N /c_{\delta,\xi} m)},
    \end{equation}
    where
    $c_{\delta,\xi} \leq \min\{1,\frac{\delta^2}{(3C\xi^2)^2}\}$
    for any $\xi>\psi+1$ and $C,C'>0$ are universal constants. \hfill $\square$
\end{theorem}
\begin{proof}
    We make use of the following generic RIP result from
    \cite[Theorem~3.3]{Adamczak2011} for matrices with i.i.d. sub-exponential
    columns:
    \begin{theorem}
        \label{thm:adam}
        Let $m\geq 1$ and $s,N$ be integers such that $1\leq s\leq \min (N,m)$.
        Let $\Rm_{:,1},...,\Rm_{:,N} \in \RR^m$ 
        be independent sub-exponential random vectors 
        normalized such that $\EE[ \|\Rm_{:,i}\|^2 ]  = m$ and
        let $\psi = \max_{i\leq N}\|\Rm_{:,i}\|_{\psi_1}$. Let $\theta^\prime \in(0,1)$,
        $K,K^\prime \geq 1$ and set $\xi = \psi K + K^\prime$. Then for the
        matrix $\Rm$ with columns $\Rm_{:,i}$
        \beq
        \delta_s\left(\frac{\Rm}{\sqrt{m}}\right) \leq C \xi^2\sqrt{\frac{s}{m}}
        \log\left(\frac{eN}{s\sqrt{\frac{s}{m}}}\right)
        + \theta^\prime
        \eeq
        holds with probability larger than
        \beq
        \begin{split}
            1 &- \exp\left(-\hat{c}K\sqrt{s}\log\left(\frac{eN}{s\sqrt{\frac{s}{m}}}\right)\right) \\
              &- \PP\left(\max_{i\leq N}\|\Rm_{:,i}\|_2\geq K^\prime\sqrt{m}\right) 
              - \PP\left(\max_{i\leq N}\left|\frac{\|\Rm_{:,i}\|_2^2}{m} - 1\right|\geq \theta^\prime\right),
        \end{split}
        \label{eq:adam_error}
        \eeq
        where $C,\hat{c} > 0$ are universal constants.  \hfill $\square$
    \end{theorem}
    
    In order to prove Theorem \ref{thm:rip} we shall apply Theorem \ref{thm:adam}. 
    Let us abbreviate $\delta_s = \delta_{s}\left(\frac{\Rm}{\sqrt{m}}\right)$.
    We set $\theta^\prime = \delta_s/2$ and
    Therefore, we consider the bound
    \begin{equation}
        \delta_{s}\leq
        C \xi^2 \sqrt{\frac{s}{m}} \log\left(\frac{eN}{s \sqrt{s/m}}\right)
        =: D, 
        \label{eq:deltas}
    \end{equation}
    that holds with probability larger than 
    \eqref{eq:adam_error}.
    Let $s = cm/\log^2(e\frac{N}{c m})$ for any $0 <c\leq1$.
    Note that the conditions $c\leq 1$ and $N\geq m$ guarantee that
    $\log (e\frac{N}{cm})\geq 1$.
    Plugging into \eqref{eq:deltas}
    we see that the RIP-constant
    satisfies 
    \begin{align}
        \delta_{s}&\leq C\xi^2\sqrt{c}\frac{
            \log(e(\frac{N}{c m})^{3/2}\log^3(e\frac{N}{c m}))
        }{\log(e\frac{N}{c m})} \\
        &\leq C\xi^2\sqrt{c}
        \left(\frac{3}{2} +
        \frac{3\log\log e\frac{N}{c m}}{\log e\frac{N}{c m}}
        \right)
        \\
        &\leq C\xi^2\sqrt{c}\left(\frac{3}{2} + \frac{3}{e} \right)\\
        &\leq 3C\xi^2\sqrt{c}
    \end{align}
    where in the first line we made use of $m\leq N$ and in the last line we used
    $\log\log x/\log x \leq 1/e$.
    We proceed to bound the probability in \eqref{eq:adam_error}.
    Using the union bound, the last two terms in \eqref{eq:adam_error} can be 
    bound as:
    \begin{align}
        \PP\left(\max_{i\leq N}\|\Rm_{:,i}\|_2\geq K^\prime\sqrt{m}\right) 
        &\leq N\PP\left(\|\Rm_{:,i}\|^2_2\geq K^{\prime2}m\right)\\
        &\leq N\PP\left(|\|\Rm_{:,i}\|^2_2 - m| \geq (K^{\prime2}-1)m\right)
    \end{align}
    and
    \begin{align}
    \PP\left(\max_{i\leq N}\left|\frac{\|\Rm_{:,i}\|_2^2}{m} - 1\right|\geq \theta^\prime\right)
        &\leq N\PP\left(|\|\Rm_{:,i}\|^2_2 - m| \geq \theta^\prime m\right)
    \end{align}
    By choosing $K^\prime = \sqrt{1+\theta}$ we can treat both terms equivalently.
    By the tail bound assumption on the norms we have
    \beq
        2N\PP\left(|\|\Rm_{:,i}\|^2_2 - m| \geq \theta^\prime m\right)
        \leq \exp(\log 4N - c\sqrt{m})
        \leq \exp(-c^\prime \sqrt{m})
        \label{eq:prob_bound2}
    \eeq
    for some positive constants $c,c^\prime$. The last inequality follows from the
    assumption $\sqrt{m} > c^\prime\log 4N$. Also we have
    \begin{align}
        \exp\left(-\widehat{c}K\sqrt{s}\log \left(e\frac{N\sqrt{m}}{s^{3/2}}\right)\right) 
        &\leq \exp\left(-\widehat{c}K\sqrt{s}\log \left(e\frac{N}{m}\right)\right) \\
        &= \exp(-\widehat{c}K\sqrt{c}\sqrt{m})
        \label{eq:prob_bound1}
    \end{align}
    where in the first line we used $s \leq m$.
    By choosing $c$ small enough such that $\hat{c}K\sqrt{c}\sqrt{m}<c^\prime$ we get from
    \eqref{eq:prob_bound2} and \eqref{eq:prob_bound1} that
    \beq
    \mathbb{P}(\delta_{s}>D) \leq
        2\exp(-\widehat{c}K\sqrt{c}\sqrt{m})
    \eeq
    The statement of
    Theorem \ref{thm:rip} 
    follows by choosing $c$ small enough such that $\delta_s \leq\delta$.
\end{proof}

\newcommand{\AmatrixR}{\Amatrix^R}
\newcommand{\AcolR}{\Acol^R}
We want to apply Theorem \ref{thm:rip}, which holds for real values matrices $\Rm$, to the matrix 
\beq
\label{eq:a_real}
\AmatrixR
:=\sqrt{2}[\Re(\mathring{\bA});\Im(\mathring{\bA})]
\in\mathbb{R}^{2\dimPilots(\dimPilots - 1)\times\dimParam},
\eeq
i.e., the real matrix obtained by stacking real and imaginary part of $\mathring{\bA}$,
with $m=2\dimPilots(\dimPilots-1)$ and $N = K_\text{tot}$. 
Consider the $k$-th column $\AmatrixR_{:,k}$ of $\AmatrixR$. We have
\beq
\begin{split}
    \EE\left[ \|[\Re(\mathring{\bA}_{:,k});\Im(\mathring{\bA}_{:,k})]\|_2^2\right]/c_L
    &= \EE\left[\|\Re(\mathring{\bA}_{:,k})\|_2^2 + \|\Im(\mathring{\bA}_{:,k})\|^2_2\right]/c_L \\
    &= \EE\left[\|\mathring{\bA}_{:,k}\|_2^2\right]/c_L \\
    &=\EE\left[\sum_{i\neq j}|a_{k,i}a_{k,j}^\herm|^2\right]\\
    &= \EE\left[\left(\sum_{i=1}^{\dimPilots}|a_{k,i}|^2\right)^2\right] - \sum_{i=1}^{\dimPilots}\EE[|a_{k,i}|^4]\\
    &= \dimPilots^2 - \dimPilots\kappa_a\\
    &= \dimPilots(\dimPilots-\kappa_a), 
\end{split}
\label{eq:norm_calc}
\eeq
Since $c_L$ was chosen exactly as $(L-1)/(L-\kappa_a)$ we have $\EE[\|\AmatrixR_{:,k}\|_2^2] = m$.
Note that a spherical vector is especially sub-Gaussian, which implies that its fourth moment 
can be bound by a constant independent of the dimension \cite{vershynin_2018}.
To apply Theorem \ref{thm:rip} we need to show that
the columns of $\AmatrixR$ are sub-exponential with
        $\psi_1$ norm independent of the dimension.

\begin{sloppypar}
    Note that for any vector $\uv\in\RR^{2\dimPilots(\dimPilots-1)}$ the marginal
    $\langle \AmatrixR_{:,k},\uv \rangle$ can be expressed as a quadratic form in 
    ${\av_k^R:=\sqrt{2}[\Re(\av_k);\Im(\av_k)]}$ as the following calculation shows.
    Let $\Um,\widetilde{\Um}\in\RR^{\dimPilots\times\dimPilots}$ 
    be two matrices with zeros on the diagonal
    such that $\uv = [\text{vec}_\text{non-diag}(\Um);\text{vec}_\text{non-diag}(\widetilde{\Um})]$. Then it holds:
\end{sloppypar}
\beq
\begin{split}
    \left\langle\AmatrixR_{:,k},\uv\right\rangle
    &=\sqrt{2}
    \sum_{i\neq j}\left( \Re(a_{k,i}a^\herm_{k,j})U_{ij}
    + \Im(a_{k,i}a^\herm_{k,j})\widetilde{U}_{ij}\right)\\
    &= (\av_k^R)^\top \Qm_\uv\av^R_k
\end{split}
\label{eq:real_marginal}
\eeq
with
\beq
\Qm_\uv = 
\frac{1}{\sqrt{2}}\left(\begin{matrix}
        \Um & \widetilde{\Um}\\
        -\widetilde{\Um} & \Um 
\end{matrix}\right)
\eeq
and therefore $\|\Qm_\uv\|^2_F  = \|\uv\|^2_2$.\\ 
This form of $\Qm_\uv$ follows from the identities:
\begin{align}
    \Re(a_{k,i}a_{k,j}^\herm) &= \Re(a_{k,i})\Re(a_{k,j})+\Im(a_{k,i})\Im(a_{k,j})\\
    \Im(a_{k,i}a_{k,j}^\herm) &= -\Re(a_{k,i})\Im(a_{k,j})+\Im(a_{k,i})\Re(a_{k,j})
\end{align}
We can now use the following concentration result for quadratic forms
from \cite{Ada2015} which states
that a random vector which satisfies the convex concentration property
also satisfies the following inequality, known as Hanson-Wright inequality \cite{Rud2013}:
\begin{theorem}[Theorem 2.5 in \cite{Ada2015}]
    \label{thm:hanson-wright-convex}
    Let $\Xm$ be a mean zero random vector in $\RR^n$, which satisfies the convex concentration
    property with constant $B$, then for any $n\times n$ matrix $\Ym$ and every $t>0$,
    \begin{align}
        &\PP(|\Xm^\top \Ym\Xm - \EE [ \Xm^\top \Ym\Xm] |>t)\nonumber\\
        &\leq 2 \exp\left(-c\min\left(\frac{t^2}{2B^4\|\Ym\|_F^2},\frac{t}{B^2\|\Ym\|_{op}}\right)\right)
    \end{align}\hfill $\square$.
\end{theorem}
Note that a random variable with such a mixed tail behavior is 
especially sub-exponential. This can be seen by bounding its moments. 
Let $Z$ be a random variable with
\beq
\PP(|Z|>t)\leq2\exp\left(-c\min\left(\frac{t^2}{B^4\|\Ym\|_F^2},\frac{t}{B^2\|\Ym\|_{op}}\right)\right)
\eeq
Since $\|\Ym\|_{op} \leq \|\Ym\|_F$, we have
$\PP(|Z|>t) \leq 2\exp(-c\min(x(t)^2,x(t)))$ for $x(t) = \frac{t}{B^2\|\Ym\|_F}$.
It follows
\begin{equation}\begin{split}
    &\EE [ |Z|^p ]
    = \int_0^\infty \PP(|Z|^p>u)\mathrm{d}u
    = p\int_0^\infty \PP(|Z|>t) t^{p-1}\mathrm{d}t \\
    &\leq 2p(B^2\|Y\|_{op})^p\left(\int_0^1e^{-x^2}x^{p-1}\mathrm{d}x
    + \int_1^\infty e^{-x}x^{p-1} \mathrm{d}x\right) \\
    &\leq 2p(B^2\|Y\|_{op})^p\left(\Gamma(p/2) + \Gamma(p)\right) \\
    &\leq 4p(B^2\|Y\|_{op})^p\Gamma(p)\leq 4p(pB^2\|Y\|_{op})^p
  \end{split}\end{equation}
where $\Gamma(\cdot)$ is the Gamma function. So
\beq
(\EE[ |Z|^p])^\frac{1}{p} \leq (4p)^\frac{1}{p}pB^2\|\Ym\|_{op} \leq cpB^2\|\Ym\|_{op}
\label{eq:subexp:momentbound}
\eeq
where $c = 4e^{1/e}$. \eqref{eq:subexp:momentbound} is equivalent to $\|Z\|_{\psi_1} \leq cB^2\|\Ym\|_{op}$
by elementary properties of sub-exponential
random variables, e.g.  \cite[Proposition 2.7.1]{vershynin_2018}.

The convex concentration property was introduced in Definition \ref{def:convex_concentration}.
In our case the pilots
$\av_k\in\CC^{\dimPilots}$ are distributed uniformly on the complex $\dimPilots$-dimensional
sphere of radius $\dimPilots$, therefore the real versions
$\av_k^R\in\RR^{2\dimPilots}$ are distributed uniformly on the sphere of radius $2\dimPilots$.
A classical result states that a spherical random variable
$\Xm \sim \text{Unif}(\sqrt{n}S^{n-1})$
has the even stronger (non-convex) concentration property
(e.g. \cite[Theorem 5.1.4]{vershynin_2018}):
\begin{theorem}[Concentration on the Sphere]
    \label{thm:sphere}
    Let $\Xm\sim\textup{Unif}(\sqrt{n}S^{n-1})$ be uniformly distributed
    on the Euclidean sphere of radius $\sqrt{n}$. Then there is an universal
    constant $c>0$, such that for every K-Lipschitz
    function $f:\sqrt{n}S^{n-1}\to\RR$
    \beq
\PP(|f(\Xm) - \EE [f(\Xm)]|>t ) \leq 2\exp(-ct^2/K)
    \eeq
    \hfill $\square$
\end{theorem}
So in particular the columns $\av_k^R$ have the convex concentration property
with some constant $c>0$, independent of the dimension and
it follows by \eqref{eq:real_marginal} and Theorem $\ref{thm:hanson-wright-convex}$ applied to
$\Xm = \av_k^R$ and $\Ym = \Qm_\uv$ 
that the marginals of $\langle \AmatrixR_{:,k},\uv\rangle$ uniformly
satisfy the tail bound of the Hanson-Wright inequality.
As shown in \eqref{eq:subexp:momentbound},
this implies that the columns of $\AmatrixR$ are sub-exponential
with 
\beq
\|\AmatrixR_{:,k}\|_{\psi_1}
= \max_{\uv\in S^{2\dimPilots(\dimPilots-1) - 1}}\|\langle\AmatrixR_{:,k},\uv\rangle\|_{\psi_1} \leq C
\eeq
for some universal constant $C>0$.
It remains to show the tail bound property \eqref{eq:tail_bound}.
From the calculation in \eqref{eq:norm_calc} we see that the column norms of $\AmatrixR$
differ only in the term $f(\av_k) := \sum_{i=1}^L|a_{k,i}|^4$, which is a $2\sqrt{L}$-Lipschitz
function of the random vector $\av_k$ which satisfies the concentration property 
in Theorem \ref{thm:sphere}. Substituting $t=tm$ and $K=2\sqrt{L}$ in Theorem \ref{thm:sphere} gives the desired
tail bound.
With this we can apply Theorem
\ref{thm:rip} and together with Theorems \ref{thm:hanson-wright-convex} and \ref{thm:sphere} it follows
that $\AmatrixR$, as defined in \eqref{eq:a_real}, has RIP
of order $2s$ with RIP constant $\delta_{2s}(\AmatrixR/\sqrt{m})< \delta$
as long as
\begin{equation}
    2s \leq C_\delta\frac{m}{\log^2(eK_\text{tot}/m)}.
\end{equation}
Then, for the complex valued $\mathring{\bA}$ it holds that
\beq
\left\Vert\frac{\mathring{\bA}\xv}{\sqrt{L(L-1)}}\right\Vert_2
= \left\Vert\frac{\sqrt{2}[\Re(\mathring{\bA});\Im(\mathring{\bA})]\xv}{\sqrt{2L(L-1)}}\right\Vert_2
= \left\Vert\frac{\bA^R\xv}{\sqrt{m}}\right\Vert_2
\eeq
for any $\xv\in\RR^{K_\text{tot}}$ and therefore the RIP of $\AmatrixR/\sqrt{m}$ implies the RIP of
$\mathring{\bA}/\sqrt{L(L-1)}$ with the same constants, which concludes the proof of Theorem \ref{thm:rip_main}.
Note that the constraint $\sqrt{m}>c^\prime \log 4N$ is naturally fulfilled when $m$ is
large enough since an exponential scaling of $N$ would makes the achievable sparsity
$s$ go to zero.

\section{Proof of the Recovery Guarantee for NNLS, Theorem \ref{NNLS-theorem}}
\label{appendix:nnlsproof}

Throughout this section let $\gammam^*$ denote the NNLS estimate
\beq
\gammam^* = \argmin_{\gammam \in \RR_+^\dimParam}\|\bA\gammam - \wv \|_2^2
\label{eq:nnls_estimator}
\eeq
as introduced in Section \ref{sec:nnls}, where
$\bA$ is the $\dimPilots^2 \times K_\text{tot}$ matrix whose $k$-th column is given
by $\vec(\bfa_k\bfa_k^\herm)$ and
\begin{equation}
    \wv=\text{vec}(\SigmaEmp - \sigma^2\Id_{\dimPilots}),
    \label{eq:s_def}
\end{equation}
where $\SigmaEmp$ is assumed to be the empirical covariance matrix \eqref{eq:samp_cov} of $M$ iid
samples from a Gaussian distribution $\mathcal{CN}(0,\Sigmam_\yv)$ with covariance matrix
\beq
\Sigmam_\yv = \sum_{i=1}^{K_\text{tot}}\gamma_k^\circ\av_k\av_k^\herm + \sigma^2\Id_{\dimPilots}
\eeq
where  $\bxtrue=(\xtrue_1,\dots,\xtrue_\dimParam)\in \bR_+^{K_\text{tot}}$
is the true (unknown) activity pattern. So $\wv$ can be expressed as
\beq
    \wv = \bA\gammam^\circ + \dv
\eeq
for $\dv:=\text{vec}(\Sigmam_\yv - \SigmaEmp)$.
Let us introduce some notation.
\begin{definition}[Robust NSP (4.21 in \cite{Foucart2013})]
    $\bA\in\CC^{\dimPilots^2\times K_\text{tot}}$ is said to satisfy the
    robust $\ell_q$ nullspace property (NSP) of
    order $s$ with parameters $0<\rho<1$ and $\tau >0$ if
    \begin{equation}
        \lVert \vc{v}_S \rVert_{\ellp{q}} \leq \frac{\rho}{s^{1-1/q}} \lVert
        \vc{v}_{\bar{S}} \rVert_{\ellp{1}} + \tau \left\lVert \Amatrix\vc{v} \right\rVert_{\ellp{2}}
        \quad \forall \vc{v}\in\mathbb{R}^\dimParam
        \label{eq:short:def:nsp}
    \end{equation}
    holds for all subsets $S\subset[\dimParam]$ with $|S|\leq s$.
    The set $\bar{S}$ denotes here the complement of $S$ in $[\dimParam]$.
\end{definition}
Furthermore let the $\ell_1$-error of
the best $s$-sparse approximation to $\bxtrue$ be denoted as:
\begin{equation}
  \sigma_s(\bxtrue)_{\ellp{1}}=\min_{\lVert \xsome\rVert_0\leq s}\lVert \bxtrue-\xsome\rVert_{\ellp{1}}
\end{equation}
If $\bxtrue$ is assumed to actually be $s$-sparse, then we obviously have
$\sigma_s(\bxtrue)_1 = 0$.
The statement of Theorem \ref{NNLS-theorem} will be an immediate consequence
of the following theorem:
\begin{theorem}
    \label{thm:nsp:bound1}
    If $\bA\in\CC^{\dimPilots^2\times K_\text{tot}}$
    has the robust $\ell_2$ NSP of order $s$ with constants
    $\tau>0$ and $\rho \in (0,1)$ and there exists a $\tv \in \CC^{K_\text{tot}}$,
    such that $\mathbf{1} = \bA^\herm\tv$, where $\mathbf{1} := (1,...,1)^\top$,
    then for $p\in[1,2]$ the NNLS estimate $\gammam^*$ in \eqref{eq:nnls_estimator} satisfies
    \beq
        \|\gammam^* - \gammam^\circ\|_p
    \leq \frac{2C \sigma_s(\gammam^\circ)_{\ellp{1}}}{s^{1-1/p}}
          + \frac{2D}{s^{\frac{1}{2} -\frac{1}{p}}}
          \left(\tau+\frac{\|\tv\|_2}{s^{\frac{1}{2}}}\right)\|\vc{d}\|_{\ellp{2}}
    \eeq
    with
    $C \coloneqq \frac{( 1 + \rho )^2}{1 - \rho}$,
    $D = \frac{(3 + \rho)}{1 - \rho}$ and
    $
    \dv = \text{vec}(\Sigmam_\yv - \SigmaEmp)
    $
    \hfill $\square$
\end{theorem}
\begin{proof}This proof is adapted from \cite{kueng2016robust} to our setting.
  First, we will need some implications which follow from the
  NSP \cite[Theorem 4.25]{Foucart2013}.
  Assume that $\bA$ satisfies the robust NSP as stated in the theorem.
  Then, for any $p \in [1,2]$ and for
  all $\vc{x},\vc{z} \in \mathbb{R}^{K_\text{tot}}$,
  \begin{align}
    \|\vc{x} - \vc{z}\|_{\ellp{p}} 
    &\leq \frac{C}{s^{1-1/p}} (\|\vc{x}\|_{\ellp{1}} - \|\vc{z}\|_{\ellp{1}} + 2 \sigma_s(\vc{x})_{\ellp{1}}) \nonumber\\
    &+ D\tau s^{1/p - 1/2} \|\Amatrix (\vc{x}-\vc{z})\|_{\ellp{2}}
    \label{eq:rahut_4.25}
  \end{align}
  holds, with $C,D$ as defined in the statement of the theorem.
  If $\vc{x},\vc{z}\geq 0$
  are non-negative and there exists $\vc{t}$ such that
  $\vc{1}=\Amatrix^\herm\vc{t}$ we  use:
  \begin{equation}
    \begin{split}
      \|\vc{x}\|_{\ellp{1}} - \|\vc{z}\|_{\ellp{1}}
      &=\langle\vc{1},\vc{x}-\vc{z}\rangle
      =\langle\Amatrix^\herm\vc{t},\vc{x}-\vc{z}\rangle\\
      &=\langle\vc{t},\Amatrix(\vc{x}-\vc{z})\rangle
      \leq\|\vc{t}\|_{\ellp{2}}\|\Amatrix(\vc{x}-\vc{z})\|_{\ellp{2}}
    \end{split}
  \end{equation}
  where we have used Cauchy-Schwarz inequality (note that $\langle\vc{t},\Amatrix(\vc{x}-\vc{z})\rangle$ is real).
  So inequality \eqref{eq:rahut_4.25} implies:
  \begin{align}
    \|&\vc{x} - \vc{z}\|_{\ellp{p}} \\
      &\leq \frac{2C \sigma_s(\vc{x})_{\ellp{1}}}{s^{1-1/p}}
        + \left(D\tau+\frac{C\cdot\|\vc{t}\|_{\ellp{2}}}{s^{1/2}}\right)
        \frac{\|\Amatrix (\vc{x}-\vc{z})\|_{\ellp{2}}}{s^{\frac{1}{2} -\frac{1}{p}}}
  \end{align}
  Now, lets take   $\vc{y}=\Amatrix\vc{x}+\vc{d}$. Since
  $\|\Amatrix (\vc{x}-\vc{z})\|_{\ellp{2}}
  \leq\|\Amatrix\vc{z}-\vc{y}\|_{\ellp{2}}+\|\vc{d}\|_{\ellp{2}}$
  we get for
  all non-negative $\vc{z}$ and $\vc{x}$:
  \begin{align}
    \|&\vc{x} - \vc{z}\|_{\ellp{p}} \\
      &= \frac{2C \sigma_s(\vc{x})_{\ellp{1}}}{s^{1-1/p}}
        + \left(D\tau+\frac{C\cdot\|\vc{t}\|_{\ellp{2}}}{s^{1/2}}\right)
        \frac{\|\Amatrix\vc{z}-\vc{y}\|_{\ellp{2}}+\|\vc{d}\|_{\ellp{2}}}{s^{\frac{1}{2} -\frac{1}{p}}}
  \end{align}
  Now take $\zv=\gammam^*$ and $\xv = \gammam^\circ$, then $\yv=\wv$ (see \eqref{eq:s_def}).
  Since $\gammam^\circ \in \RR^{K_\text{tot}}_+$ is itself is a feasible point of the
  minimization we have $\min_{\gammam\in\RR^{K_\text{tot}}_+}\|\bA\gammam-\bv\|_2 \leq \|\dv\|_2$,
  yielding:
  \begin{align}
      \|&\gammam^* - \gammam^\circ\|_{\ellp{p}}
        \leq \frac{2C \sigma_s(\gammam^\circ)_{\ellp{1}}}{s^{1-1/p}}
        + 2\left(D\tau+\frac{C\cdot\|\vc{t}\|_{\ellp{2}}}{s^{1/2}}\right)
        \frac{\|\vc{d}\|_{\ellp{2}}}{s^{\frac{1}{2} -\frac{1}{p}}}
  \end{align}    
  It is easily checked that $C\leq D$ for $\rho \in (0,1)$, which gives the result.
\end{proof}

In our case we choose
$\vc{t}=t\cdot\text{vec}(\Id_{\dimPilots})\in\mathbb{R}^{\dimPilots^2}$
with some $t>0$. Let $\bA_k$ be the $k$-th column of $\bA$.
It holds that
\beq
\bA_k^\herm\text{vec}(\Id_{\dimPilots}) = \tr(\av_k\av_k^\herm) = \|\av_k\|_2^2.
\eeq
Using the normalization of the pilots
$\lVert \avec_k\rVert_{\ellp{2}}^2=\dimPilots$, we get that:
\begin{equation*}
    \Amatrix^\herm\vc{t}=t\dimPilots\cdot\mathbf{1}
\end{equation*}
so $t=1/\dimPilots$, and therefore
$\lVert \vc{t}\rVert^2_{\ellp{2}}=1/\dimPilots$ gives the desired condition
$\bA^\herm\tv = \mathbf{1}$.
Before we can make use of Theorem \ref{thm:nsp:bound1} it remains to show that $\bA$ has
the robust $\ell_2$-NSP with high probability. 
To this end, we will
restrict to those measurements which are related to the isotropic part
of $\bA$.
Let $\mathring{\bA}$ be the centered version of $\bA$ defined in \eqref{eq:a_centered}.
Now it is
easy to check (revert the vectorization) that this special structure gives us the inequality:
\begin{equation}
  \lVert\Amatrix\vc{v}\rVert^2_{\ellp{2}}
  =\lVert\mathring{\bA}\vc{v}\rVert^2_{\ellp{2}}+\|\bA^\diag\vv\|_2^2
  \geq\lVert\mathring{\bA}\vc{v}\rVert^2_{\ellp{2}}
  \label{eq:thm:short:Av:lowerbound}
\end{equation}
where $\bA^\diag\in\CC^{L\times K_\text{tot}}$ is defined as the non-isotropic part of $\bA$ with its $k$-th column
defined by $\bA^\diag_{:,k}=\text{vec}(\diag(\av_k\av_k^\herm))$.
This shows that if $\mathring{\bA}$ has the $\ell_2$-NSP of order $s$ with constants
$\tau$ and $\rho$, then so does $\bA$, since
\beq
\begin{split}
\lVert \vc{v}_S \rVert_{\ellp{2}}
&= \frac{\rho}{\sqrt{s}} \lVert
      \vc{v}_{\bar{S}} \rVert_{\ellp{1}} + \tau\left\lVert \mathring{\bA}\vc{v} \right\rVert_{\ellp{2}}\\
      &\leq
      \frac{\rho}{\sqrt{s}} \lVert
      \vc{v}_{\bar{S}} \rVert_{\ellp{1}} + \tau\left\lVert \bA\vc{v} \right\rVert_{\ellp{2}}
\end{split}
\label{eq:nsp_centered}
\eeq
holds for all subsets $S \subset [K_\text{tot}]$ with $|S|\leq s$.
It is well-known that the robust $\ell_2$-NSP of order $s$ is implied by the RIP of order $2s$
with sufficiently small constants \cite{Foucart2013}.
The following theorem specifies how RIP is related to the $\ell_2$-NSP
\begin{theorem}
    \label{thm:rip_nsp}
If $\mathring{\bA}$ has RIP of order $2s$ with a constant bound as
$\delta_{2s}(\mathring{\bA}/\sqrt{L(L-1)})\leq \delta<4/\sqrt{41}\approx0.62$ then
$\mathring{\bA}/\sqrt{L(L-1)}$ has the robust $\ell_2$-NSP of
order $s$ with parameters $\rho$ and $\tau'$
with $\rho\leq\delta/(\sqrt{1-\delta^2}-\delta/4)$ and
$\tau'\leq\sqrt{1+\delta}/(\sqrt{1-\delta^2}-\delta/4)$.\\
Furthermore 
$\mathring{\bA}$ has the robust $\ell_2$-NSP of
order $s$ with parameters $\rho$ and $\tau=\tau'/\sqrt{L(L-1)}$.
    \hfill $\square$
\end{theorem}
\begin{proof}
    The first part is shown in \cite[Theorem 6.13]{Foucart2013}. The last statements
    follows immediately from
    \begin{equation}
        \begin{split}
            \tau \dimPilots\left\lVert \frac{1}{\dimPilots}\mathring{\bA}\vc{v} \right\rVert_{\ellp{2}}
            = \tau\left\lVert \mathring{\bA}\vc{v} \right\rVert_{\ellp{2}}
        \end{split}
    \end{equation} 
\end{proof}

Theorem \ref{thm:rip_main} establishes the RIP of $\mathring{\bA}/\sqrt{\dimPilots(\dimPilots-1)}$
under the assumptions of \ref{NNLS-theorem}.
If we fix $\delta <  4/\sqrt{41}$, Theorem \ref{thm:rip_nsp}
implies the robust $\ell_2$-NSP of order $s$ for $\mathring{\bA}$ with explicit bounds on $\tau$ and $\rho$.
For example, $\delta=0.5$ gives
$\rho < 0.68$ and $\tau< 3/\dimPilots$.
As shown in \eqref{eq:nsp_centered} the robust $\ell_2$-NSP of $\mathring{\bA}$ implies the $\ell_2$-NSP of the uncentered version
$\bA$ of the same order with the same constants.
Finally, the application of Theorem \ref{thm:nsp:bound1} concludes the proof of Theorem \ref{NNLS-theorem}.
\section{Analysis of Error of the Sample Covariance Matrix}
\label{cov_err_app}
Let $\Sigmam_\yv\in\RR^{L\times L}$ be fixed and let $\{\bfy(t): t \in [M]\}$
be $M$ i.i.d. samples from $\mathcal{CN}(0,\Sigmam_\yv)$.
We first consider the simple case where $\Sigmam_\yv$ is diagonal, given by
$\Sigmam_\bfy=\diag(\betam)$ and let $\Deltam=\widehat{\Sigmam}_\bfy -  \Sigmam_\bfy$
be the deviation of the sample covariance matrix from its mean. Then $\|\dv\|_2 = \|\Deltam\|_F$
and the $(i,j)$-th component of $\Deltam$ is given by
\begin{align}
    \Delta_{ij}
    &=\frac{1}{M} \sum_{t \in [M]} y_i(t)y_j^*(t) - \beta_i \delta_{ij}\\
    &= \frac{\sqrt{\beta_i\beta_j}}{M}\sum_{t \in [M]}
        \left(\frac{y_i(t)}{\sqrt{\beta_i}}\frac{y_j^*(t)}{\sqrt{\beta_j}} - \delta_{ij}\right)
\end{align}
where $\delta_{ij}=\mathbb{1}_{\{i=j\}}$ denotes the discrete delta function.
Let $Y_{ij}(t):=\frac{y_i(t)}{\sqrt{\beta_i}}\frac{y_j^*(t)}{\sqrt{\beta_j}} - \delta_{ij}$.
Then 
\beq
|\Delta_{ij}|^2 = \frac{\beta_i\beta_j}{M^2}\left\vert\sum_{t=1}^M Y_{ij}(t)\right\vert^2
\label{eq:abs_delta_ij}
\eeq
Since all $Y_{ij}(t)$ are zero mean and are independent for fixed $i,j$. Therefore the
variance of their sum $\EE\left[\left\vert\sum_{t=1}^M Y_{ij}(t)\right\vert^2\right]$ is the sum
of their variances. In the following we show that $\EE[|Y_{ij}|^2] = 1$ for all $i,j$.
For $i \not = j$, we have that  
\beq
\begin{split}
    \bE[|Y_{ij}|^2]
    &=\frac{\bE[\left\vert y_i(t) y_j(t)^*\right\vert^2]}{\beta_i\beta_j}\\
    &\stackrel{(a)}{=}\frac{\bE[|y_i(t)|^2]}{\beta_i} \frac{\bE[|y_j(t)|^2]}{\beta_j}\\
    &=1,
\end{split}
    \label{cov_err_dumm_1}
\eeq
where in $(a)$ we used the independence of the different components of $\bfy(t)$.
Also, for $i=j$, we have that
\beq
\begin{split}
    \bE[|Y_{ij}|^2]
    &=\bE\left[\left\vert \frac{|y_i(t)|^2}{\beta_i}-1  \right\vert^2\right]\\
    &=\frac{\bE[|y_i(t)|^4]}{\beta_i^2} -2 \frac{\bE[|y_i(t)|^2]}{\beta_i}+ 1\\
&\stackrel{(a)}{=}2 -2+1\\
&=1,
\end{split}
    \label{cov_err_dumm_2}
\eeq
where in $(a)$ we used the identity $\bE[|y_i(t)|^4]=2\bE[|y_i(t)|^2]^2$ for complex Gaussian random variables.
Overall, from \eqref{cov_err_dumm_1} and \eqref{cov_err_dumm_2}, we can write $\bE[|\Delta_{ij}|^2]=\frac{\beta_i \beta_j}{M}$. 
Thus, we have that
\begin{align}
\bE[\|\Deltam\|_{\sfF}^2]=\sum_{ij} \bE[|\Delta_{ij}|^2] &= \frac{\sum_{i,j} \beta_i \beta_j}{M}\nonumber\\
&=  \frac{(\sum \beta_i)^2}{M} = \frac{\tr(\Sigmam_\bfy)^2}{M}.\label{exp_del}
\end{align}

To see how fast $\|\Deltam\|_{\sfF}$ concentrates around its mean, note that for fixed $i,j$
the $Y_{ij}(t)$ are
independent sub-exponential random variables with sub-exponential norm $\leq 1$
(see e.g. \cite[Lemma 2.7.7]{vershynin_2018}). Therefore, by the elemental Bernstein inequality we can
estimate that for any $\alpha >0$
\beq
\begin{split}
    \PP\left(\left\vert\sum_{t=1}^MY_{ij}(t)\right\vert^2>\alpha\right) 
    &= \PP\left(\left\vert\sum_{t=1}^MY_{ij}(t)\right\vert>\sqrt{\alpha}\right) \\
    &\leq 2\exp(-c\min\{\alpha/M,\sqrt{\alpha}\})
\end{split}
\eeq
for some universal constant $c>0$. By a union bound we can see that
\beq
\begin{split}
    &\PP\left(\min_{i,j}\left\vert\sum_{t=1}^MY_{ij}(t)\right\vert^2>\alpha\right)\\
    &\leq \binom{\dimPilots}{2}\PP\left(\left\vert\sum_{t=1}^MY_{ij}(t)\right\vert^2>\alpha\right) \\
    &\leq 2\exp\left(2\log (e\dimPilots) -c\min\{\alpha/M,\sqrt{\alpha}\}\right)
    \label{eq:frobenius_union}
\end{split}
\eeq
By choosing $\alpha$ properly we can get the following statement:
\begin{theorem}
    \label{thm:fro_bound}
    Let $\epsilon >0$
    \beq
    \|\Deltam\|_\sfF \leq \frac{\tr(\Sigmam_\yv)}{\sqrt{M}}\sqrt{\frac{\log\left(\frac{2(e\dimPilots)^2}{\epsilon}\right)}{c}}
    \eeq
    holds with probability exceeding $1-\epsilon$, if $cM > \log (2(e\dimPilots)^2/\epsilon)$,
    where $c>0$ is the constant in \eqref{eq:frobenius_union}.
    \hfill $\square$
\end{theorem}
\begin{proof}
    In \eqref{eq:frobenius_union} choose $\alpha = M\delta$ with $\delta = \log(2(e\dimPilots)^2/\epsilon)/c$.
    Then $\min\{\alpha/M,\sqrt{\alpha}\} = \min\{\delta,\sqrt{\delta M}\}$. Under the condition on $M$
    stated in the Theorem, $\min\{\delta,\sqrt{\delta M}\} = \delta$. So 
    \beq
    \begin{split}
        &\PP\left(\min_{i,j}\left\vert\sum_{t=1}^MY_{ij}(t)\right\vert^2>\delta M\right)\\
    &\leq 2\exp\left(2\log (e\dimPilots) -2\log(e\dimPilots) + \log(\epsilon/2)\}\right)\\
        &= \epsilon.
        \label{eq:union_delta}
    \end{split}
    \eeq
    and the statement of the Theorem follows from
    \beq
    \begin{split}
        &\PP\left(\|\Deltam\|_\sfF > \frac{\tr(\Sigmam_\yv)}{\sqrt{M/\delta}}\right)\\
        &=\PP\left(\|\Deltam\|_\sfF^2 > \frac{\tr(\Sigmam_\yv)^2}{M/\delta}\right)\\
        &=\PP\left(\sum_{ij}\frac{\beta_i\beta_j}{M^2}\left\vert\sum_{t=1}^MY_{ij}(t)\right\vert^2 > \frac{\tr(\Sigmam_\yv)^2}{M/\delta}\right)\\
        &\leq\PP\left(\min_{ij}\left\vert\sum_{t=1}^MY_{ij}(t)\right\vert^2>\delta M\right)\\
        &\leq \epsilon
    \end{split}
    \eeq
    where in the second equality we used \eqref{eq:abs_delta_ij} and in the last inequality we used
    \eqref{eq:union_delta}.
\end{proof}

Now, assume that the covariance matrix $\Sigmam_\bfy$ is not in a diagonal form and let
$\Sigmam_\bfy=\bfU \diag(\betam) \bfU^\herm$ be the singular value decomposition of
$\Sigmam_{\bfy}$.
By multiplying all the vectors $\bfy(t)$ by the orthogonal matrix $\bfU^\herm$ to
whiten them and noting the fact that multiplying by $\bfU^\herm$ does not change the
Frobenius norm of a matrix, we can see that the bound in Theorem \ref{thm:fro_bound},
which depends on $\Sigmam_\yv$ only through its trace, holds true in general
also for non-diagonal covariance matrices.
Finally, since in Theorem \ref{NNLS-theorem}
$\Sigmam_\yv = \sum_{k=1}^{K_\text{tot}}\gamma_k\av_k\av_k^\herm + \sigma^2\Id_L$
and the pilot sequences satisfy $\|\av_k\|_2^2 = L$,
it holds that 
\beq
\tr\left(\Sigmam_\yv\right)= \sum_{k=1}^{K_\text{tot}}\gamma_k\tr(\av_k\av_k^\herm) + \sigma^2\tr(\Id_L)
=  L\left(\|\gammam\|_1 + \sigma^2\right),
\eeq
which gives 
\eqref{eq:d_expectation} and \eqref{eq:d_bound}.
\begin{remark}
It is worthwhile to mention that although  \eqref{exp_del} was derived under the Gaussianity of the observations $\{\bfy(t): t\in [M]\}$, the result can be easily modified for general distribution of the components of $\bfy(t)$. More specifically, let us define  
\begin{align}\label{eta_bound}
\max_{i} \frac{\bE[|y_i(t)|^4]}{\bE[|y_i(t)|^2]^2}=:\varsigma < \infty.
\end{align}
Then, using \eqref{cov_err_dumm_1} and applying \eqref{eta_bound} to \eqref{cov_err_dumm_2}, we can obtain the following upper bound
\begin{align}
\bE[\|\Deltam\|_{\sfF}^2] &\leq \max\{\varsigma-1,1\} \times \frac{\sum_{i,j} \beta_i \beta_j }{M}\\
& \leq \max\{\varsigma-1,1\} \times \frac{\tr(\Sigmam_\bfy)^2}{M},
\end{align}
which is equivalent to \eqref{exp_del} up to the constant multiplicative factor $\max\{\varsigma-1,1\}$. \hfill $\lozenge$
\end{remark}

\section*{Acknowledgement} 
The authors would like to thank R. Kueng for inspiring discussions and helpful comments. P.J. is supported by DFG grant JU 2795/3.
AF is supported by the Alexander-von-Humbold foundation.

\balance 

{\small
\bibliographystyle{IEEEtran}
\bibliography{references,ref_alex}
}

\end{document}